\documentclass[nonacm,screen]{acmart}

\authorsaddresses{}

\usepackage[capitalise,nameinlink,noabbrev]{cleveref}

\usepackage{tcolorbox}
\tcbuselibrary{breakable}
\tcbuselibrary{skins}

\usepackage{tikz}
\usetikzlibrary{arrows.meta}

\usepackage{subcaption}
\allowdisplaybreaks

\DeclareRobustCommand{\abbrevcrefs}{%
	\crefname{proposition}{Prop.}{Props.}%
}
\DeclareRobustCommand{\cshref}[1]{{\abbrevcrefs\cref{#1}}}

\DeclareMathOperator*{\argmin}{argmin}
\DeclareMathOperator*{\argmax}{argmax}
\newcommand{\sz}{\textup{size}}

\def\tfnp/{\textup{TFNP}}
\def\ppa/{\textup{PPA}}
\def\ppad/{\textup{PPAD}}
\def\ppp/{\textup{PPP}}
\def\pls/{\textup{PLS}}
\def\cls/{\textup{CLS}}
\def\ppadpls/{\text{$\textup{PPAD} \cap \textup{PLS}$}}
\def\fixp/{\textup{FIXP}}
\def\linearfixp/{\textup{Linear-FIXP}}
\def\linearcls/{\textup{Linear-CLS}}

\def\clo/{\textup{\textsc{Continuous-Localopt}}}
\def\gclo/{\textup{\textsc{General-Continuous-Localopt}}}
\def\kkt/{\textup{\textsc{KKT}}}
\def\gdls/{\textup{\textsc{GD-Local-Search}}}
\def\gdfp/{\textup{\textsc{GD-Fixpoint}}}
\def\gdfd/{\textup{\textsc{GD-Finite-Diff}}}
\def\gbrouwer/{\textup{\textsc{General-Brouwer}}}
\def\grlo/{\textup{\textsc{General-Real-Localopt}}}
\def\linclo/{\textup{\textsc{Linear-Continuous-Localopt}}}
\def\eol/{\textup{\textsc{End-of-Line}}}
\def\iter/{\textup{\textsc{Iter}}}

\theoremstyle{definition}
\newtheorem{definition}{Definition}
\newtheorem{remark}{Remark}

\theoremstyle{plain}
\newtheorem{theorem}{Theorem}[section]
\newtheorem{lemma}[theorem]{Lemma}
\newtheorem{proposition}[theorem]{Proposition}
\newtheorem{corollary}[theorem]{Corollary}

\begin{document}

\title{The Complexity of Gradient Descent: CLS = PPAD \texorpdfstring{$\cap$}{∩} PLS}

\author{John Fearnley}
\orcid{0000-0003-0791-4342}
\affiliation{%
	\institution{University of Liverpool}
	\department{Department of Computer Science}
	\city{Liverpool}
	\country{UK}
}
\email{john.fearnley@liverpool.ac.uk}

\author{Paul W. Goldberg}
\orcid{0000-0002-5436-7890}
\affiliation{%
	\institution{University of Oxford}
	\department{Department of Computer Science}
	\city{Oxford}
	\country{UK}
}
\email{paul.goldberg@cs.ox.ac.uk}

\author{Alexandros Hollender}
\orcid{0000-0001-5255-9349}
\affiliation{%
	\institution{University of Oxford}
	\department{Department of Computer Science}
	\city{Oxford}
	\country{UK}
}
\email{alexandros.hollender@cs.ox.ac.uk}

\author{Rahul Savani}
\orcid{0000-0003-1262-7831}
\affiliation{%
	\institution{University of Liverpool}
	\department{Department of Computer Science}
	\city{Liverpool}
	\country{UK}
}
\email{rahul.savani@liverpool.ac.uk}

\begin{abstract}
We study search problems that can be solved by performing Gradient Descent on a bounded convex polytopal domain and show that this class is equal to the intersection of two well-known classes: PPAD and PLS. As our main underlying technical contribution, we show that computing a Karush-Kuhn-Tucker (KKT) point of a continuously differentiable function over the domain $[0,1]^2$ is PPAD\,$\cap$\,PLS-complete. This is the first non-artificial problem to be shown complete for this class. Our results also imply that the class CLS (Continuous Local Search) -- which was defined by Daskalakis and Papadimitriou as a more ``natural'' counterpart to PPAD\,$\cap$\,PLS and contains many interesting problems -- is itself equal to PPAD\,$\cap$\,PLS.
\end{abstract}

\maketitle

\section{Introduction}

It is hard to overstate the importance of Gradient Descent. As noted by \citet{JinNGKJ21},
``Machine learning algorithms generally arise via formulations as optimization problems, and, despite a massive classical toolbox of sophisticated optimization algorithms and a major modern effort to further develop that toolbox,  the simplest algorithms---gradient descent, which dates to the 1840s \citep{Cauchy1847} and stochastic gradient descent, which dates to the 1950s \citep{RM1951}---reign supreme in machine learning.'' \citet{JinNGKJ21} continue by highlighting the simplicity of Gradient Descent as a key selling-point, and the importance of theoretical analysis in understanding its efficacy in non-convex optimization.

In its simplest form, which we consider in this paper, Gradient Descent attempts to find a minimum of a continuously differentiable function $f$ over some domain $D$, by starting at some point $x_0$ and iterating according to the update rule
$$x_{k+1} \leftarrow x_k - \eta \nabla f(x_k)$$
where $\eta$ is some fixed step size. The algorithm is based on the fundamental fact that for any point $x$ the term $-\nabla f(x)$ points in the direction of steepest descent in some sufficiently small neighbourhood of $x$. However, in the unconstrained setting---where the domain is the whole space---it is easy to see that Gradient Descent can at best find a stationary point. Indeed, if the gradient is zero at some point, then there is no escape. Note that a stationary point might be a local minimum, but it could also be a saddle point or even a local maximum. Similarly, in the constrained setting---where the domain $D$ is no longer the whole space---Gradient Descent can at best find a point $x$ that satisfies the Karush-Kuhn-Tucker (KKT) optimality conditions. Roughly, the KKT conditions say that the gradient of $f$ is zero at $x$, or if not, $x$ is on the boundary of $D$ and any further local improvement would take us outside $D$.

In this paper we investigate the complexity of finding a point where Gradient Descent terminates---or equivalently, as we will see, a KKT point---when the domain is \emph{bounded}. It is known that a global or even a local minimum cannot be found in polynomial time unless P\,=\,NP \citep{MurtyK1987,AhmadiZ22-polytope}. Indeed, even deciding whether a point is a local minimum is already co-NP-hard \citep{MurtyK1987}. In contrast, it is easy to check whether a point satisfies the KKT conditions. In general, finding a KKT point is hard, since even deciding whether a KKT point exists is NP-hard in the unconstrained setting \citep{AhmadiZ22-unconstrained}. However, when the domain is bounded, a KKT point is guaranteed to exist! This means that in our case, we are looking for something that can be verified efficiently and that necessarily exists. Intuitively, it seems that this problem should be more tractable. This intuition can be made formal by noting that these two properties place the problem in the complexity class \tfnp/ of {\em total} search problems in NP: any instance has at least one solution, and a solution can be checked in polynomial time. A key feature of such problems is that they cannot be NP-hard unless NP\,=\,co-NP \citep{MegiddoP1991-tfnp}. \tfnp/ problems have been classified via certain ``syntactic subclasses'' of \tfnp/, of which \ppad/ and \pls/ are two of the most important ones.

\subsection{NP total search classes: \ppad/, \pls/, and \cls/}

As discussed by \citet{Pap94}, \tfnp/ is unlikely to have complete problems, and various \emph{syntactic} subclasses have been used to classify the many diverse problems that belong to it.
Among them, the classes \ppad/ and \pls/ (introduced by \citet{Pap94} and \citet{JPY88} respectively) have been hugely successful in this regard.
Each of these classes has a corresponding {\em computationally inefficient existence proof principle}, one that when applied in a general context, does not yield a polynomial-time algorithm.\footnote{The other well-known such classes, less relevant to the present paper, are \ppa/ and \ppp/; it is known that \ppad/ is a subset of \ppa/ and also of \ppp/. These set-theoretic containments correspond directly to the strength, or generality, of the corresponding proof principles.} In the case of \ppad/ this is the {\em parity argument on a directed graph}, equivalent to the existence guarantee of {\em Brouwer fixpoints}: a Brouwer function is a continuous function $f:D\rightarrow D$ where $D$ is a convex compact domain, and Brouwer's fixed point theorem guarantees a point $x$ for which $f(x)=x$. \ppad/ has been widely used to classify problems of computing game-theoretic equilibria (a long line of work on Nash equilibrium computation beginning with \citet{DGP09,CDT09}, and market equilibria, e.g., \citet{CDDT09}). \ppad/ also captures diverse problems in combinatorics and cooperative game theory \citep{KPRST13}.  

\pls/, for ``Polynomial Local Search'', captures problems of finding a local minimum of an objective function $f$, in contexts where any candidate solution $x$ has a local neighbourhood within which we can readily check for the existence of some other point having a lower value of $f$. Many diverse local optimization problems have been shown complete for \pls/, attesting to its importance. Examples include searching for a local optimum of the TSP according to the Lin-Kernighan heuristic \citep{Papadimitriou92}, and finding pure Nash equilibria in many-player congestion games \citep{FPT04}.

The complexity class \cls/ (``Continuous Local Search'') was introduced by \citet{DaskalakisP2011-CLS} to classify various important problems that lie in both \ppad/ and \pls/. \ppad/ and \pls/ are believed to be strictly incomparable---one is not a subset of the other---a belief supported by oracle separations \citep{Morioka01-Mthesis-PLS,BureshM04-NP-search-problems,buss2012propositional}. It follows from this that problems belonging to both classes cannot be complete for either one of them. \cls/ is seen as a strong candidate for capturing the complexity of some of those important problems, but, prior to this work, only two problems related to general versions of Banach's fixed point theorem were known to be \cls/-complete \citep{DaskTZ18,FGMS17}. An important result---supporting the claim that \cls/-complete problems are hard to solve---is that the hardness of \cls/ can be based on various cryptographic assumptions such as indistinguishability obfuscation \citep{HubacekY2017-CLS}, the soundness of the Fiat-Shamir heuristic applied to the sumcheck protocol \citep{ChoudhuriHKPRR19-Fiat-Shamir}, or the assumption that Learning With Errors (LWE) is sub-exponentially hard \citep{JawaleKKZ21-PPAD-LWE}. Prior to the present paper, it was generally believed that \cls/ is a proper subset of \ppadpls/, as conjectured by \citet{DaskalakisP2011-CLS}.

\subsection{Our contribution and its significance}\label{sec:intro-results}

As our main result, we show that finding a point where Gradient Descent on a continuously differentiable function terminates---or equivalently a KKT point---is \ppadpls/-complete, when the domain is a bounded convex polytope. This continues to hold even when the domain is as simple as the unit square $[0,1]^2$. The \ppadpls/-completeness result applies to the ``white box'' model, where functions are represented as arithmetic circuits.

\paragraph{\bf Computational Hardness}
As an immediate consequence, our result provides convincing evidence that the problem is computationally hard.
First of all, there are reasons to believe that \ppadpls/ is hard simply because \ppad/ and \pls/ are believed to be hard. Indeed, if \ppadpls/ could be solved in polynomial time, then, given an instance of a \ppad/-complete problem and an instance of a \pls/-complete problem, we would be able to solve at least one of the two instances in polynomial time. Furthermore, since \cls/ $\subseteq$ \ppadpls/, the above-mentioned cryptographic hardness of \cls/ applies automatically to \ppadpls/, and thus to our problem of interest. Note that our result says that finding a stationary point (or, to be more precise, a KKT point) is computationally hard, not only for the Gradient Descent algorithm, but for \emph{any} algorithm.

\paragraph{\bf Continuous Local Search} Since Gradient Descent is just a special case of continuous local search, our hardness result implies that
$$\cls/ = \ppadpls/$$
which disproves the widely believed conjecture by \citet{DaskalakisP2011-CLS}
that the containment is strict. Our result also allows us to resolve an
ambiguity in the original definition of \cls/ by showing that the
high-dimensional version of the class reduces to the 2-dimensional version of
the class (the 1-dimensional version is computationally tractable, so no further progress is to be made). Equality to \ppadpls/ also applies to a linear version of \cls/ analogous to the class \linearfixp/ of \citet{EtessamiPRY20}.

\paragraph{\bf $\boldsymbol{\ppadpls/}$} Perhaps more importantly, our result establishes \ppadpls/ as an important complexity class that captures the complexity of interesting problems. It was previously known that one can construct a problem complete for \ppadpls/ by gluing together two problems, one for each class (see \cref{sec:classesdefs}), but the resulting problem is highly artificial. In contrast, the Gradient Descent problem we consider is clearly natural and of separate interest. Some \tfnp/ classes can be characterized as the set of all problems solved by some type of algorithm, where ``solved'' is interpreted as ``solved eventually, without any efficiency requirement''. For instance, \ppad/ is the class of all problems that can be solved by ``path-following'' algorithms such as the Lemke-Howson algorithm. \pls/ is the class of all problems that can be solved by general local search methods. Analogously, one can define the class GD containing all problems that can be solved by the Gradient Descent algorithm on a bounded domain, i.e., that reduce to our Gradient Descent problem in polynomial time. Our result shows that GD\,=\,\ppadpls/. In other words, the class \ppadpls/, which is obtained by combining \ppad/ and \pls/ in a completely artificial way, turns out to have a very natural characterization:
\begin{center}
	\textit{\ppadpls/ is the class of all problems that can be solved\\
		by performing Gradient Descent on a bounded domain.}
\end{center}
Our new characterization has already been very useful in the context of
Algorithmic Game Theory, where it was recently used by \citet{BabR21}, to show
\ppadpls/-completeness of computing mixed Nash equilibria of congestion games.

\subsection{Further related work}

Following the definition of \cls/ by \citet{DaskalakisP2011-CLS}, two \cls/-complete problems were identified: \textsc{Banach} \citep{DaskTZ18} and \textsc{MetametricContraction} \citep{FGMS17}. \textsc{Banach} is a computational presentation of Banach's fixed point theorem in which the metric is presented as part of the input (and could be complicated). Banach fixpoints are unique, but \cls/ problems do not in general have unique solutions, and the problem \textsc{Banach} circumvents that obstacle by allowing certain ``violation'' solutions, such as a pair of points witnessing that $f$ is not a contraction map. \textsc{MetametricContraction} is a generalisation of \textsc{Banach}, where the metric is replaced by a slightly relaxed notion called a meta-metric.

\citet{ChaRW19} showed that online gradient descent can encode general PSPACE computations. In contrast, our result provides evidence that the problem itself (which gradient descent attempts to solve) is hard. (Although, our result does not automatically apply to the specific machine learning setting of \citep{ChaRW19}.) The distinction between these two types of statements is most clearly apparent in the case of linear programming, where the simplex method can encode arbitrary PSPACE computations \citep{DisserS19-simplex-NP-mighty,AdlerPR14-simplex,FeaS15}, while the problem itself can be solved in polynomial time.

\citet{DSZ21} study nonlinear {\em min-max} optimization, namely the optimization problem: $\min_x \max_y f(x,y)$ where $(x,y) \in \mathbb{R}^{d_1} \times \mathbb{R}^{d_2}$ is constrained to lie in a polytope ${\mathcal P}$ and $f$ is a smooth function. They define an appropriate notion of approximate local solution for this setting and prove that computing such a solution is \ppad/-complete, for a suitable regime of parameters. A contrasting aspect here is that our hardness result requires inverse-exponential parameters, whereas \citet{DSZ21} achieve hardness with inverse-polynomial parameters---for us the inverse-exponential parameters are a necessary evil, since the problem can otherwise be solved in polynomial time, even in high dimension (by running Gradient Descent, see \cref{lem:kkt-poly-params}). Furthermore, the \ppad/-hardness of \citet{DSZ21} requires ${\mathcal P}$ to be a carefully chosen subset of the unit hypercube, whereas for us, the feasible region is the unit square.

Finally, note that in contrast to our hardness result, in the special case of \emph{convex} optimization our problem can be solved efficiently, even in high dimension and with inverse-exponential precision. Related work in nonlinear optimization is covered in \cref{sec:nonlinearopt}.

\section{Overview}\label{sec:overview}

In this section we give a condensed and informal overview of the concepts, ideas, and techniques of this paper. We begin by providing informal definitions of the problems of interest and the complexity classes. We then present an overview of our results, along with the high-level ideas of our main reduction.

\subsection{The problems of interest}\label{sec:defs}

The motivation for the problems we study stems from the ultimate goal of minimizing a continuously differentiable function $f: \mathbb{R}^n \to \mathbb{R}$ over some domain $D$. As mentioned in the introduction, this problem is known to be intractable, and so we instead consider relaxations where we are looking for a point where Gradient Descent terminates, or for a KKT point. Our investigation is restricted to bounded domains, namely we consider the setting where the domain $D$ is a bounded convex polytope defined by a collection of linear inequalities. Furthermore, we also assume that the function $f$ and its gradient $\nabla f$ are Lipschitz-continuous over $D$, for some Lipschitz constant $L$ provided in the input. Let $C^1_L(D,\mathbb{R})$ denote the set of continuously differentiable functions $f$ from $D$ to $\mathbb{R}$, such that $f$ and $\nabla f$ are $L$-Lipschitz.

In order to define our Gradient Descent problem, we need to specify what we mean by ``a point where Gradient Descent terminates''. We consider the following two stopping criteria for Gradient Descent: (a) stop when we find a point such that the next iterate does not improve the objective function value, or (b) stop when we find a point such that the next iterate is the same point. In practice, of course, Gradient Descent is performed with some underlying precision parameter $\varepsilon > 0$. Thus, the appropriate stopping criteria are: (a) stop when we find a point such that the next iterate improves the objective function value by less than $\varepsilon$, or (b) stop when we find a point such that the next iterate is at most $\varepsilon$ away. Importantly, note that, given a point, both criteria can be checked efficiently. This ensures that the resulting computational problems lie in \tfnp/. The totality of the problems follows from the simple fact that a local minimum must exist (since the domain is bounded) and any local minimum satisfies the stopping criteria. The first stopping criterion has a local search flavour and so we call the corresponding problem \gdls/. The second stopping criterion is essentially asking for an approximate fixed point of the Gradient Descent dynamics, and yields the \gdfp/ problem.

Since we are performing Gradient Descent on a bounded domain, we have to ensure that the next iterate indeed lies in the domain $D$. The standard way to achieve this is to use so-called Projected Gradient Descent, which computes the next iterate as usual and then projects it onto the domain. Define $\Pi_D$ to be the projection operator, that maps any point in $D$ to itself, and any point outside $D$ to its closest point in $D$ (under the Euclidean norm). The two Gradient Descent problems are defined as follows.

\begin{tcolorbox}
	{\centering
		\gdls/ and \gdfp/ \textit{(informal)}
		
	}
	
	\medskip
	
	\noindent\textbf{Input}: $\varepsilon > 0$, step size $\eta > 0$, domain $D$, $f \in C^1_L(D,\mathbb{R})$ and its gradient $\nabla f$.
	
	\smallskip
	
	\noindent\textbf{Goal}: Compute any point where (projected) gradient descent for $f$ on $D$ terminates. Namely, find $x \in D$ such that $x$ and its next iterate $x' = \Pi_D(x-\eta\nabla f(x))$ satisfy:
	\begin{itemize}
		\item for \gdls/: \quad $f(x') \geq f(x) - \varepsilon,$ \hfill \textit{($f$ decreases by at most $\varepsilon$)}
		\item for \gdfp/: \hspace{1.1cm} \quad $\|x - x'\| \leq \varepsilon.$ \hfill \textit{($x'$ is $\varepsilon$-close to $x$)}
	\end{itemize}
\end{tcolorbox}

In a certain sense, \gdls/ is a \pls/-style version of Gradient Descent, while \gdfp/ is a \ppad/-style version.\footnote{A very similar version of \gdfp/ was also defined by \citet{DSZ21} and shown to be equivalent to finding an \emph{approximate} local minimum (which is essentially the same as a KKT point).} We show that these two versions are computationally equivalent by a triangle of reductions (see \cref{fig:reductions}). The other problem in that triangle of equivalent problems is the \kkt/ problem, defined below.

\begin{tcolorbox}
	{\centering
		\kkt/ \textit{(informal)}
		
	}
	
	\medskip

	\noindent\textbf{Input}: $\varepsilon > 0$, domain $D$, $f \in C^1_L(D,\mathbb{R})$ and its gradient $\nabla f$.
	
	\smallskip
	
	\noindent\textbf{Goal}: Compute any $\varepsilon$-KKT point of the minimization problem for $f$ on domain $D$.
\end{tcolorbox}

A point $x$ is a KKT point if $x$ is feasible (it belongs to the domain $D$), and $x$ is either a zero-gradient point of $f$, or alternatively $x$ is on the boundary of $D$ and the boundary constraints prevent local improvement of $f$. ``$\varepsilon$-KKT'' relaxes the KKT condition so as to allow inexact KKT solutions with limited numerical precision. For a formal definition of these notions see \cref{sec:nonlinearopt}.

\paragraph{\bf Representation of $\boldsymbol{f}$ and $\boldsymbol{\nabla f}$} We consider these computational problems in the ``white box'' model, where some computational device computing $f$ and $\nabla f$ is provided in the input. In our case, we assume that $f$ and $\nabla f$ are presented as arithmetic circuits. In more detail, following \citet{DaskalakisP2011-CLS}, we consider arithmetic circuits that use the operations $\{+,-,\times,\max,\min,<\}$, as well as rational constants.\footnote{A subtle issue is that it might not always be possible to evaluate such a circuit efficiently, because the $\times$-gates can be used to perform ``repeated squaring''. To avoid this issue, we restrict ourselves to what we call \emph{well-behaved} arithmetic circuits. See \cref{sec:arithmetic-cls} of the preliminaries for more details.} Another option would be to assume that the functions are given as polynomial-time Turing machines, but this introduces some extra clutter in the formal definitions of the problems. Overall, the definition with arithmetic circuits is cleaner, and, in any case, the complexity of the problems is the same in both cases.

\paragraph{\bf Promise-version and total-version} Given an arithmetic circuit for $f$ and one for $\nabla f$, we know of no easy way of checking that the circuit for $\nabla f$ indeed computes the gradient of $f$, and that the two functions are indeed $L$-Lipschitz. There are two ways to handle this issue: (a) consider the promise version of the problem, where we restrict our attention to instances that satisfy these conditions, or (b) introduce ``violation'' solutions in the spirit of \citet{DaskalakisP2011-CLS}, i.e. allow as a solution a witness of the fact that one of the conditions is not satisfied. The first option is more natural, but the second option ensures that the problem is formally in \tfnp/. Thus, we use the second option for the formal definitions of our problems in \cref{sec:nonlinearopt-problems}. However, we note that our ``promise-preserving'' reductions ensure that \emph{our hardness results also hold for the promise versions of the problems.}

\subsection{Complexity classes}\label{sec:classesdefs}

In this section we provide informal definitions of the relevant complexity classes, and discuss their key features. The formal definitions can be found in \cref{sec:model-classes-circuits}, but the high-level descriptions presented here are intended to be sufficient to follow the overview of our main proof in \cref{sec:proof-overview}.

\paragraph{\bf \ppad/} The complexity class \ppad/ is defined as the set of TFNP problems that reduce in polynomial time to the problem \eol/.

\begin{tcolorbox}
	{\centering
		\eol/ \textit{(informal)}
		
	}
	
	\medskip
	
	\noindent\textbf{Input}: A directed graph on the vertex set $[2^n]$, such that every vertex has in- and out-degree at most $1$, and such that vertex $1$ is a source.
	
	\smallskip
	
	\noindent\textbf{Goal}: Find a sink of the graph, or any other source.
\end{tcolorbox}

Importantly, the graph is not provided explicitly in the input, but instead we are given Boolean circuits that efficiently compute the successor and predecessor of each vertex. This means that the size of the graph can be exponential with respect to its description length. A problem is {\em complete} for \ppad/ if it belongs to \ppad/ and if \eol/ reduces in polynomial time to that problem. Many variants of the search for a fixed point of a Brouwer function turn out to be \ppad/-complete. This is essentially the reason why \gdfp/, and thus the other two equivalent problems, lie in \ppad/. See \cref{fig:EOL-example} for an example of an instance of \eol/.

\begin{figure}
	\centering
	\scalebox{0.8}{
		\includegraphics{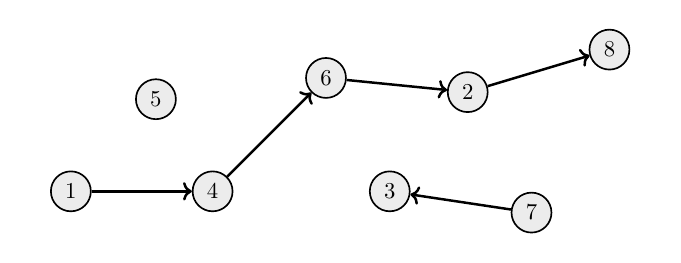}
	}
	\caption{Example of an \eol/ instance for $n=3$. The $2^n$ ($=8$) vertices are represented by circular nodes and the directed edges by arrows. Note that the graph is not provided explicitly in the input, but is only represented implicitly by a successor and predecessor circuit. In this example, the \eol/ solutions are the vertices $3$, $7$ and $8$. In more detail, vertices $3$ and $8$ are sinks, while vertex $7$ is a source. Note that the ``trivial'' source $1$ is not a solution. Finally, the isolated vertex $5$ is also not a solution.}
	\label{fig:EOL-example}
\end{figure}

\paragraph{\bf \pls/} The complexity class \pls/ is defined as the set of TFNP problems that reduce in polynomial time to the problem \textsc{Localopt}.

\begin{tcolorbox}
	{\centering
		\textsc{Localopt} \textit{(informal)}
		
	}
	
	\medskip
	
	\noindent\textbf{Input}: Functions $V: [2^n] \to \mathbb{R}$ and $S: [2^n] \to [2^n]$.
	
	\smallskip
	
	\noindent\textbf{Goal}: Find $v \in [2^n]$ such that $V(S(v)) \geq V(v)$.
\end{tcolorbox}

The functions are given as Boolean circuits. A problem is {\em complete} for \pls/ if it belongs to \pls/ and if \textsc{Localopt} reduces in polynomial time to that problem. \pls/ embodies general local search methods where one attempts to optimize some objective function by considering local improving moves. Our problem \gdls/ is essentially a special case of local search, and thus lies in \pls/. In this paper we make use of the problem \iter/, defined below, which is known to be \pls/-complete \citep{Morioka01-Mthesis-PLS}.

\begin{tcolorbox}
	{\centering
		\iter/ \textit{(informal)}
		
	}
	
	\medskip
	
	\noindent\textbf{Input}: A function $C: [2^n] \to [2^n]$ such that $C(v) \geq v$ for all $v \in [2^n]$, and $C(1) > 1$.
	
	\smallskip
	
	\noindent\textbf{Goal}: Find $v$ such that $C(v) > v$ and $C(C(v)) = C(v)$.
\end{tcolorbox}

For this problem, it is convenient to think of the nodes in $[2^n]$ as lying on a line, in increasing order. Then, any node is either a fixed point of $C$, or it is mapped to some node further to the right. We are looking for any node that is not a fixed point, but is mapped to a fixed point. It is easy to see that the condition $C(1) > 1$ ensures that such a solution must exist. See \cref{fig:ITER-example} for an example of an instance of \iter/.

\begin{figure}
	\centering
	\includegraphics{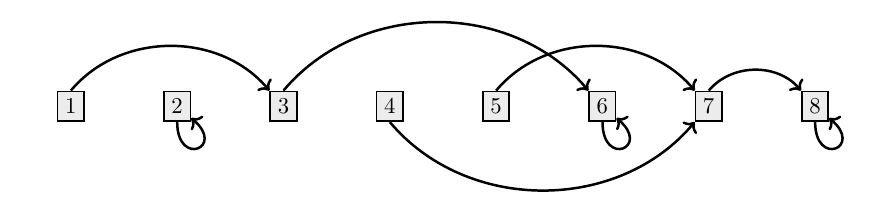}
	\caption{Example of an \iter/ instance $C$ for $n=3$. The $2^n$ ($=8$) nodes are represented by squares. The arrows indicate the mapping given by the circuit $C$. In this example, nodes $2$, $6$ and $8$ are the fixed points of $C$. Any node that is mapped by $C$ to a fixed point is a solution to the \iter/ instance. Thus, in this example, the solutions are nodes $3$ and $7$.}\label{fig:ITER-example}
\end{figure}

\paragraph{\bf $\boldsymbol{\ppadpls/}$} The class \ppadpls/ contains, by definition, all TFNP problems that lie in both \ppad/ and in \pls/. Prior to our work, the only known way to obtain \ppadpls/-complete problems was to combine a \ppad/-complete problem $A$ and a \pls/-complete problem $B$ as follows \citep{DaskalakisP2011-CLS}.

\begin{tcolorbox}
	{\centering
		\textsc{Either-Solution($A$,$B$)}
		
	}
	
	\medskip
	
	\noindent\textbf{Input}: An instance $I_A$ of $A$ and an instance $I_B$ of $B$.
	
	\smallskip
	
	\noindent\textbf{Goal}: Find a solution of $I_A$ or a solution of $I_B$.
\end{tcolorbox}

In particular, the problem \textsc{Either-Solution(\eol/,\iter/)} is \ppadpls/-complete, and this is the problem we reduce from to obtain our results.

\paragraph{\bf \cls/} Noting that all known \ppadpls/-complete problems looked very artificial, \citet{DaskalakisP2011-CLS} defined the class \cls/\,$\subseteq$\,\ppadpls/, which combines \ppad/ and \pls/ in a more natural way. The class \cls/ is defined as the set of TFNP problems that reduce to the problem 3D-\clo/.

\begin{tcolorbox}
	{\centering
		3D-\clo/ \textit{(informal)}
		
	}
	
	\medskip
	
	\noindent\textbf{Input}: $\varepsilon > 0$, $L$-Lipschitz functions $p: [0,1]^3 \to [0,1]$ and $g: [0,1]^3 \to [0,1]^3$.
	
	\smallskip
	
	\noindent\textbf{Goal}: Compute any approximate local optimum of $p$ with respect to $g$. Namely, find $x \in [0,1]^3$ such that
	$$p(g(x)) \geq p(x) - \varepsilon.$$
\end{tcolorbox}

This problem is essentially a special case of the \textsc{Localopt} problem, where we perform local search over a continuous domain and where the functions are continuous. The formal definition of 3D-\clo/ includes violation solutions for the Lipschitz-continuity of the functions. We also consider a more general version of this problem, which we call \gclo/, where we allow any bounded convex polytope as the domain.

\subsection{Results}\label{sec:results}

The main technical contribution of this work is \cref{thm:main-kkt-hard}, which shows that the \kkt/ problem is \ppadpls/-hard, even when the domain is the unit square $[0,1]^2$. The hardness also holds for the promise version of the problem, because the hard instances that we construct always satisfy the promises. We present the main ideas needed for this result in the next section, but we first briefly present the consequences of this reduction here.

A chain of reductions, presented in \cref{sec:all-ppadpls-complete} and shown in \cref{fig:reductions}, which includes the ``triangle'' between the three problems of interest, establishes the following theorem.

\theoremstyle{plain}
\newtheorem*{thm:overview:all-ppadpls-complete}{Theorem \ref*{thm:all-ppadpls-complete}}
\begin{thm:overview:all-ppadpls-complete}
	The problems \kkt/, \gdls/, \gdfp/ and \gclo/ are $\ppad/~\cap$ \pls/-complete, even when the domain is fixed to be the unit square $[0,1]^2$. This hardness result continues to hold even if one considers the promise-versions of these problems, i.e., only instances without violations.
\end{thm:overview:all-ppadpls-complete}

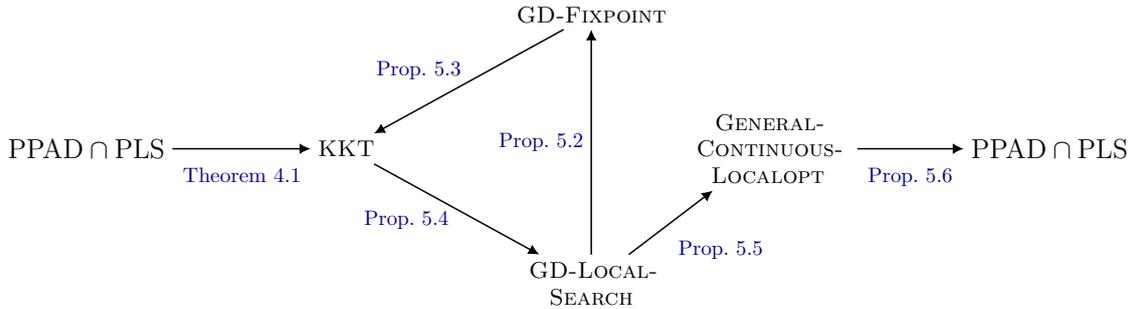
\begin{figure}[h]
	\centering
	\scalebox{0.72}{ 
		\begin{tikzpicture}[node distance=3.5cm]
		\tikzset{>={Latex[width=2mm,length=2mm]}}

		\tikzstyle{cc}=[font=\Large]
		\tikzstyle{prob}=[font=\large]
		
		\node[cc] (pp) {\ppadpls/};
		\node[prob, right of=pp, xshift=1.25cm] (kkt) {\kkt/};
		\node[prob, above right of=kkt, xshift=2cm, text width=3cm, align=center] (gdfp) {\gdfp/};
		\node[prob, below right of=kkt, xshift=2cm, text width=3.5cm, align=center] (gdls) {\gdls/};
		\node[prob, above right of=gdls, xshift=0.75cm, text width=3cm, align=center] (gclo) {\gclo/};
		\node[cc, right of=gclo, xshift=1.65cm] (pp2) {\ppadpls/};
		
		\path[->, thick] (pp) edge node[below, yshift=-2mm, text width=4cm, align=center] {\cref{thm:main-kkt-hard}} (kkt);
		\path[->, thick] (gdfp) edge node[left, yshift=2mm] {\cshref{prop:gdfp2kkt}} (kkt);
		\path[->, thick] (kkt) edge node[left, yshift=-2mm] {\cshref{prop:kkt2gdls}} (gdls);
		\path[->, thick] (gdls) edge node[left] {\cshref{prop:gdls2gdfp}} (gdfp);
		\path[->, thick] (gdls) edge node[right, xshift=0mm, yshift=-5mm] {\cshref{prop:gdls2gclo}} (gclo);
		\path[->, thick] (gclo) edge node[below, yshift=-2mm] {\cshref{prop:gclo2ppadpls}} (pp2);
		
		\end{tikzpicture}
	}
	\caption{Our reductions. The main one (\cref{thm:main-kkt-hard}) is on the left; note that the other reductions are all domain- and promise-preserving.}\label{fig:reductions}
\end{figure}

These reductions are domain-preserving---which means that they leave the domain $D$ unchanged---and promise-preserving---which means that they are also valid reductions between the promise versions of the problems. As a result, the other problems ``inherit'' the hardness result for \kkt/, including the fact that it holds for $D=[0,1]^2$ and even for the promise versions.

\paragraph{\bf Consequences for \cls/} The \ppadpls/-hardness of \gclo/ on domain $[0,1]^2$, and thus also on domain $[0,1]^3$, immediately implies the following surprising collapse.

\theoremstyle{plain}
\newtheorem*{thm:overview:cls-equal-ppadpls}{Theorem \ref*{thm:cls-equal-ppadpls}}
\begin{thm:overview:cls-equal-ppadpls}
	\cls/ $=$ \ppadpls/.
\end{thm:overview:cls-equal-ppadpls}

\noindent As a result, it also immediately follows that the two known \cls/-complete problems \citep{DaskTZ18,FGMS17} are in fact \ppadpls/-complete.

\theoremstyle{plain}
\newtheorem*{cor:overview:banach-metametric}{Corollary \ref*{cor:banach-metametric}}
\begin{cor:overview:banach-metametric}
	\textup{\textsc{Banach}} and \textup{\textsc{MetametricContraction}} are \ppadpls/-complete.
\end{cor:overview:banach-metametric}

\noindent Our results also show that the definition of \cls/ is robust to various modifications. The fact that our hardness result holds on domain $[0,1]^2$ implies
that the $n$-dimensional variant of \cls/ is equal to the two-dimensional version,
a fact that was not previously known. 
Furthermore, since our results hold even for the promise version of \gclo/, this implies that the definition of \cls/ is
robust with respect to the removal of violations
(promise-\cls/\,=\,\cls/). Finally, we also show that restricting the circuits
to be linear arithmetic circuits (that compute piecewise-linear
functions) does not yield a weaker class, i.e., 2D-Linear-\cls/\,=\,\cls/. This
result is obtained by showing that linear circuits can be used to efficiently
approximate any Lipschitz-continuous function with \emph{arbitrary precision}
(\cref{sec:approx-linear-circuit}), which might be of independent interest. All
the consequences for \cls/ are discussed in detail in \cref{sec:cls}. In that section, we also define a Gradient Descent problem where we do not have access to the gradient of the function (which might, in fact, not even be differentiable) and instead use ``finite differences'' to compute an approximate gradient. We show that this problem remains \ppadpls/-complete.

\subsection{Proof overview for \texorpdfstring{\cref{thm:main-kkt-hard}}{Theorem~\ref*{thm:main-kkt-hard}}}\label{sec:proof-overview}

In this section we provide a brief overview of our reduction from the \ppadpls/-complete problem \textsc{Either-Solution} \textsc{(\eol/,\iter/)} to the \kkt/ problem on domain $[0,1]^2$. We note that the proof can be simplified using subsequent work by \citet{GoosHJMPRT22-collapses}. See \cref{sec:future-directions} for more details on this.

Given an instance $I^{\rm EOL}$ of \eol/ and an instance $I^{\rm ITER}$ of \iter/, we construct an instance $I^{\rm KKT}=(\varepsilon, f, \nabla f, L)$ of the \kkt/ problem on domain $[0,1]^2$ such that from any $\varepsilon$-KKT point of $f$, we can efficiently obtain a solution to either $I^{\rm EOL}$ or $I^{\rm ITER}$. The function $f$ and its gradient $\nabla f$ are first defined on an exponentially small grid on $[0,1]^2$, and then extended within every small square of the grid by using bicubic interpolation. This ensures that the function is continuously differentiable on the whole domain. The most interesting part of the reduction is how the function is defined on the grid points, by using information from $I^{\rm EOL}$, and then, where necessary, also from $I^{\rm ITER}$.

\paragraph{\bf Embedding $\boldsymbol{I^{\rm EOL}}$} The domain is first subdivided into
$2^n \times 2^n$ \emph{big squares}, where $[2^n]$ is the set of vertices in $I^{\rm EOL}$. The big squares on the diagonal (shaded in \cref{fig:overview-eol-embedding}) represent the vertices of $I^{\rm EOL}$ and the function $f$ is constructed so as to embed the directed edges in the graph of $I^{\rm EOL}$. If the edge $(v_1,v_2)$ in $I^{\rm EOL}$ is a forward edge, i.e, $v_1 < v_2$, then there will be a ``green path'' going from the big square of $v_1$ to the big square of $v_2$. On the other hand, if the edge $(v_1,v_2)$ in $I^{\rm EOL}$ is a backward edge, i.e., $v_1 > v_2$, then there will be an ``orange path'' going from the big square of $v_1$ to the big square of $v_2$. These paths are shown in \cref{fig:overview-eol-embedding} for the corresponding example instance of \cref{fig:EOL-example}. The idea of embedding the vertices on the diagonal, in a ``staircase embedding'', was introduced by \citet{HubacekY2017-CLS}. However, their work only required the embedding of forward edges, whereas we have to be able to implement backward edges as well.

\begin{figure}[t]
	\centering
	\includegraphics[width=10cm]{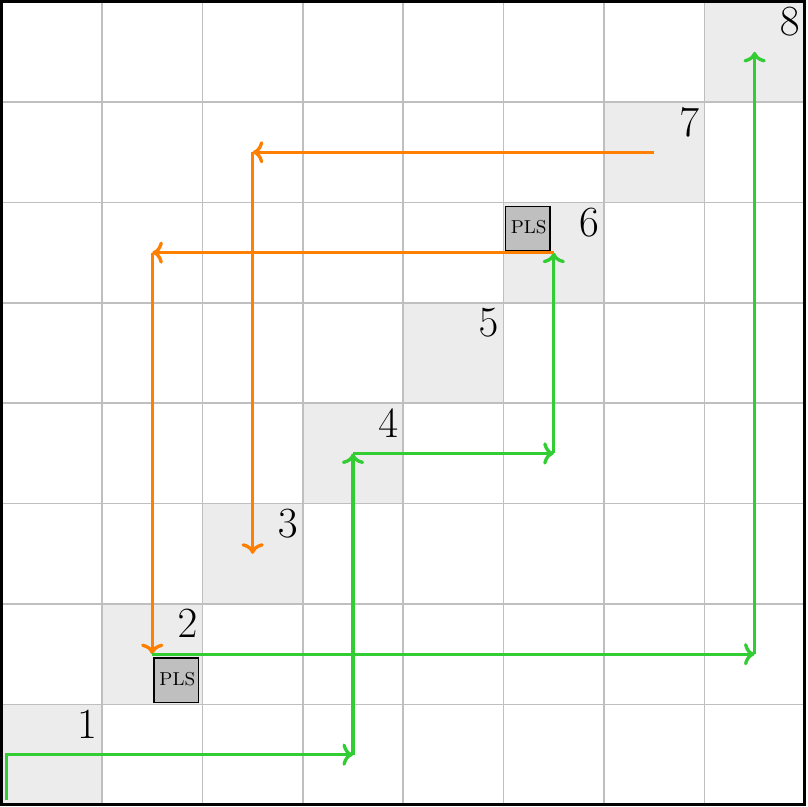}
	\caption{A high-level illustration of our construction. The shaded squares on the diagonal correspond to vertices of the graph represented by $I^{\rm EOL}$, in this case corresponding to the graph in \cref{fig:EOL-example}. The green and orange arrows encode the directed edges of the graph. The positions where $I^{\rm ITER}$ is encoded, i.e., the PLS-Labyrinths, are shown as boxes labelled ``PLS''. They are located at points where the embedding of $I^{\rm EOL}$ would introduce false solutions, and their purpose is to hide those false solutions by co-locating any such solution with a solution to $I^{\rm ITER}$.}
	\label{fig:overview-eol-embedding}
\end{figure}

The function $f$ is constructed such that when we move along a green path the value of $f$ decreases. Conversely, when we move along an orange path the value of $f$ increases. Outside the paths, $f$ is defined so as to decrease towards the origin $(0,0) \in [0,1]^2$, where the green path corresponding to the source of $I^{\rm EOL}$ starts. As a result, we show that an $\varepsilon$-KKT point can only occur in a big square corresponding to a vertex $v$ of $I^{\rm EOL}$ such that (a) $v$ is a solution of $I^{\rm EOL}$, or (b) $v$ is \emph{not} a solution of $I^{\rm EOL}$, but its two neighbours (in the $I^{\rm EOL}$ graph) are both greater than $v$, or alternatively both less than $v$. Case (b) exactly corresponds to the case where a green path ``meets'' an orange path. In that case, it is easy to see that an $\varepsilon$-KKT point is unavoidable.

\paragraph{\bf The PLS-Labyrinth} In order to resolve the issue with case (b) above, we use the following idea: hide the (unavoidable) $\varepsilon$-KKT point in such a way that locating it requires solving $I^{\rm ITER}$! This is implemented by introducing a gadget, that we call the PLS-Labyrinth, at the point where the green and orange paths meet (within some big square). An important point is that the PLS-Labyrinth only works properly when it is positioned at such a meeting point. If it is positioned elsewhere, then it will either just introduce additional unneeded $\varepsilon$-KKT points, or even introduce $\varepsilon$-KKT points that are easy to locate. Indeed, if we were able to position the PLS-Labyrinth wherever we wanted, this would presumably allow us to show \pls/-hardness, which as we noted earlier we do not expect. In \cref{fig:overview-eol-embedding}, the positions where a PLS-Labyrinth is introduced are shown as grey boxes labelled ``PLS''.

Every PLS-Labyrinth is subdivided into exponentially many \emph{medium squares} such that the medium squares on the diagonal (shaded in \cref{fig:overview-iter-embedding}) correspond to the nodes of $I^{\rm ITER}$. The point where the green and orange paths meet, which lies just outside the PLS-Labyrinth, creates an ``orange-blue path'' which then makes its way to the centre of the medium square for node $1$ of $I^{\rm ITER}$. Similarly, for every node $u$ of $I^{\rm ITER}$ that is a candidate to be a solution (i.e., with $C(u)>u$), there is an orange-blue path starting from the orange path (which runs along the PLS-Labyrinth) and going to the centre of the medium square corresponding to $u$. Sinks of orange-blue paths introduce $\varepsilon$-KKT points, and so for those $u$ that are not solutions of $I^{\rm ITER}$, the orange-blue path of $u$ turns into a ``blue path'' that goes and merges into the orange-blue path of $C(u)$. This ensures that sinks of orange-blue paths (that do not turn into blue paths) exactly correspond to the solutions of $I^{\rm ITER}$. An interesting point to note is that sources of blue paths do \emph{not} introduce $\varepsilon$-KKT points. This allows us to handle crossings between paths in a straightforward way. \cref{fig:overview-iter-embedding} shows an overview of the PLS-Labyrinth that encodes the \iter/ example of \cref{fig:ITER-example}.

\begin{figure}[t]
	\centering
	\scalebox{0.5}{
		\includegraphics{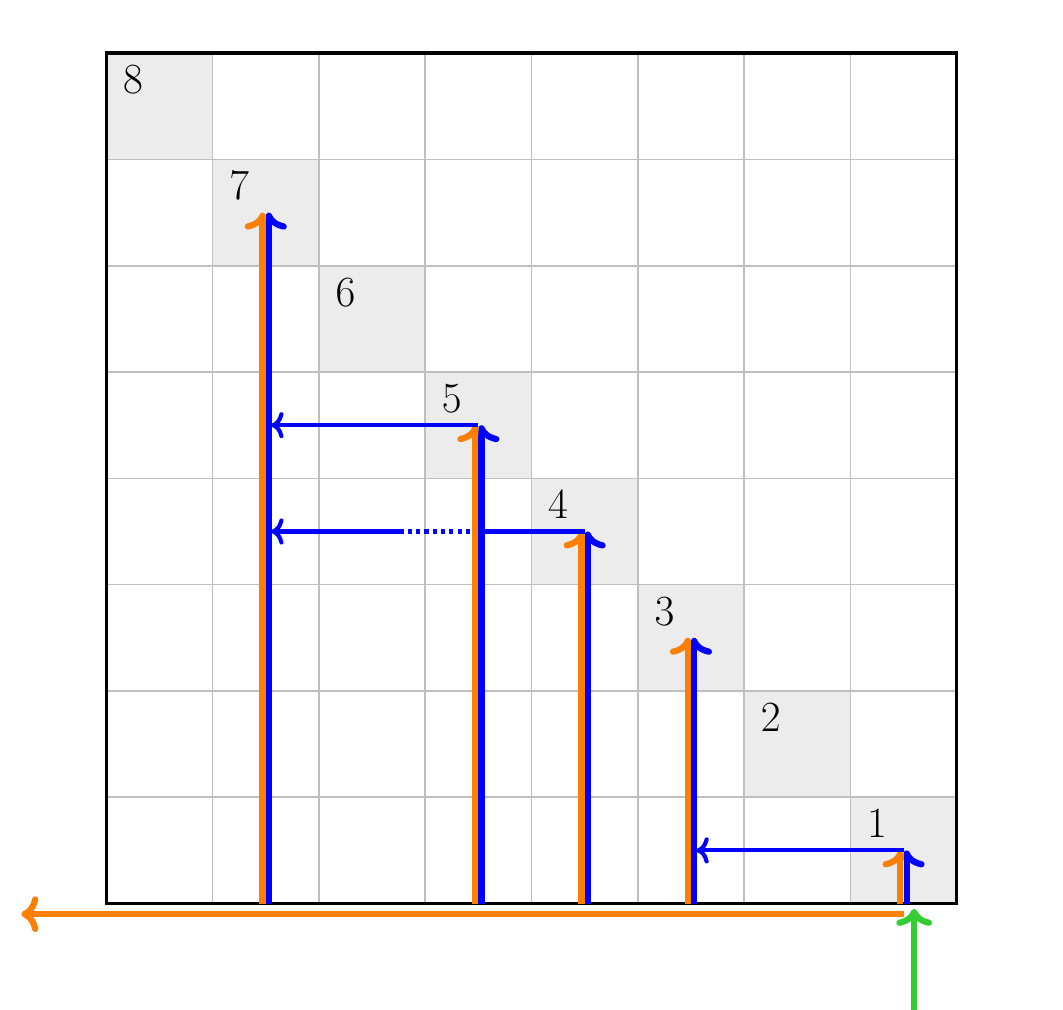}
	}
	\caption{High-level illustration of the \pls/-Labyrinth corresponding to the \iter/ example of \cref{fig:ITER-example}. Shaded squares on the diagonal correspond to the nodes of \iter/. Colours of lines determine how $f$ is constructed at these points. The horizontal blue lines (pointing left) correspond to the 3 edges in \cref{fig:ITER-example} that go out from non-solutions, and we do not use similar lines going out from solutions (nodes 3 and 7).}\label{fig:overview-iter-embedding}
\end{figure}

\paragraph{\bf Bicubic interpolation}
In our construction, we specify how the objective function
$f$ behaves within the ``small squares'' of $[0,1]^2$. At this stage, we have
values of $f$ and $\nabla f$ at the corners of the small squares, and we then
need to smoothly interpolate within the interior of each square. We use 
bicubic interpolation to do this. It constructs a
polynomial over every small square given values for $f$ and $\nabla f$ at the 
square's corners, in such a way that putting all these pieces together yields a continuously differentiable function over the whole domain $[0,1]^2$. We note that the simpler \emph{bilinear} interpolation, which was used by \citet{HubacekY2017-CLS}, yields a continuous function, but not a continuously differentiable function over the whole domain. Since \citet{HubacekY2017-CLS} reduce to continuous local search, it is sufficient in their case to construct two continuous functions $p$ and $g$. However, in our case we reduce to a gradient descent problem, so we have to construct a single continuously differentiable function $f$.

We must prove that using bicubic interpolation does not introduce any $\varepsilon$-KKT points
within any small square, unless that small square
corresponds to a solution of  $I^{\rm ITER}$ or $I^{\rm EOL}$. 
Each individual small square leads to a different class of polynomials, based on
the color-coding of the grid point, and the direction of the gradient at each
grid point. Our construction uses 101 distinct small squares, and we must prove
that no unwanted solutions are introduced in any of them. By making use of various symmetries we are able to group these 101 squares into just four different cases for which we can directly verify that the desired statement holds: an $\varepsilon$-KKT point can only appear in a small square that yields a solution of $I^{\rm ITER}$ or $I^{\rm EOL}$.

\section{Preliminaries}\label{sec:prelims}

Let $n \in \mathbb{N}$ be a positive integer. Throughout this paper we use $\|\cdot\|$ to denote the standard Euclidean norm in $n$-dimensional space, i.e., the $\ell_2$-norm in $\mathbb{R}^n$. The maximum-norm, or $\ell_\infty$-norm, is denoted by $\|\cdot\|_\infty$. For $x,y \in \mathbb{R}^n$, $\langle x,y \rangle := \sum_{i=1}^n x_i y_i$ denotes the inner product. For any non-empty closed convex set $D \subseteq \mathbb{R}^n$, let $\Pi_D: \mathbb{R}^n \to D$ denote the projection onto $D$ with respect to the Euclidean norm. Formally, for any $x \in \mathbb{R}^n$, $\Pi_D(x)$ is the unique point $y \in D$ that minimizes $\|x-y\|$. For $k \in \mathbb{N}$, let $[k] := \{1,2,\dots,n\}$.

\subsection{Computational model, classes and arithmetic circuits}\label{sec:model-classes-circuits}

We work in the standard Turing machine model. Rational numbers are represented as irreducible fractions, with the numerator and denominator of the irreducible fraction given in binary. Note that given any fraction, it can be made irreducible in polynomial time using the Euclidean algorithm. For a rational number $x$, we let $\sz(x)$ denote the number of bits needed to represent $x$, i.e., the number of bits needed to write down the numerator and denominator (in binary) of the irreducible fraction for $x$.

\subsubsection{NP total search problems and reductions}\label{sec:def-tfnp}

\paragraph{\bf Search Problems} Let $\{0,1\}^*$ denote the set of all finite length bit-strings and let $|x|$ be the length of $x \in \{0,1\}^*$. A computational search problem is given by a relation $R \subseteq \{0,1\}^* \times \{0,1\}^*$, interpreted as the following problem: given an instance $x \in \{0,1\}^*$, find $w \in \{0,1\}^*$ such that $(x,w) \in R$, or return that no such $w$ exists.

The search problem $R$ is in FNP (\emph{search problems in NP}), if $R$ is polynomial-time computable (i.e., $(x,w) \in R$ can be decided in polynomial time in $|x|+|w|$) and polynomially-balanced (i.e., there exists some polynomial $p$ such that $(x,w) \in R \implies |w| \leq p(|x|)$). Intuitively, FNP contains all search problems where all solutions have size polynomial in the size of the instance and any solution can be checked in polynomial time. (The solution $w$ is thought-of as a witness.) The class of all search problems in FNP that can be solved by a polynomial-time algorithm is denoted by FP. The question FP vs.\ FNP is equivalent to the P vs.\ NP question.

The class TFNP (\emph{total search problems in NP}) is defined as the set of all FNP problems~$R$ that are \emph{total}, i.e., every instance has at least one solution. Formally, $R$ is total, if for every $x \in \{0,1\}^*$ there exists $w \in \{0,1\}^*$ such that $(x,w) \in R$. In a certain sense,\footnote{Clearly, TFNP $\subseteq$ FNP, but with a slight abuse of notation we can also say that FP $\subseteq$ TFNP. Indeed, any problem $R$ in FP can be turned into a TFNP problem by including a pair $(x,\textup{NO})$ in $R$ for each instance $x$ that does not have a solution, where $\textup{NO}$ is some dedicated bit-string. Note that despite this minor modification, the search problem remains essentially the same.} TFNP lies between FP and FNP.

Note that the totality of TFNP problems should not rely on any promise. Instead, there is a \emph{syntactic} guarantee of totality: for any instance, there is a solution. It is easy to see that a TFNP problem cannot be NP-hard, unless NP = co-NP. Furthermore, it is also believed that no TFNP-complete problem exists. For more details on this, see \citep{MegiddoP1991-tfnp}.

\paragraph{\bf Reductions between TFNP problems} Let $R$ and $S$ be two TFNP problems. We say that $R$ reduces to $S$ if there exist polynomial-time computable functions $f : \{0,1\}^* \to \{0,1\}^*$ and $g : \{0,1\}^* \times \{0,1\}^* \to \{0,1\}^*$ such that, for all $x,w \in \{0,1\}^*$,
$$(f(x), w) \in S \implies (x, g(x,w)) \in R$$
Intuitively, this says that for any instance $x$ of $R$, if we can find a solution $w$ to instance $f(x)$ of $S$, then $g(x,w)$ gives us a solution to instance $x$ of $R$. In particular, note that if $S$ is polynomial-time solvable, then so is $R$.

\subsubsection{The classes \ppad/, \pls/ and \ppadpls/}

Since TFNP problems likely cannot be NP-hard, or TFNP-complete, one instead attempts to classify the problems inside TFNP. Various subclasses of TFNP have been defined and natural problems have been proved complete for these subclasses. In this section we formally define the subclasses \ppad/ and \pls/, which have both been very successful in capturing the complexity of interesting problems.

The most convenient way to define these classes is using problems on Boolean circuits. A Boolean circuit $C: \{0,1\}^n \to \{0,1\}^n$ with $n$ inputs and $n$ outputs, is allowed to use the logic gates $\land$ (AND), $\lor$ (OR) and $\lnot$ (NOT), where the $\land$ and $\lor$ gates have fan-in 2, and the $\lnot$ gate has fan-in 1. For ease of notation, we identify $\{0,1\}^n$ with $[2^n]$.

\paragraph{\bf \ppad/} The class \ppad/ is defined as the set of all TFNP problems that reduce to the problem \eol/ \citep{Pap94,DGP09}.

\begin{tcolorbox}
	\begin{definition}\label{def:eol}
		\eol/:
		
		\noindent\textbf{Input}: Boolean circuits $S,P : [2^n] \to [2^n]$ with $P(1) = 1 \neq S(1)$.
		
		\noindent\textbf{Goal}: Find $v \in [2^n]$ such that $P(S(v)) \neq v$ or $S(P(v)) \neq v \neq 1$.
	\end{definition}
\end{tcolorbox}

The successor circuit $S$ and the predecessor circuit $P$ implicitly define a directed graph on the vertex set $[2^n]$. There is an edge from $v_1$ to $v_2$ if $S(v_1)=v_2$ and $P(v_2)=v_1$. Every vertex has at most one outgoing edge and at most one incoming edge. Since the vertex $1$ has one outgoing edge and no incoming edge, it is a source. The goal is to find another end of line, i.e., another source, or a sink of the graph. Note that such a vertex is guaranteed to exist. The condition $P(1) = 1 \neq S(1)$ can be enforced syntactically, so this is indeed a TFNP problem and not a promise problem. See \cref{fig:EOL-example} for an example of an \eol/ instance.

\paragraph{\bf \pls/} The class \pls/ is defined as the set of all TFNP problems that reduce to the problem \textsc{Localopt} \citep{JPY88,DaskalakisP2011-CLS}.

\begin{tcolorbox}
	\begin{definition}
		\textsc{Localopt}:
		
		\noindent\textbf{Input}: Boolean circuits $S,V : [2^n] \to [2^n]$.
		
		\noindent\textbf{Goal}: Find $v \in [2^n]$ such that $V(S(v)) \geq V(v)$.
	\end{definition}
\end{tcolorbox}

This problem embodies local search over the node set $[2^n]$. The output of the circuit $V$ represents a value and ideally we would like to find a node $v \in [2^n]$ that minimises $V(v)$. The circuit $S$ helps us in this task by proposing a possibly improving node $S(v)$ for any $v$. We stop our search, when we find a $v$ such that $V(S(v)) \geq V(v)$, i.e., $S$ no longer helps us decrease the value of $V$. This is local search, because the circuit $S$ represents the search for an improving node in some small (polynomial-size) neighbourhood.

In this paper, we also make use of the following \pls/-complete problem \citep{Morioka01-Mthesis-PLS}.

\begin{tcolorbox}
	\begin{definition}\label{def:iter}
		\iter/:
		
		\noindent\textbf{Input}: Boolean circuit $C : [2^n] \to [2^n]$ with $C(1) > 1$.
		
		\noindent\textbf{Goal}: Find $v$ such that either
		\begin{itemize}
			\item $C(v) < v$, or
			\item $C(v) > v$ and $C(C(v)) = C(v)$.
		\end{itemize}
	\end{definition}
\end{tcolorbox}

In this problem, it is convenient to think of the nodes in $[2^n]$ as lying on a line from left to right. Then, we are looking for any node $v$ that is mapped to the left by $C$, or any node $v$ that is mapped to the right and such that $C(v)$ is a fixed point of $C$. Since $C(1) > 1$, i.e., node $1$ is mapped to the right, it is easy to see that such a solution must exist (apply $C$ repeatedly on node $1$). Note that the condition $C(1) > 1$ can be enforced syntactically, so this is indeed a TFNP problem and not a promise problem. See \cref{fig:ITER-example} for an example of an \iter/ instance.

\paragraph{\bf $\boldsymbol{\ppadpls/}$} The class \ppadpls/ is the set of all TFNP problems that lie in both \ppad/ and in \pls/. A problem in \ppadpls/ cannot be \ppad/- or \pls/-complete, unless \ppad/ $\subseteq$ \pls/ or \pls/ $\subseteq$ \ppad/. Neither of these two containments is believed to hold, and this is supported by oracle separations between the classes \citep{Morioka01-Mthesis-PLS,BureshM04-NP-search-problems,buss2012propositional}. It is easy to construct ``artificial'' \ppadpls/-complete problems from \ppad/- and \pls/-complete problems.

\begin{proposition}[{\citet{DaskalakisP2011-CLS}}]\label{prop:either}
	For any \tfnp/ problems $A$ and $B$, let \textup{\textsc{Either-Solution($A$,$B$)}} denote the problem: given an instance $I_A$ of $A$ and an instance $I_B$ of $B$, find a solution of $I_A$ or a solution of $I_B$. If $A$ is \ppad/-complete and $B$ is \pls/-complete, then \textup{\textsc{Either-Solution($A$,$B$)}} is \ppadpls/-complete.
\end{proposition}

\noindent As a result, we obtain the following corollary, which we will use to show our main \ppadpls/-hardness result.

\begin{corollary}\label{cor:either}
	\textup{\textsc{Either-Solution(\eol/,\iter/)}} is \ppadpls/-complete.
\end{corollary}

\noindent Prior to our work, the problems \textsc{Either-Solution($A$,$B$)}, where $A$ is \ppad/-complete and $B$ is \pls/-complete, were the only known \ppadpls/-complete problems.

\subsubsection{Arithmetic circuits and the class \cls/}\label{sec:arithmetic-cls}

Noting that \ppadpls/ only seemed to have artificial complete problems, \citet{DaskalakisP2011-CLS} defined a subclass of \ppadpls/ with a more natural definition, that combines \ppad/ and \pls/ nicely in a single problem. Unlike \ppad/ and \pls/, \cls/ is defined using arithmetic circuits.

\paragraph{\bf Arithmetic circuits} An arithmetic circuit representing a function $f: \mathbb{R}^n \to \mathbb{R}^m$ is a circuit with $n$ inputs and $m$ outputs, and every internal node is a gate with fan-in 2 performing an operation in $\{+, -, \times, \max, \min, >\}$ or a rational constant (modelled as a gate with fan-in 0). The comparison gate $>$, on input $a,b \in \mathbb{R}$, outputs $1$ if $a > b$, and $0$ otherwise. For an arithmetic circuit $f$, we let $\sz(f)$ denote the size of the circuit, i.e., the number of bits needed to describe the circuit, including the rational constants used therein. Obviously, there are various different ways of defining arithmetic circuits, depending on which gates we allow. The definition we use here is the same as the one used by \citet{DaskalakisP2011-CLS} in their original definition of \cls/. Many variants are equivalent to this definition, for example $<$, $\leq$, and  $\geq$ can be implemented using the given operations.

These circuits are very natural, but they suffer from a subtle issue that seems to have been overlooked in some prior works. Using the multiplication gate, such an arithmetic circuit can perform repeated squaring to construct numbers that have exponential representation size with respect to the size of the circuit and the input to the circuit. In other words, the circuit can construct numbers that are \emph{doubly} exponential (or the inverse thereof). Thus, in some cases, it might not be possible to evaluate the circuit on some input efficiently, i.e., in time polynomial in the size of the circuit and the given input. Indeed, evaluating the circuit is at least as hard as solving the PosSLP problem (of checking whether the output of a given straight-line program is positive) introduced by \citet{AllenderBKM09}, which is not believed to be polynomial-time solvable.

This subtle issue was recently also noticed by Daskalakis and Papadimitriou, who proposed a way to fix it in a corrigendum\footnote{\url{http://people.csail.mit.edu/costis/CLS-corrigendum.pdf}} to the definition of \cls/. Their modification consists in having an additional input $K$ (in unary) provided as part of the input such that the evaluation of the arithmetic circuit---purportedly---only involves numbers of bit-size at most $K \cdot \sz(x)$ on input $x$. Any point $x$ where the arithmetic circuit fails to satisfy this property is accepted as a solution.

In this paper, we use an alternative way to resolve the issue. We restrict our attention to what we call \emph{well-behaved} arithmetic circuits. An arithmetic circuit $f$ is well-behaved if, on any directed path that leads to an output, there are at most $\log (\sz(f))$ \emph{true} multiplication gates. A true multiplication gate is one where both inputs are non-constant nodes of the circuit. In particular, note that we allow our circuits to perform multiplication by a constant as often as needed without any restriction. Indeed, these operations cannot be used to do repeated squaring.

It is easy to see that given an arithmetic circuit $f$, we can check in polynomial time whether $f$ is well-behaved. Furthermore, these circuits can always be efficiently evaluated.

\begin{lemma}\label{lem:well-behaved-efficient}
	Let $f$ be a well-behaved arithmetic circuit with $n$ inputs. Then, for any rational $x \in \mathbb{R}^n$, $f(x)$ can be computed in time $\textup{poly}(\sz(f),\sz(x))$.
\end{lemma}
\noindent We provide a proof of this Lemma in \cref{app:more-circuits}.

Using well-behaved arithmetic circuits, instead of the solution proposed by Daskalakis and Papadimitriou, has the advantage that we do not need to add any additional inputs, or any additional violation solutions to our problems. Indeed, the restriction to well-behaved circuits can be enforced syntactically. Furthermore, we note that our problems defined with well-behaved circuits easily reduce to the versions using the solution proposed by Daskalakis and Papadimitriou (see \cref{rem:well-behaved-corrigendum} below). Thus, this restriction only makes our hardness results stronger. In fact, for \cls/ we show that restricting the circuits even further to only use gates $\{+,-,\max,\min, \times \zeta\}$ and rational constants (where $\times \zeta$ is multiplication by a constant), so-called linear arithmetic circuits (representing piecewise linear functions), does not make the class any weaker (see \cref{sec:linear-cls}).

For the problems we consider, it is quite convenient to use arithmetic circuits instead of, say, polynomial-time Turing machines to represent the functions involved. Indeed, the problems could also be defined with polynomial-time Turing machines, but that would introduce some technical subtleties in the definitions (the polynomial used as an upper bound on the running time of the machines would have to be fixed). The important thing to note is that the Turing machine variants of the problems would continue to lie in \ppadpls/. Thus, using arithmetic circuits just makes our hardness results stronger. Note also that in the hard instances we construct, the arithmetic circuits only perform a constant number of true multiplications (see the proof of \cref{lem:kkt-function-properties}).

\begin{remark}\label{rem:well-behaved-corrigendum}
	The proof of \cref{lem:well-behaved-efficient} (in \cref{app:more-circuits}) shows that if we evaluate a well-behaved arithmetic circuit $f$ on some input $x$, then, the value $v(g)$ at any gate $g$ of the circuit will satisfy $\sz(v(g)) \leq 6 \cdot \sz(f)^3 \cdot \sz(x)$. As a result, it immediately follows that problems with well-behaved arithmetic circuits can be reduced to the versions of the problems with the modification proposed by Daskalakis and Papadimitriou in the corrigendum of the \cls/ paper. Indeed, it suffices to let $K = 6 \cdot \sz(f)^3$, which can be written down in unary. In particular, this holds for the definition of \cls/.
\end{remark}

\begin{remark}
	Our definition of well-behaved circuits is robust in the following sense. For any $k \in \mathbb{N}$, say that a circuit $f$ is $k$-well-behaved if, on any path that leads to an output, there are at most $k \cdot \log (\sz(f))$ true multiplication gates. In particular, a circuit is well-behaved if it is 1-well-behaved. It is easy to see that for any fixed $k \in \mathbb{N}$, if we are given a circuit $f$ that is $k$-well-behaved, we can construct in time $\textup{poly}(\sz(f))$ a circuit $f'$ that is well-behaved and computes the same function as $f$. This can be achieved by adding $(\sz(f))^k$ dummy gates to the circuit $f$, i.e., gates that do not alter the output of the circuit. For example, we can add gates that repeatedly add $0$ to the output of the circuit.
\end{remark}

\paragraph{\bf Lipschitz-continuity} Note that even well-behaved arithmetic circuits might not yield continuous functions, because of the comparison gate. Some of our problems require continuity of the function, and the most convenient type of continuity for computational purposes is Lipschitz-continuity. A function $f: \mathbb{R}^n \to \mathbb{R}^m$ is Lipschitz-continuous on the domain $D \subseteq \mathbb{R}^n$ with Lipschitz-constant $L$, if for all $x,y \in D$
$$\|f(x) - f(y)\| \leq L \cdot \|x-y\|.$$

\paragraph{\bf Violations and promise-preserving reductions} There is no known way of syntactically enforcing that an arithmetic circuit be Lipschitz-continuous. Thus, to ensure that our problems indeed lie in TFNP, we allow any well-behaved circuit in the input, together with a purported Lipschitz-constant $L$, and also accept a pair $(x,y)$ witnessing a violation of $L$-Lipschitz-continuity as a solution. This ``trick'' was also used by \citet{DaskalakisP2011-CLS} for the definition of \cls/.

One might wonder whether defining a problem in this way, with violations, makes it harder than the (more natural) promise version, where we only consider inputs that satisfy the promise (namely, $L$-Lipschitz-continuity). We show that for our problems, the promise versions are just as hard. Indeed, the hard instances we construct for the \kkt/ problem satisfy the promises and we then obtain this for the other problems ``for free'', because all of our reductions are \emph{promise-preserving}, as defined in \citep[Definition 7]{FGMS2020-UEOPL}. A reduction $(f,g)$ from problem $R$ to problem $S$ is promise-preserving, if for any instance $x$ of $R$, for any violation solution $y$ of instance $f(x)$ of $S$, it holds that $g(x,y)$ is a violation solution of instance $x$ of $R$. Informally: any violation solution of $S$ is mapped back to a violation solution of $R$.

\paragraph{\bf \cls/} The class \cls/ is defined as the set of all TFNP problems that reduce to 3D-\clo/.

\begin{tcolorbox}[breakable,enhanced]
	\begin{definition}\label{def:clo}
		\clo/:
		
		\noindent\textbf{Input}:
		\begin{itemize}
			\item precision/stopping parameter $\varepsilon > 0$,
			\item well-behaved arithmetic circuits $p: [0,1]^n \to [0,1]$ and $g: [0,1]^n \to [0,1]^n$,
			\item Lipschitz constant $L > 0$.
		\end{itemize}
		
		\noindent\textbf{Goal}: Compute an approximate local optimum of $p$ with respect to $g$. Formally, find $x \in [0,1]^n$ such that
		$$p(g(x)) \geq p(x) - \varepsilon.$$
		Alternatively, we also accept one of the following violations as a solution:
		\begin{itemize}
			\item ($p$ is not $L$-Lipschitz) $x,y \in [0,1]^n$ such that $|p(x) - p(y)| > L \|x-y\|$,
			\item ($g$ is not $L$-Lipschitz) $x,y \in [0,1]^n$ such that $\|g(x) - g(y)\| > L \|x-y\|$.
		\end{itemize}
	\end{definition}
\end{tcolorbox}

For $k \in \mathbb{N}$, we let $k$D-\clo/ denote the problem \clo/ where $n$ is fixed to be equal to $k$.

\clo/ is similar to \textsc{Localopt}, in the sense that we are looking for a minimum of $p$ over the domain $[0,1]^n$ using the help of a function $g$. The membership of the problem in \pls/ and in \ppad/ is easy to show \citep{DaskalakisP2011-CLS}. The membership in \ppad/ follows from the observation that $g$ is a Brouwer function and that every (approximate) fixed point of $g$ also yields a solution to the \clo/ instance.

Note that the original definition of \clo/ in \citep{DaskalakisP2011-CLS} uses arithmetic circuits without the ``well-behaved'' restriction. As argued above, these circuits cannot always be evaluated efficiently, and so we instead use well-behaved arithmetic circuits, to ensure that the problem lies in \tfnp/. The interesting problems shown to lie in \cls/ by \citet{DaskalakisP2011-CLS} still reduce to \clo/ even with this restriction on the circuits. It also turns out that this restriction does not make the class any weaker, since we show that 2D-\clo/ with well-behaved arithmetic circuits is \ppadpls/-hard.

For some applications it is convenient to allow more general domains than just $[0,1]^n$ and so we also define a more general version of \clo/.

\begin{tcolorbox}[breakable,enhanced]
	\begin{definition}\label{def:gclo}
		\gclo/:
		
		\noindent\textbf{Input}:
		\begin{itemize}
			\item precision/stopping parameter $\varepsilon > 0$,
			\item $(A,b) \in \mathbb{R}^{m \times n} \times \mathbb{R}^m$ defining a bounded non-empty domain $D = \{x \in \mathbb{R}^n: Ax \leq b\}$,
			\item well-behaved arithmetic circuits $p: \mathbb{R}^n \to \mathbb{R}$ and $g: \mathbb{R}^n \to \mathbb{R}^n$,
			\item Lipschitz constant $L > 0$.
		\end{itemize}
		
		\noindent\textbf{Goal}: Compute an approximate local optimum of $p$ with respect to $g$ on domain $D$. Formally, find $x \in D$ such that
		$$p\bigl(\Pi_D(g(x))\bigr) \geq p(x) - \varepsilon.$$
		Alternatively, we also accept one of the following violations as a solution:
		\begin{itemize}
			\item ($p$ is not $L$-Lipschitz) $x,y \in D$ such that $|p(x) - p(y)| > L \|x-y\|$,
			\item ($g$ is not $L$-Lipschitz) $x,y \in D$ such that $\|g(x) - g(y)\| > L \|x-y\|$.
		\end{itemize}
	\end{definition}
\end{tcolorbox}

\noindent Note that given $(A,b) \in \mathbb{R}^{m \times n} \times \mathbb{R}^m$, it is easy to check whether the domain $D = \{x \in \mathbb{R}^n: Ax \leq b\}$ is bounded and non-empty by using linear programming.

We use the projection $\Pi_D$ in this definition, because it is not clear whether there is some syntactic way of ensuring that $g(x) \in D$. Namely, it is unclear whether $\Pi_D$ can be computed inside our arithmetic circuits. However, $\Pi_D$ can be computed efficiently by a Turing machine, since it can be formulated as a convex quadratic program, known to be solvable in polynomial time \citep{KozlovTK80-convex-quadratic}. To be more precise, given a rational vector $x \in \mathbb{R}^n$, $\Pi_D(x)$ can be computed exactly in time $\textup{poly}(\sz(x), \sz(A), \sz(b))$. Note that when $D=[0,1]^n$, the projection $\Pi_D$ can easily be computed by arithmetic circuits, so $\Pi_D$ is not needed in the definition of \clo/. Indeed, when $D=[0,1]^n$, we have $[\Pi_D(x)]_i = \min \{1, \max\{0,x_i\}\}$ for all $i \in [n]$ and $x \in \mathbb{R}^n$.

The definition of \cls/ using 3D-\clo/, instead of 2D-\clo/, \clo/, or \gclo/, leaves open various questions about whether all these different ways of defining it are equivalent. We prove that this is indeed the case. We discuss this, as well as the robustness of the definition of \cls/ with respect to other modifications in \cref{sec:cls}.

\subsection{Computational problems from nonlinear optimization}\label{sec:nonlinearopt-problems}

In this section we formally define our three problems of interest. We begin by a brief introduction to nonlinear optimization.

\subsubsection{Background on nonlinear optimization}\label{sec:nonlinearopt}

The standard problem of nonlinear optimization (also called nonlinear programming) can be formulated as follows:
\begin{equation}\label{eq:opt-problem}
\begin{tabular}{lll}
& $\min\limits_{x \in \mathbb{R}^n} f(x)$ &\\
& &\\
s.t. & $g_i(x) \leq 0$ & $\forall i \in [m]$\\
\end{tabular}
\end{equation}
where $f: \mathbb{R}^n \to \mathbb{R}$ is the objective function to be minimised, and $g_1, \dots, g_m: \mathbb{R}^n \to \mathbb{R}$ are the inequality constraint functions. It is assumed that $f,g_i$ are $C^1$, i.e., continuously differentiable. Throughout this paper we consider the minimisation problem, but our results also apply to the maximisation problem, since we consider function classes that are closed under negation.

\paragraph{\bf Global minimum} Unfortunately, solving the optimization problem \eqref{eq:opt-problem}, namely computing a \emph{global} minimum, is intractable, even for relatively simple objective functions and constraints (\citep{MurtyK1987} in the context of quadratic programming, \citep{BlumR92} in the context of neural networks).

\paragraph{\bf Local minima} The most natural way to relax the requirement of a global minimum, is to look for a \emph{local} minimum instead. A point $x \in \mathbb{R}^n$ is a local minimum of \eqref{eq:opt-problem}, if it satisfies all the constraints, namely $x \in D$, where $D= \{y \in \mathbb{R}^n \, | \, g_i(x) \leq 0 \, \forall i \in [m]\}$, and if there exists $\varepsilon > 0$ such that
\begin{equation}\label{eq:local-minimum}
f(x) \leq f(y) \qquad \forall y \in D \cap B_\varepsilon(x)
\end{equation}
where $B_\varepsilon(x) = \{y \in \mathbb{R}^n \, | \, \|y-x\| \leq \varepsilon\}$.

However, while the notion of a local minimum is very natural, an important issue arises when the problem is considered from the computational perspective. Looking at expression \eqref{eq:local-minimum}, it not clear how to efficiently check whether a given point $x$ is a local minimum or not. Indeed, it turns out that deciding whether a given point is a local minimum is co-NP-hard, even for simple objective and constraint functions \citep{MurtyK1987}. Furthermore, it was recently shown that computing a local minimum, even when it is guaranteed to exist, cannot be done in polynomial time unless P = NP \citep{AhmadiZ22-polytope}, even for quadratic functions where the domain is a polytope.

\paragraph{\bf Necessary optimality conditions} In order to avoid this issue, one can instead look for a point satisfying some so-called \emph{necessary optimality conditions}. As the name suggests, these are conditions that must be satisfied for any local minimum, but might also be satisfied for points that are not local minima. Importantly, these conditions can usually be checked in polynomial time. For this reason, algorithms attempting to solve \eqref{eq:opt-problem}, usually try to find a point that satisfies some necessary optimality conditions instead.

\paragraph{\bf KKT points} The most famous and simplest necessary optimality conditions are the Karush-Kuhn-Tucker (KKT) conditions. The KKT conditions are first-order conditions in the sense that they only involve the first derivatives (i.e., the gradients) of the functions in the problem statement. Formally, a point $x \in \mathbb{R}^n$ satisfies the KKT conditions if it is feasible, i.e., $x \in D = \{y \in \mathbb{R}^n \, | \, g_i(x) \leq 0 \, \forall i \in [m]\}$, and if there exist $\mu_1, \dots, \mu_m \geq 0$ such that
$$\nabla f(x) + \sum_{i=1}^m \mu_i \nabla g_i(x) = 0$$
and $\mu_i g_i(x) = 0$ for all $i \in [m]$. This last condition ensures that $\mu_i > 0$ can only occur if $g_i(x) = 0$, i.e., if the $i$th constraint is tight. In particular, if no constraint is tight at $x$, then $x$ is a KKT point if $\nabla f(x) = 0$ (in other words, if it is a stationary point). A point $x$ that satisfies the KKT conditions is also called a KKT point. Note that given access to $\nabla f(x)$, $g_i(x)$ and $\nabla g_i(x)$, one can check in polynomial time whether $x$ is a KKT point, since this reduces to checking the feasibility of a linear program.

Every local minimum of \eqref{eq:opt-problem} must satisfy the KKT conditions, as long as the problem satisfies some so-called regularity conditions or constraint qualifications. In this paper, we restrict our attention to linear constraints (i.e., $g_i(x) = \langle a_i, x \rangle - b_i$). In this case, it is known that every local minimum is indeed a KKT point.

\paragraph{\bf $\boldsymbol{\varepsilon}$-KKT points} In practice, but also when studying the computational complexity in the standard Turing model (because of issues of representation), it is unreasonable to expect to find a point that exactly satisfies the KKT conditions. Instead, one looks for an \emph{approximate} KKT point. Given $\varepsilon \geq 0$, we say that $x \in \mathbb{R}^n$ is an $\varepsilon$-KKT point if $x \in D$ and if there exist $\mu_1, \dots, \mu_m \geq 0$ such that
$$\left\|\nabla f(x) + \sum_{i=1}^m \mu_i \nabla g_i(x)\right\| \leq \varepsilon$$
and $\mu_i g_i(x) = 0$ for all $i \in [m]$.\footnote{This is the usual definition of an $\varepsilon$-KKT point \citep{vavasis1993localmin,DuttaDTA13}. A weaker notion where the complementary slackness condition is relaxed to $\mu_i g_i(x) \geq -\varepsilon$ can also be considered \citep{DuttaDTA13}. Our hardness result holds for this weaker notion as well. Indeed, it is easy to check that an $\varepsilon$-KKT point can be obtained by first computing a ``weak'' $\varepsilon'$-KKT point $x \in [0,1]^2$ (for some $\varepsilon'$ efficiently computable given $\varepsilon$ and the other parameters of the problem), and by then ``rounding'' $x_i$ to $0$ or to $1$ if it is sufficiently close to $0$ or $1$, respectively.} In particular, if no constraint is tight at $x$, then $x$ is an $\varepsilon$-KKT point if $\|\nabla f(x)\| \leq \varepsilon$ (i.e., if $x$ is an $\varepsilon$-stationary point). Since $\|\cdot\|$ denotes the $\ell_2$-norm, we can check whether a point is an $\varepsilon$-KKT point in polynomial time by using a convex quadratic program, which can be solved efficiently \citep{KozlovTK80-convex-quadratic}. If we instead use the $\ell_\infty$-norm or the $\ell_1$-norm in the definition of $\varepsilon$-KKT point, then we can check whether a point is an $\varepsilon$-KKT point in polynomial time by solving a linear program.

Since we focus on the case where $D = \{y \in \mathbb{R}^n \,|\, Ay \leq b\}$, $(A,b) \in \mathbb{R}^{m \times n} \times \mathbb{R}^m$, we can rewrite the KKT conditions as follows. A point $x \in \mathbb{R}^n$ is an $\varepsilon$-KKT point if $x \in D$ and if there exist $\mu_1, \dots, \mu_m \geq 0$ such that
$$\Bigl\|\nabla f(x) + A^T\mu\Bigr\| \leq \varepsilon$$
and $\langle \mu, Ax - b \rangle = 0$. Note that this exactly corresponds to the earlier definition adapted to this case. In particular, the condition ``$\mu_i [Ax - b]_i = 0$ for all $i \in [m]$'' is equivalent to $\langle \mu, Ax - b \rangle = 0$, since $\mu_i \geq 0$ and $[Ax - b]_i \leq 0$ for all $i \in [m]$.

It is known that if there are no constraints, then it is NP-hard to decide whether a KKT point exists~\citep{AhmadiZ22-unconstrained}. This implies that, in general, unless P = NP, there is no polynomial-time algorithm that computes a KKT point of \eqref{eq:opt-problem}. However, this hardness result does not say anything about one very important special case, namely when the feasible region $D$ is a compact set (in particular, when it is a bounded polytope defined by linear constraints). Indeed, in that case, a KKT point is guaranteed to exist---since a local minimum is guaranteed to exist---and easy to verify, and thus finding a KKT point is a total search problem in the class \tfnp/. In particular, this means that, for compact $D$, the problem of computing a KKT point cannot be NP-hard, unless NP = co-NP~\citep{MegiddoP1991-tfnp}. In this paper, we provide strong evidence that the problem remains hard for such bounded domains, and, in fact, even when the feasible region is as simple as $D=[0,1]^2$.

The problem of finding an $\varepsilon$-KKT point has primarily been studied in the ``black box'' model, where we only have oracle access to the function and its gradient, and count the number of oracle calls needed to solve the problem. \citet{vavasis1993localmin} proved that at least $\Omega(\sqrt{L/\varepsilon})$ calls are needed to find an $\varepsilon$-KKT point of a continuously differentiable function $f:[0,1]^2 \to \mathbb{R}$ with $L$-Lipschitz gradient. It was recently shown by \citet{BM2020gradient} that this bound is tight up to a logarithmic factor. For the high-dimensional case, \citet{CarmonDHS20-stationary} showed a tight bound of $\Theta(1/\varepsilon^2)$, when the Lipschitz constant is fixed.

\subsubsection{The \kkt/ problem}

Given the definition of $\varepsilon$-KKT points in the previous section, we can formally define a computational problem where the goal is to compute such a point. Our formalisation of this problem assumes that $f$ and $\nabla f$ are provided in the input as arithmetic circuits. However, it is unclear if, given a circuit $f$, we can efficiently determine whether it corresponds to a continuously differentiable function, and whether the circuit for $\nabla f$ indeed computes its gradient. Thus, one has to either consider the promise version of the problem (where this is guaranteed to hold for the input), or add violation solutions like in the definition of \clo/. In order to ensure that our problem is in TFNP, we formally define it with violation solutions. However, we note that our hardness results also hold for the promise versions.

The type of violation solution that we introduce to ensure that $\nabla f$ is indeed the gradient of $f$ is based on the following version of Taylor's theorem, which is proved in \cref{app:taylor}.

\begin{lemma}[Taylor's theorem]\label{lem:taylor}
	Let $f: \mathbb{R}^n \to \mathbb{R}$ be continuously differentiable and let $D \subseteq \mathbb{R}^n$ be convex. If $\nabla f$ is $L$-Lipschitz-continuous (w.r.t.\ the $\ell_2$-norm) on $D$, then for all $x,y \in D$ we have
	$$\bigl|f(y) - f(x) - \langle \nabla f(x), y-x \rangle \bigr| \leq \frac{L}{2} \|y-x\|^2.$$
\end{lemma}

\noindent We are now ready to formally define our \kkt/ problem.

\begin{tcolorbox}[breakable,enhanced]
	\begin{definition}\label{def:kkt-general}
		\kkt/:
		
		\noindent\textbf{Input}:
		\begin{itemize}
			\item precision parameter $\varepsilon > 0$,
			\item $(A,b) \in \mathbb{R}^{m \times n} \times \mathbb{R}^m$ defining a bounded non-empty domain $D = \{x \in \mathbb{R}^n: Ax \leq b\}$,
			\item well-behaved arithmetic circuits $f: \mathbb{R}^n \to \mathbb{R}$ and $\nabla f: \mathbb{R}^n \to \mathbb{R}^n$,
			\item Lipschitz constant $L > 0$.
		\end{itemize}
		
		\noindent\textbf{Goal}: Compute an $\varepsilon$-KKT point for the minimization problem of $f$ on domain $D$.
		
		\noindent Formally, find $x \in D$ such that there exist $\mu_1, \dots, \mu_m \geq 0$ such that
		$$\Bigl\| \nabla f(x) + A^T\mu \Bigr\| \leq \varepsilon$$
		and $\langle \mu, Ax-b \rangle = 0$.
		
		\noindent Alternatively, we also accept one of the following violations as a solution:
		\begin{itemize}
			\item ($f$ or $\nabla f$ is not $L$-Lipschitz) $x,y \in D$ such that
			$$|f(x) - f(y)| > L \|x-y\| \qquad \text{or} \qquad \|\nabla f(x) - \nabla f(y)\| > L \|x-y\|,$$
			\item ($\nabla f$ is not the gradient of $f$) $x,y \in D$ that contradict Taylor's theorem (\cref{lem:taylor}), i.e.,
			$$\bigl|f(y) - f(x) - \langle \nabla f(x), y-x \rangle \bigr| > \frac{L}{2} \|y-x\|^2.$$
		\end{itemize}
	\end{definition}
\end{tcolorbox}

Note that all conditions on the input of the \kkt/ problem can be checked in polynomial time. In particular, we can use linear programming to check that the domain is bounded and non-empty. With regards to a solution $x \in D$, there is no need to include the values $\mu_1, \dots, \mu_m$ as part of a solution. Indeed, given $x \in D$, we can check in polynomial time whether there exist such $\mu_1, \dots, \mu_m$ by solving the following convex quadratic program:
\begin{equation*}
\begin{tabular}{lll}
$\min\limits_{\mu \in \mathbb{R}^m}$ & $\bigl\|\nabla f(x) + A^T\mu\bigr\|^2$ &\\
& &\\
s.t. & $\langle \mu, Ax-b \rangle=0$&\\
& $\mu \geq 0$&\\
\end{tabular}
\end{equation*}
If the optimal value of this program is strictly larger than $\varepsilon^2$, then $x$ is not an $\varepsilon$-KKT point. Otherwise, it is an $\varepsilon$-KKT point and the optimal $\mu_1, \dots, \mu_m$ certify this. If we use the $\ell_\infty$-norm or the $\ell_1$-norm instead of the $\ell_2$-norm for the definition of $\varepsilon$-KKT points, then we can check whether a point is an $\varepsilon$-KKT point using the same approach (except that we do not take the square of the norm, and we simply obtain a linear program). Whether we use the $\ell_2$-norm, the $\ell_\infty$-norm or the $\ell_1$-norm for the definition of $\varepsilon$-KKT points has no impact on the complexity of the \kkt/ problem defined above. Indeed, is is easy to reduce the various versions to each other.

Note that $\varepsilon$ and $L$ are provided in binary representation in the input. This is important, since our hardness result in \cref{thm:main-kkt-hard} relies on at least one of those two parameters being exponential in the size of the input. If both parameters are provided in unary, then the problem can be solved in polynomial time on the domain $[0,1]^n$ (see \cref{lem:kkt-poly-params}).

\subsubsection{Gradient Descent problems}

In this section we formally define our two versions of the Gradient Descent problem. Since we consider Gradient Descent on bounded domains $D$, we need to ensure that the next iterate indeed lies in $D$. The standard way to handle this is by using so-called \emph{Projected} Gradient Descent, where the next iterate is computed using a standard Gradient Descent step and then projected onto $D$ using $\Pi_D$. Formally,
$$x^{(k+1)} \leftarrow \Pi_D\left(x^{(k)} - \eta \nabla f\bigl(x^{(k)}\bigr)\right)$$
where $\eta > 0$ is the step size. Throughout, we only consider the case where the step size is fixed, i.e., the same in all iterations.
Our first version of the problem considers the stopping criterion: stop if the next iterate improves the objective function value by less than $\varepsilon$.

\begin{tcolorbox}[breakable,enhanced]
	\begin{definition}\label{def:gdls}
		\gdls/:
		
		\noindent\textbf{Input}:
		\begin{itemize}
			\item precision/stopping parameter $\varepsilon > 0$,
			\item step size $\eta > 0$,
			\item $(A,b) \in \mathbb{R}^{m \times n} \times \mathbb{R}^m$ defining a bounded non-empty domain $D = \{x \in \mathbb{R}^n: Ax \leq b\}$,
			\item well-behaved arithmetic circuits $f: \mathbb{R}^n \to \mathbb{R}$ and $\nabla f: \mathbb{R}^n \to \mathbb{R}^n$,
			\item Lipschitz constant $L > 0$.
		\end{itemize}
		
		\noindent\textbf{Goal}: Compute any point where (projected) gradient descent for $f$ on domain $D$ with fixed step size $\eta$ terminates. Formally, find $x \in D$ such that
		$$f\Bigl(\Pi_D\bigl(x - \eta \nabla f(x)\bigr)\Bigr) \geq f(x) - \varepsilon.$$
		Alternatively, we also accept one of the following violations as a solution:
		\begin{itemize}
			\item ($f$ or $\nabla f$ is not $L$-Lipschitz) $x,y \in D$ such that
			$$|f(x) - f(y)| > L \|x-y\| \qquad \text{or} \qquad \|\nabla f(x) - \nabla f(y)\| > L \|x-y\|,$$
			\item ($\nabla f$ is not the gradient of $f$) $x,y \in D$ that contradict Taylor's theorem (\cref{lem:taylor}), i.e.,
			$$\bigl|f(y) - f(x) - \langle \nabla f(x), y-x \rangle \bigr| > \frac{L}{2} \|y-x\|^2.$$
		\end{itemize}
	\end{definition}
\end{tcolorbox}

Our second version of the problem considers the stopping criterion: stop if the next iterate is $\varepsilon$-close to the current iterate.

\begin{tcolorbox}[breakable,enhanced]
	\begin{definition}\label{def:gdfp}
		\gdfp/:
		
		\noindent\textbf{Input}:
		\begin{itemize}
			\item precision/stopping parameter $\varepsilon > 0$,
			\item step size $\eta > 0$,
			\item $(A,b) \in \mathbb{R}^{m \times n} \times \mathbb{R}^m$ defining a bounded non-empty domain $D = \{x \in \mathbb{R}^n: Ax \leq b\}$,
			\item well-behaved arithmetic circuits $f: \mathbb{R}^n \to \mathbb{R}$ and $\nabla f: \mathbb{R}^n \to \mathbb{R}^n$,
			\item Lipschitz constant $L > 0$.
		\end{itemize}
		
		\noindent\textbf{Goal}: Compute any point that is an $\varepsilon$-approximate fixed point of (projected) gradient descent for $f$ on domain $D$ with fixed step size $\eta$. Formally, find $x \in D$ such that
		$$\left\|x - \Pi_D\bigl(x - \eta \nabla f(x)\bigr)\right\| \leq \varepsilon.$$
		Alternatively, we also accept one of the following violations as a solution:
		\begin{itemize}
			\item ($f$ or $\nabla f$ is not $L$-Lipschitz) $x,y \in D$ such that
			$$|f(x) - f(y)| > L \|x-y\| \qquad \text{or} \qquad \|\nabla f(x) - \nabla f(y)\| > L \|x-y\|,$$
			\item ($\nabla f$ is not the gradient of $f$) $x,y \in D$ that contradict Taylor's theorem (\cref{lem:taylor}), i.e.,
			$$\bigl|f(y) - f(x) - \langle \nabla f(x), y-x \rangle \bigr| > \frac{L}{2} \|y-x\|^2.$$
		\end{itemize}
	\end{definition}
\end{tcolorbox}

The comments made about the \kkt/ problem in the previous section also apply to these two problems.
In particular, we show that even the promise versions of the two Gradient Descent problems remain \ppadpls/-hard. In other words, the hard instances we construct have no violations.

\begin{remark}
An interesting question is: what happens if we omit the last violation (namely, the one about Taylor's theorem) from the definitions of these problems? For \gdls/ it turns out that this does not change the complexity of the problem. Indeed, removing the last violation means that the functions $f$ and $\nabla f$ can now be completely unrelated. However, for \gdls/ it is not hard to see that the problem remains in \cls/ (in fact, note that the proof of \cref{prop:gdls2gclo} which reduces \gdls/ to \gclo/ does not use violations to Taylor's theorem). Thus, the problem remains \ppadpls/-complete.

On the other hand, for \gdfp/ and \kkt/ it turns out that omitting the last violation does change the complexity of the problem. Indeed, note that unlike \gdls/, the property that some $x$ must satisfy in order to be a (non-violation) solution only depends on $\nabla f$, and not at all on $f$. As a result, it is easy to reduce from the problem of finding an approximate Brouwer fixed point of a function $g: [0,1]^2 \to [0,1]^2$ to either of these two problems, by letting $f(x) = 0$, $\nabla f(x) = x - g(x)$, and setting the remaining parameters appropriately. It follows that \gdfp/ and \kkt/ without the last violation are \ppad/-hard, and in fact it can be shown that they are \ppad/-complete. Finally, note that it is easy to see that \gdfp/ and \kkt/ remain equivalent if we remove the last violation: one direction is given by the proof of \cref{prop:gdfp2kkt} (which still works without violations to Taylor's theorem), and the other direction can be obtained by using the arguments in step 2 of the proof of \cref{prop:kkt2gdls}.
\end{remark}

\section{KKT is \texorpdfstring{$\boldsymbol{\ppadpls/}$}{PPAD ∩ PLS}-hard}\label{sec:kkt-ppadpls}

In this section, we prove our main technical result.

\begin{theorem}\label{thm:main-kkt-hard}
	\kkt/ is \ppadpls/-hard, even when the domain is fixed to be the unit square $[0,1]^2$. The hardness continues to hold even if one considers the promise-version of the problem, i.e., only instances without violations.
\end{theorem}

\noindent In order to show this we provide a polynomial-time many-one reduction from \textsc{Either-Solution(End-of-Line,Iter)} to \textsc{KKT} on the unit square.

\paragraph{\bf Overview} Consider any instance of \eol/ with $2^n$ vertices and any instance of \iter/ with $2^m$ nodes. We construct a function $f$ for the KKT problem as follows. We first work on the domain $[0,N]^2$ with a grid $G = \{0,1,2, \dots, N\}^2$, where $N = 2^n \cdot 2^{m+4}$. In the conceptually most interesting part of the reduction, we carefully specify the value of the function $f$ and the direction of $-\nabla f$ (\emph{the direction of steepest descent}) at all the points of the grid $G$. Then, in the second part of the reduction, we show how to extend $f$ within every square of the grid, so as to obtain a continuously differentiable function on $[0,N]^2$. Finally, we scale down the domain to $[0,1]^2$. We show that any $\varepsilon$-KKT point of $f$ (for some sufficiently small $\varepsilon$) must yield a solution to the \eol/ instance or a solution to the \iter/ instance.

\subsection{Defining the function on the grid}\label{sec:fngrid}

\paragraph{\bf Overview of the embedding} We divide the domain $[0,N]^2$ into $2^n \times 2^n$ big squares. For any $v_1,v_2 \in [2^n]$, let $B(v_1,v_2)$ denote the big square
$$\left[(v_1-1)\frac{N}{2^n},v_1\frac{N}{2^n}\right] \times \left[(v_2-1)\frac{N}{2^n},v_2\frac{N}{2^n}\right].$$
We use the following interpretation: the vertex $v \in [2^n]$ of the \eol/ instance is embedded at the centre of the big square $B(v,v)$. Thus, the vertices are arranged along the main diagonal of the domain. In particular, the trivial source $1 \in [2^n]$ is located at the centre of the big square that lies in the bottom-left corner of the domain and contains the origin.

We seek to embed the edges of the \eol/ instance in our construction. For every directed edge $(v_1,v_2)$ of the \eol/ instance, we are going to embed a directed path in the grid $G$ that goes from the centre of $B(v_1,v_1)$ to the centre of $B(v_2,v_2)$. The type of path used and the route taken by the path will depend on whether the edge $(v_1,v_2)$ is a ``forward'' edge or a ``backward'' edge. In more detail:
\begin{itemize}
	\item if $v_1 < v_2$ (``forward'' edge), then we will use a so-called \emph{green} path that can only travel to the right and upwards. The path starts at the centre of $B(v_1,v_1)$ and moves to the right until it reaches the centre of $B(v_2,v_1)$. Then, it moves upwards until it reaches its destination: the centre of $B(v_2,v_2)$.
	\item if $v_1 > v_2$ (``backward'' edge), then we will use a so-called \emph{orange} path that can only travel to the left and downwards. The path starts at the centre of $B(v_1,v_1)$ and moves to the left until it reaches the centre of $B(v_2,v_1)$. Then, it moves downwards until it reaches its destination: the centre of $B(v_2,v_2)$.
\end{itemize}
\cref{fig:kkt-big-picture} illustrates the high-level idea of the embedding with an example.

For points of the grid $G$ that are part of the ``environment'', namely that do not lie on a path, the function $f$ will simply be defined by $(x,y) \mapsto x + y$. Thus, if there are no paths at all, the only local minimum of $f$ will be at the origin. However, a green path starts at the origin and this will ensure that there is no minimum there. This green path will correspond to the outgoing edge of the trivial source $1 \in [2^n]$ of the \eol/ instance.

The green paths will be constructed such that if one moves along a green path the value of $f$ decreases, which means that we are improving the objective function value. Furthermore, the value of $f$ at any point on a green path will be below the value of $f$ at any point in the environment. Conversely, the orange paths will be constructed such that if one moves along an orange path the value of $f$ increases, so the objective function value becomes worse. Additionally, the value of $f$ at any point on an orange path will be above the value of $f$ at any point in the environment.

We say that a path starting at $B(v_1,v_1)$ ``starts in the environment'', if there is no path ending at $B(v_1,v_1)$. Similarly, a path ending at $B(v_2,v_2)$ ``ends in the environment'', if there is no path starting at $B(v_2,v_2)$. If any path starts or ends in the environment, the construction ensures that there is a stationary point (and thus a KKT point) there. The only exception is the path corresponding to the outgoing edge of the trivial vertex $1 \in [2^n]$. The start of that path will not create a KKT point. Thus, in the example of \cref{fig:kkt-big-picture}, there will certainly be KKT points in $B(3,3)$, $B(7,7)$ and $B(8,8)$, but not in $B(1,1)$.

Recall that every vertex $v \in [2^n]$ has at most one incoming edge and at most one outgoing edge. Thus, for any vertex $v \neq 1$, one of the following cases occurs:
\begin{itemize}
	\item $v$ is an isolated vertex. In this case, the big square $B(v,v)$ will not contain any path and will fully be in the environment, thus not containing any KKT point. Example: vertex $5$ in \cref{fig:kkt-big-picture}.
	\item $v$ has one outgoing edge and no incoming edge. In this case, the big square $B(v,v)$ will contain the start of a green or orange path. There will be a KKT point at the start of the path, which is fine, since $v$ is a (non-trivial) source of the \eol/ instance. Example: vertex $7$ in \cref{fig:kkt-big-picture}.
	\item $v$ has one incoming edge and no outgoing edge. In this case, the big square $B(v,v)$ will contain the end of a green or orange path. There will be a KKT point at the end of the path, which is again fine, since $v$ is a sink of the \eol/ instance. Example: vertices $3$ and $8$ in \cref{fig:kkt-big-picture}.
	\item $v$ has one outgoing and one incoming edge. In this case, there are two sub-cases:
	\begin{itemize}
		\item If both edges yield paths of the same colour, then we will be able to ``connect'' the two paths at the centre of $B(v,v)$ and avoid introducing a KKT point there. Example: vertex $4$ in \cref{fig:kkt-big-picture}.
		\item If one of the paths is green and the other one is orange, then there will be a local maximum or minimum in $B(v,v)$ (and thus a KKT point). It is not too hard to see that this is in fact unavoidable. Indeed, if the incoming path is green and the outgoing path is orange, then there will necessarily be a local minimum at the end of the green path (Example: vertex $6$ in \cref{fig:kkt-big-picture}). If the incoming path is orange and the outgoing path is green, then there will necessarily be a local maximum at the end of the orange path (Example: vertex $2$ in \cref{fig:kkt-big-picture}). This is where we use the main new ``trick'' of our reduction: we ``hide'' the exact location of the KKT point inside $B(v,v)$ in such a way, that finding it requires solving a \pls/-complete problem, namely the \iter/ instance. This is achieved by introducing a new gadget at the point where the two paths meet. We call this the PLS-Labyrinth gadget.
	\end{itemize}
\end{itemize}
The construction of the green and orange paths is described in detail in \cref{sec:kkt-eol-paths}. The PLS-Labyrinth gadget is described in detail in \cref{sec:kkt-pls-labyrinth}.

\begin{figure}
	\centering
	\includegraphics[width=0.7\textwidth]{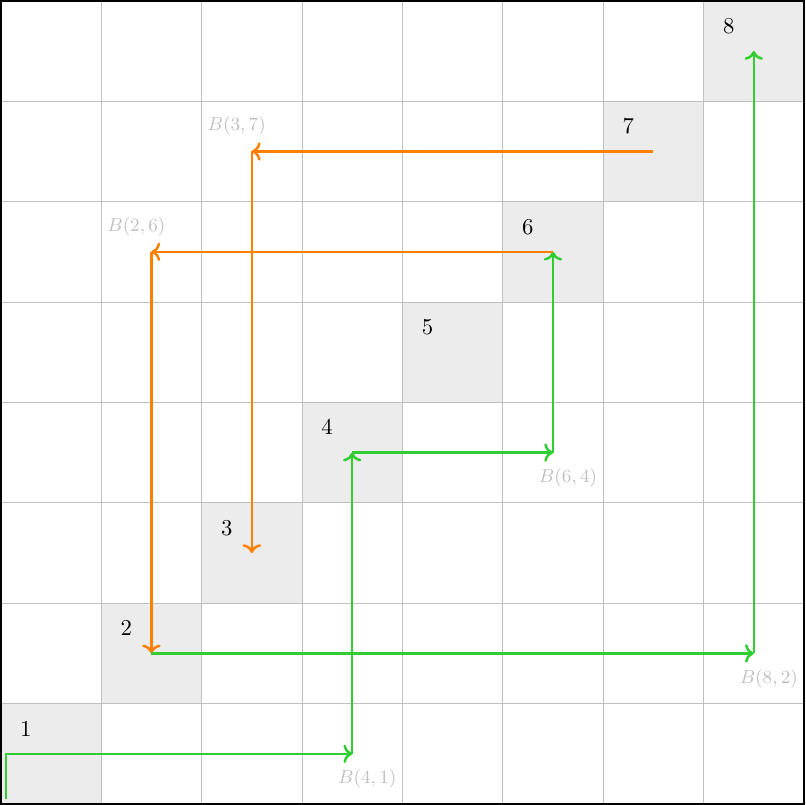}
	\caption{Example of the high-level idea for the embedding of an \eol/ instance in the domain. In this example we are embedding an \eol/ instance with the set of vertices $[8]$ (i.e., $n=3$) and the directed edges: $(1,4)$, $(2,8)$, $(4,6)$, $(6,2)$ and $(7,3)$ (see \cref{fig:EOL-example}). The domain is divided into $8 \times 8$ big squares, and the big squares corresponding to the vertices of the \eol/ graph are coloured in grey. The solutions of this \eol/ instance are the vertices $3$, $7$ and $8$.}
	\label{fig:kkt-big-picture}
\end{figure}

\subsubsection{Pre-processing}

Consider any instance $((S,P),C)$ of \textsc{Either-Solution(End-of-Line,Iter)}, i.e., $S,P: [2^n] \to [2^n]$ is an instance of \eol/ and $C: [2^m] \to [2^m]$ is an instance of \iter/. Without loss of generality, we can assume that these instances satisfy the following:
\begin{enumerate}
	\item The successor and predecessor circuits $S,P$ agree on all edges. Formally, for all $v \in [2^n]$, it holds that
	\begin{itemize}
		\item if $S(v) \neq v$, then $P(S(v)) = v$, and
		\item if $P(v) \neq v$, then $S(P(v)) = v$.
	\end{itemize}
	This property can be ensured by a simple pre-processing step. We modify the circuit $S$, so that before outputting $S(v)$, it first checks whether $(S(v) \neq v) \land (P(S(v)) \neq v)$, and, if this holds, outputs $v$ instead of $S(v)$. It is easy to see that this new circuit for $S$ can be constructed in polynomial time in the size of $S$ and $P$. We also perform the analogous modification for $P$. It is easy to check that this does not introduce any new solutions.
	\item For all $u \in [2^m]$ we have $C(u) \geq u$. We can ensure that this holds by modifying the circuit $C$, so that before outputting $C(u)$, it checks whether $C(u) < u$, and, if this is the case, outputs $u$ instead of $C(u)$. Again, the modification can be done in polynomial time and does not introduce new solutions, nor does it stop the problem from being total.
\end{enumerate}

\subsubsection{The value regimes}\label{sec:valueregimes}

Recall that we want to specify the value of $f$ and $-\nabla f$ (\emph{the direction of steepest descent}) at all points on the grid $G = \{0,1,2, \dots, N\}^2$, where $N = 2^n \cdot 2^{m+4}$. In order to specify the value of $f$, it is convenient to define \emph{value regimes}. Namely, if a point $(x,y) \in G$ is in:
\begin{itemize}
	\item the red value regime, then $f(x,y) := x-y + 4N + 20$.
	\item the orange value regime, then $f(x,y) := -x - y + 4N + 10$.
	\item the black value regime, then $f(x,y) := x + y$.
	\item the green value regime, then $f(x,y) := -x - y - 10$.
	\item the blue value regime, then $f(x,y) := x-y - 2N - 20$.
\end{itemize}
Note that at any point on the grid, the value regimes are ordered: red $>$ orange $>$ black $>$ green $>$ blue. Furthermore, it is easy to check that the gap between any two regimes at any point is at least $10$. \cref{fig:kkt-value-regimes} illustrates the main properties of the value regimes.

\begin{figure}[h]
	\centering
	\scalebox{0.8}{
		\includegraphics{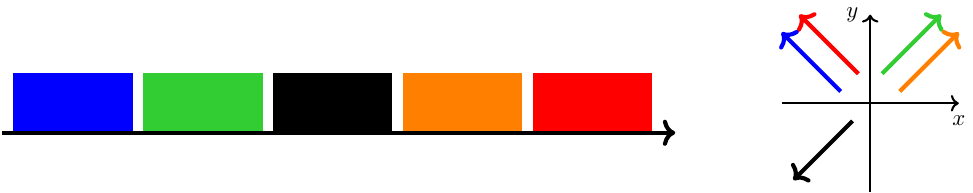}
	}
	\caption{The value regimes. On the left, the colours are ordered according to increasing value, from left to right. On the right, we indicate for each value regime, the direction in which it improves, i.e., decreases, in the $x$-$y$-plane.}
	\label{fig:kkt-value-regimes}
\end{figure}

The black value regime will be used for the environment. Thus, unless stated otherwise, every grid point is coloured in black, i.e., belongs to the black value regime. Furthermore, in our construction, we will set the direction of steepest descent, i.e., $- \nabla f(x,y)$, at every grid point $(x,y)$ to be one of the four possible cardinal directions, i.e., left, right, up, or down. Unless stated otherwise, at every black grid point $(x,y)$, the direction of steepest descent $- \nabla f(x,y)$ will point to the left.\footnote{Notice that this is not the same as the negative gradient of the ``black regime function'' $(x,y)\mapsto x+y$, which would point south-west, as shown in \cref{fig:kkt-value-regimes}. Nevertheless, as we show later, this is enough to ensure that the bicubic interpolation that we use does not introduce any points with zero gradient in a region of the environment. Similarly, for grid points coloured with one of the other colours we will also not use the diagonal negative gradient of the corresponding value regime function, but instead one of the four cardinal directions.} The only exceptions to this (i.e., the only black grid points where the direction of steepest descent does not point left) are black grid points that lie in paths, or black grid points that lie on the left boundary of the domain (i.e., points with $x=0$).

\subsubsection{Embedding the \eol/ instance: The green and orange paths}\label{sec:kkt-eol-paths}

Our construction specifies for each grid point a colour (which represents the value of $f$ at that point) and an arrow that represents the direction of $- \nabla f$ at that point. A general ``rule'' that we follow throughout our construction is that the function values should be consistent with the arrows. For example, if some grid point has an arrow pointing to the right, then the adjacent grid point to the right should have a lower function value, while the adjacent grid point to the left should have a higher function value. This rule is not completely sufficient by itself to avoid KKT points, but it is already a very useful guide.

Recall that the grid $G = \{0,1,2, \dots, N\}^2$ subdivides every big square $B(v_1,v_2)$ into $2^{m+4} \times 2^{m+4}$ small squares. The width of the paths we construct will be two small squares. This corresponds to a width of three grid points.

\paragraph{\bf Green paths}
When a green path moves to the right, the two lower grid points will be coloured in green, and the grid point at the top will be in black. \cref{fig:G1-horizontal-green} shows a big square that is traversed by a green path from left to right. Such a big square is said to be of type G1. The black arrows indicate the direction of $- \nabla f$ at every grid point.

When a green path moves upwards, the two right-most grid points will be coloured in green, and the grid point on the left will be in black. \cref{fig:G2-vertical-green} shows a big square of type G2, namely one that is traversed by a green path from the bottom to the top.

Recall that a green path implementing an edge $(v_1,v_2)$ (where $v_1 < v_2$) comes into the big square $B(v_2,v_1)$ from the left and leaves at the top. Thus, the path has to ``turn''. \cref{fig:G3-turn-up} shows how this turn is implemented. The big square $B(v_2,v_1)$ is said to be of type G3.

If a vertex $v \in [2^n]$ has one incoming edge $(v_1,v)$ and one outgoing edge $(v,v_2)$ such that $v_1 < v < v_2$, then both edges will be implemented by green paths. The green path corresponding to $(v_1,v)$ will enter $B(v,v)$ from the bottom and stop at the centre of $B(v,v)$. The green path corresponding to $(v,v_2)$ will start at the centre of $B(v,v)$ and leave the big square on the right. In order to avoid introducing any KKT points in $B(v,v)$ (since $v$ is not a solution of the \eol/ instance), we will connect the two paths at the centre of $B(v,v)$. This will be achieved by a simple turn, as shown in \cref{fig:G4-turn-right}. The big square $B(v,v)$ is said to be of type G4.

If a vertex $v \in [2^n] \setminus \{1\}$ has one outgoing edge $(v,v_2)$ such that $v < v_2$, and no incoming edge, then this will yield a green path starting at the centre of $B(v,v)$ and going to the right, as shown in \cref{fig:G5-green-source}. The big square $B(v,v)$ is said to be of type G5 in that case. It is not hard to see that there will be a KKT point at the source of that green path. On the other hand, if a vertex $v \in [2^n] \setminus \{1\}$ has one incoming edge $(v_1,v)$ such that $v_1 < v$, and no outgoing edge, then this will yield a green path coming from the bottom and ending at the centre of $B(v,v)$, as shown in \cref{fig:G6-green-sink}. The big square $B(v,v)$ is said to be of type G6 in that case. Again, there will be a KKT point at the sink of that green path.

\begin{figure}
	\begin{subfigure}[t]{0.47\textwidth}
		\centering
		\scalebox{0.7}{
			\includegraphics{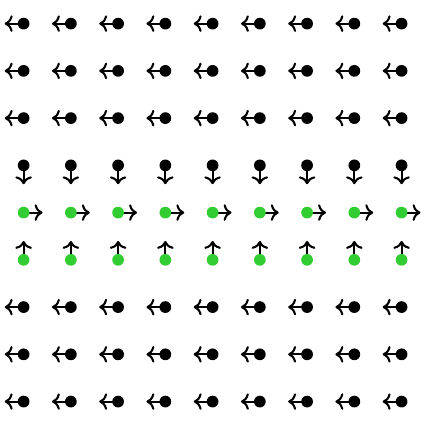}
		}
		\caption{[G1] Green path traversing big square from left to right.}\label{fig:G1-horizontal-green}
	\end{subfigure}\hfill
	\begin{subfigure}[t]{0.47\textwidth}
		\centering
		\scalebox{0.7}{
			\includegraphics{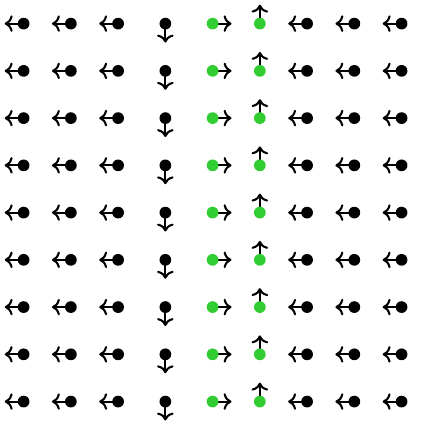}
		}
		\caption{[G2] Green path traversing big square from bottom to top.}\label{fig:G2-vertical-green}
	\end{subfigure}
	
	\bigskip
	
	\begin{subfigure}[t]{0.47\textwidth}
		\centering
		\scalebox{0.7}{
			\includegraphics{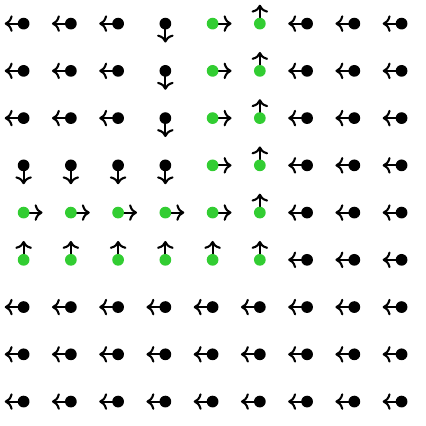}
		}
		\caption{[G3] Green path entering big square from the left, turning, and leaving at the top.}\label{fig:G3-turn-up}
	\end{subfigure}\hfill
	\begin{subfigure}[t]{0.47\textwidth}
		\centering
		\scalebox{0.7}{
			\includegraphics{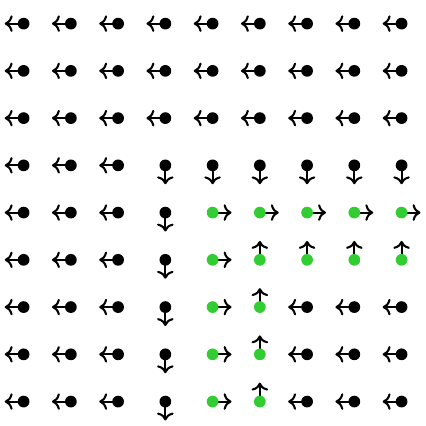}
		}
		\caption{[G4] Green path entering big square from the bottom, turning, and leaving on the right.}\label{fig:G4-turn-right}
	\end{subfigure}
	
	\bigskip
	
	\begin{subfigure}[t]{0.47\textwidth}
		\centering
		\scalebox{0.7}{
			\includegraphics{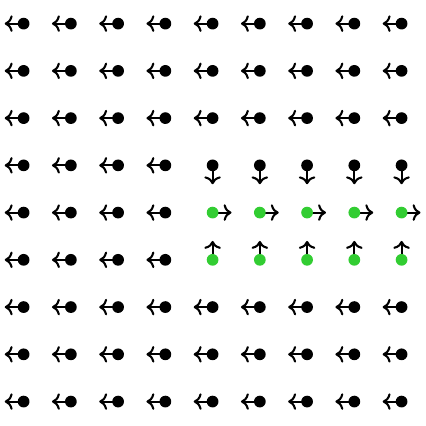}
		}
		\caption{[G5] Source: green path starting at the centre of big square and leaving on the right.}\label{fig:G5-green-source}
	\end{subfigure}\hfill
	\begin{subfigure}[t]{0.47\textwidth}
		\centering
		\scalebox{0.7}{
			\includegraphics{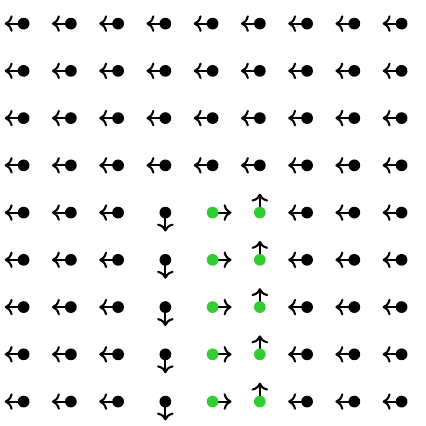}
		}
		\caption{[G6] Sink: green path entering big square from the bottom and ending at the centre.}\label{fig:G6-green-sink}
	\end{subfigure}
	\caption{Construction of the green paths. The figures show various types of big squares containing different portions of green paths. In these illustrations, the big squares are assumed to have size $8 \times 8$ instead of $2^{m+4} \times 2^{m+4}$.}\label{fig:greenpaths}
\end{figure}

\paragraph{\bf Orange paths}
The structure of orange paths is, in a certain sense, symmetric to the structure of green paths. When an orange path moves to the left, the two upper grid points will be coloured in orange, and the grid point at the bottom will be in black. \cref{fig:O1-horizontal-orange} shows a big square that is traversed by an orange path from right to left. Such a big square is said to be of type O1.

When an orange path moves downwards, the two left-most grid points will be coloured in orange, and the grid point on the right will be in black. \cref{fig:O2-vertical-orange} shows a big square of type O2, namely one that is traversed by an orange path from top to bottom. Note that the arrows on an orange path essentially point in the \emph{opposite} direction compared to the direction of the path. This is because we want the value to increase (i.e., worsen) when we follow an orange path.

An orange path implementing an edge $(v_1,v_2)$ (where $v_1 > v_2$) comes into the big square $B(v_2,v_1)$ from the right and leaves at the bottom. This turn is implemented as shown in \cref{fig:O3-turn-down}. The big square $B(v_2,v_1)$ is said to be of type O3.

If a vertex $v \in [2^n]$ has one incoming edge $(v_1,v)$ and one outgoing edge $(v,v_2)$ such that $v_1 > v > v_2$, then both edges will be implemented by orange paths. The orange path corresponding to $(v_1,v)$ will enter $B(v,v)$ from the top and stop at the centre of $B(v,v)$. The orange path corresponding to $(v,v_2)$ will start at the centre of $B(v,v)$ and leave the big square on the left. As above, we avoid introducing a KKT point by connecting the two paths at the centre of $B(v,v)$. This is achieved by the turn shown in \cref{fig:O4-turn-left}. The big square $B(v,v)$ is said to be of type O4.

If a vertex $v \in [2^n] \setminus \{1\}$ has one outgoing edge $(v,v_2)$ such that $v > v_2$, and no incoming edge, then this will yield an orange path starting at the centre of $B(v,v)$ and going to the left, as shown in \cref{fig:O5-orange-source}. The big square $B(v,v)$ is said to be of type O5 in that case. It is not hard to see that there will be a KKT point at the source of that orange path. On the other hand, if a vertex $v \in [2^n] \setminus \{1\}$ has one incoming edge $(v_1,v)$ such that $v_1 > v$, and no outgoing edge, then this will yield an orange path coming from the top and ending at the centre of $B(v,v)$, as shown in \cref{fig:O6-orange-sink}. The big square $B(v,v)$ is said to be of type O6 in that case. Again, there will be a KKT point at the sink of that orange path.

\begin{figure}
	\begin{subfigure}[t]{0.47\textwidth}
		\centering
		\scalebox{0.7}{
			\includegraphics{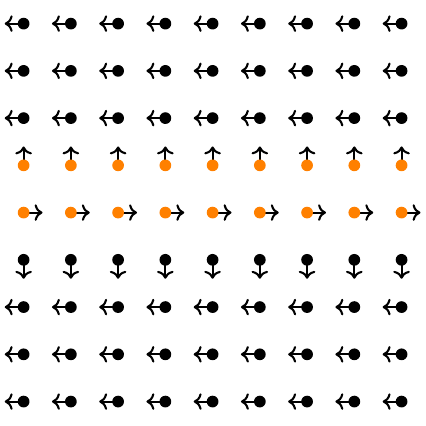}
		}
		\caption{[O1] Orange path traversing big square from right to left.}\label{fig:O1-horizontal-orange}
	\end{subfigure}\hfill
	\begin{subfigure}[t]{0.47\textwidth}
		\centering
		\scalebox{0.7}{
			\includegraphics{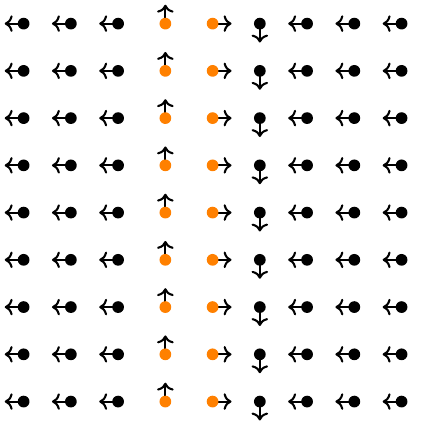}
		}
		\caption{[O2] Orange path traversing big square from top to bottom.}\label{fig:O2-vertical-orange}
	\end{subfigure}
	
	\bigskip
	
	\begin{subfigure}[t]{0.47\textwidth}
		\centering
		\scalebox{0.7}{
			\includegraphics{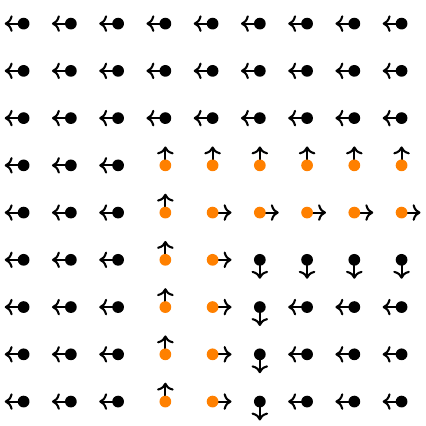}
		}
		\caption{[O3] Orange path entering big square from the right, turning, and leaving at the bottom.}\label{fig:O3-turn-down}
	\end{subfigure}\hfill
	\begin{subfigure}[t]{0.47\textwidth}
		\centering
		\scalebox{0.7}{
			\includegraphics{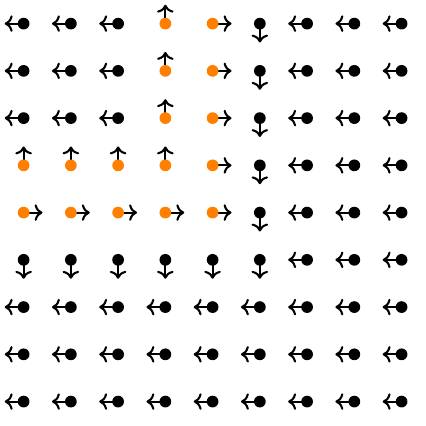}
		}
		\caption{[O4] Orange path entering big square from the top, turning, and leaving on the left.}\label{fig:O4-turn-left}
	\end{subfigure}
	
	\bigskip
	
	\begin{subfigure}[t]{0.47\textwidth}
		\centering
		\scalebox{0.7}{
			\includegraphics{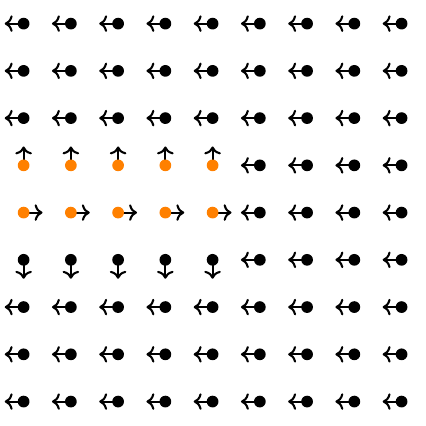}
		}
		\caption{[O5] Source: orange path starting at the centre of big square and leaving on the left.}\label{fig:O5-orange-source}
	\end{subfigure}\hfill
	\begin{subfigure}[t]{0.47\textwidth}
		\centering
		\scalebox{0.7}{
			\includegraphics{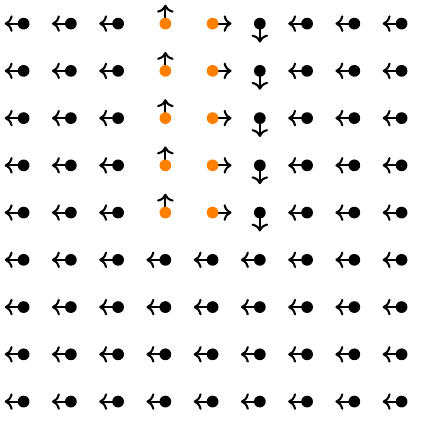}
		}
		\caption{[O6] Sink: orange path entering big square from the top and ending at the centre.}\label{fig:O6-orange-sink}
	\end{subfigure}
	
	\caption{Construction of the orange paths. The figures show various types of big squares containing different portions of orange paths. In these illustrations, the big squares are assumed to have size $8 \times 8$ instead of $2^{m+4} \times 2^{m+4}$.}
\end{figure}

\paragraph{\bf Crossings}
Note that, by construction, green paths only exist below the diagonal, and orange paths only exist above the diagonal. Thus, there is no point where an orange path crosses a green path. However, there might exist points where green paths cross, or orange paths cross. First of all, note that it is impossible to have more than two paths traversing a big square, and thus any crossing involves exactly two paths. Furthermore, no crossing can occur in big squares where a ``turn'' occurs, since, in that case, the turn connects the two paths.

The only way for two green paths to cross is the case where a green path traverses a big square from left to right, and a second green path traverses the same big square from bottom to top. In that case, we say that the big square is of type G7. This problem always occurs when one tries to embed an \eol/ instance in a two-dimensional domain. \citet{ChenD2009-2D} proposed a simple, yet ingenious, trick to resolve this issue. The idea is to locally re-route the two paths so that they no longer cross. This modification has the following two crucial properties: a) it is completely local, and b) it does not introduce any new solution (in our case a KKT point). \cref{fig:G7-green-crossing} shows how this modification is implemented for crossing green paths, i.e., what our construction does for big squares of type G7.

The same issue might arise for orange paths. By the same arguments as above, this can only happen when an orange path traverses a big square from right to left, and a second orange path traverses the same big square from top to bottom. In that case, we say that the big square is of type O7. \cref{fig:O7-orange-crossing} shows how the issue is locally resolved in that case, i.e., what our construction does for big squares of type O7.

\begin{figure}
	\begin{subfigure}{\textwidth}
		\centering
		\scalebox{0.62}{
			\includegraphics{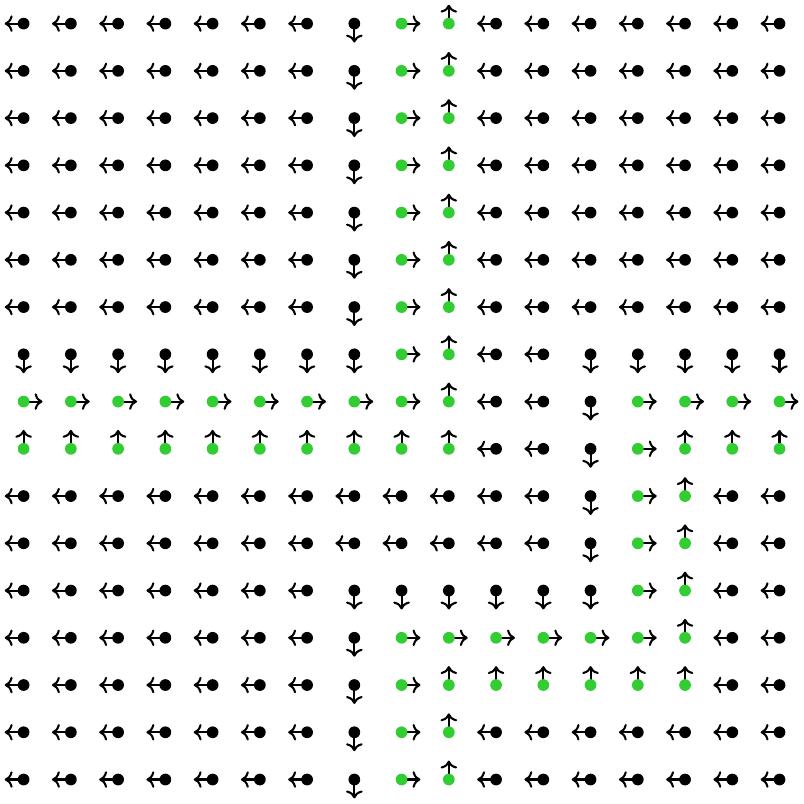}
		}
		\caption{[G7] Crossing of green paths.}\label{fig:G7-green-crossing}
	\end{subfigure}
	
	\bigskip
	
	\begin{subfigure}{\textwidth}
		\centering
		\scalebox{0.62}{
			\includegraphics{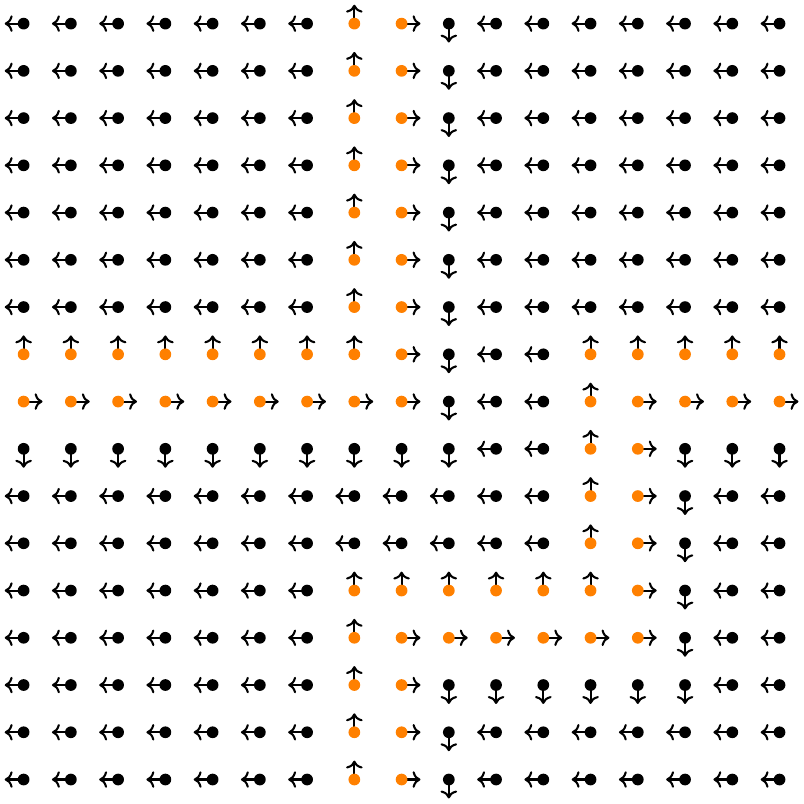}
		}
		\caption{[O7] Crossing of orange paths.}\label{fig:O7-orange-crossing}
	\end{subfigure}
	
	\caption{Crossing gadgets for green and orange paths. In these two illustrations, the big squares are assumed to have size $16 \times 16$ instead of $2^{m+4} \times 2^{m+4}$.}
\end{figure}

\paragraph{\bf Boundary and origin squares}
Any big square that is not traversed by any path (including all big squares $B(v,v)$ where $v$ is an isolated vertex of the \eol/ instance), will have all its grid points coloured in black, and $- \nabla f$ pointing to the left. These big squares, which are said to be of type E1, are as represented in \cref{fig:E1-standard-environment}. The only exceptions to this rule are the big squares $B(1,v)$ for all $v \in [2^n] \setminus \{1\}$. In those big squares, which are said to be of type E2, the grid points on the left boundary have $-\nabla f$ pointing downwards, instead of to the left. The rest of the grid points have $-\nabla f$ pointing to the left as before. Note that none of these big squares is ever traversed by a path, so they are always as shown in \cref{fig:E2-left-boundary-environment}.

\begin{figure}
	\begin{subfigure}[t]{0.47\textwidth}
		\centering
		\scalebox{0.7}{
			\includegraphics{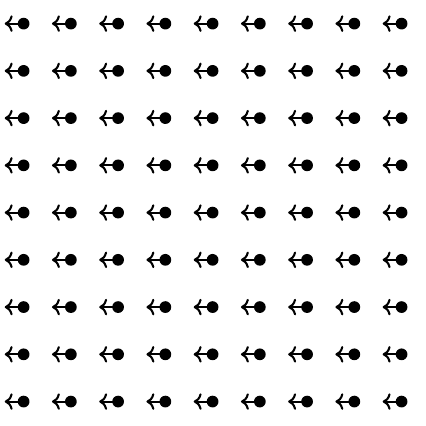}
		}
		\caption{[E1] Big square not traversed by any path.}\label{fig:E1-standard-environment}
	\end{subfigure}\hfill
	\begin{subfigure}[t]{0.47\textwidth}
		\centering
		\scalebox{0.7}{
			\includegraphics{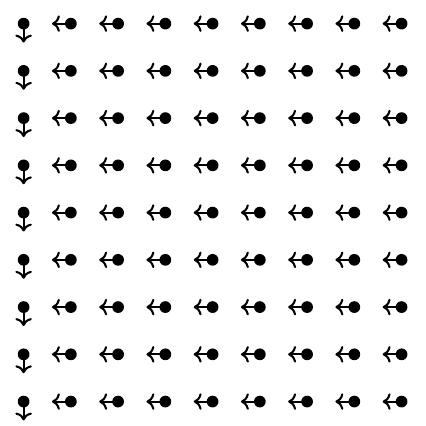}
		}
		\caption{[E2] Big square on left boundary of domain.}\label{fig:E2-left-boundary-environment}
	\end{subfigure}
	\caption{Big squares not traversed by any path. In these two illustrations, the big squares are assumed to have size $8 \times 8$ instead of $2^{m+4} \times 2^{m+4}$.}
\end{figure}

The big square $B(1,1)$ is special and we say that it is of type S. Since it corresponds to the trivial source of the \eol/ instance, it has one outgoing edge (which necessarily corresponds to a green path) and no incoming edge. Normally, this would induce a KKT point at the centre of $B(1,1)$ (as in \cref{fig:G5-green-source}). Furthermore, recall that, by the definition of the black value regime, there must also be a KKT point at the origin, if it is coloured in black. By a careful construction (which is very similar to the one used by \citet{HubacekY2017-CLS} for \clo/) we can ensure that these two KKT points neutralise each other. In other words, instead of two KKT points, there is no KKT point at all in $B(1,1)$. The construction for $B(1,1)$ is shown in \cref{fig:S-origin-big-square}.

\begin{figure}
	\centering
	\scalebox{0.7}{
		\includegraphics{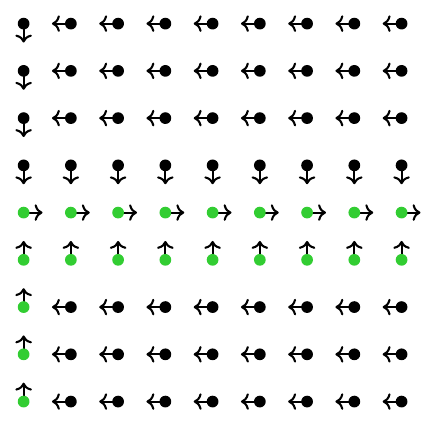}
	}
	\caption{[S] Construction for big square $B(1,1)$ (for size $8 \times 8$ instead of $2^{m+4} \times 2^{m+4}$).}\label{fig:S-origin-big-square}
\end{figure}

\smallskip

\cref{fig:X-full-boundary} shows the whole construction for a small example where $n=1$ and big squares have size $8 \times 8$ (instead of $2^{m+4} \times 2^{m+4}$).

\begin{figure}
	\centering
	\scalebox{0.7}{
		\includegraphics{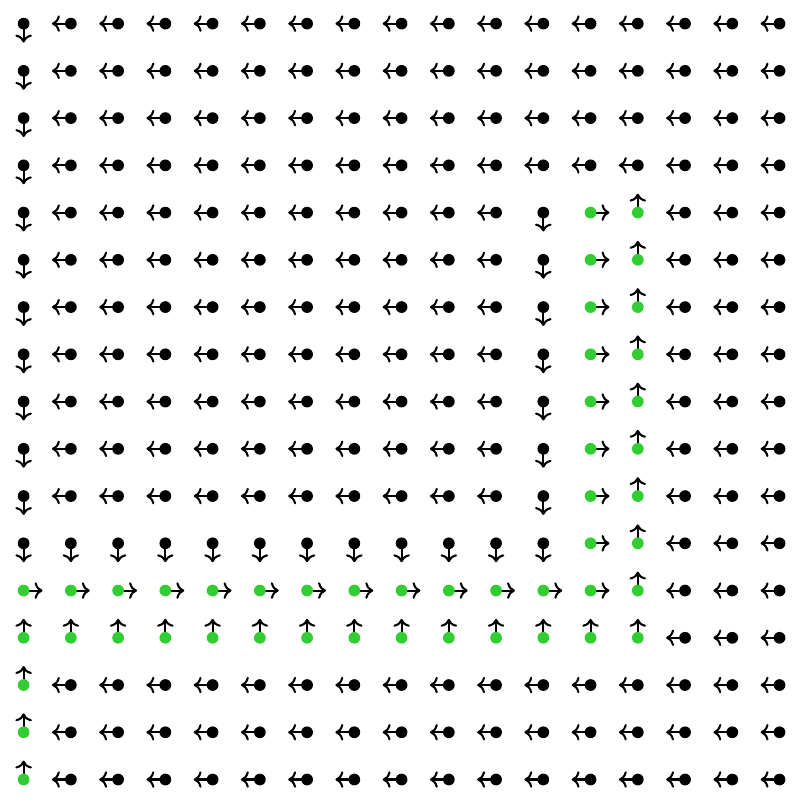}
	}
	\caption{[X] Full construction for a small example, in particular showing the whole boundary. Here $n=1$ and big squares have size $8 \times 8$ (instead of $2^{m+4} \times 2^{m+4}$).}\label{fig:X-full-boundary}
\end{figure}

\paragraph{\bf Green and orange paths meeting}
Our description of the construction is almost complete, but there is one crucial piece missing. Indeed, consider any vertex $v$ that has one incoming edge $(v_1,v)$ and one outgoing edge $(v,v_2)$ such that: A) $v_1 < v$ and $v_2 < v$, or B) $v_1 > v$ and $v_2 > v$. As it stands, a green path and an orange path meet at the centre of $B(v,v)$ which means that there is a local minimum or maximum at the centre of $B(v,v)$, and thus a KKT point. However, $v$ is not a solution to the \eol/ instance. Even though we cannot avoid having a KKT point in $B(v,v)$, we can ``hide'' it, so that finding it requires solving the \iter/ instance. This is implemented by constructing a PLS-Labyrinth gadget at the point where the green and orange paths meet. \cref{fig:LA-PLS-position,fig:LB-PLS-position} show where this PLS-Labyrinth gadget is positioned inside a big square of type LA (namely when case A above occurs) and a big square of type LB (namely when case B above occurs) respectively. The PLS-Labyrinth gadget can only be positioned at a point where a green path and an orange path meet. In particular, it cannot be used to ``hide'' a KKT point occurring at a source or sink of a green or orange path, i.e., at a solution of the \eol/ instance.

\begin{figure}
	\begin{subfigure}[t]{0.47\textwidth}
		\centering
		\includegraphics{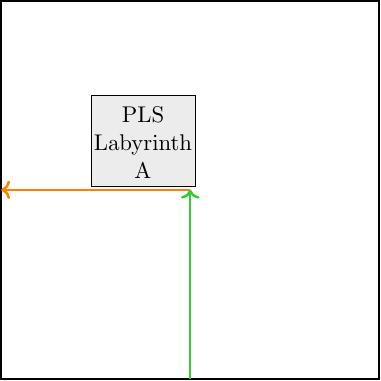}
		\caption{[LA] Position of PLS-Labyrinth gadget in big square of type LA.}\label{fig:LA-PLS-position}
	\end{subfigure}\hfill
	\begin{subfigure}[t]{0.47\textwidth}
		\centering
		\includegraphics{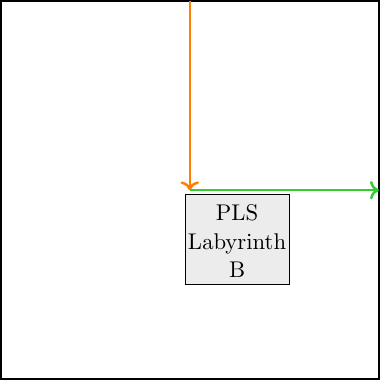}
		\caption{[LB] Position of PLS-Labyrinth gadget in big square of type LB.}\label{fig:LB-PLS-position}
	\end{subfigure}
	
	\caption{Position of PLS-Labyrinth gadget in big squares of type LA and LB.}
\end{figure}

\smallskip

In our construction, every big square is of type G1, G2, $\dots$, G7, O1, O2, $\dots$, O7, E1, E2, S, LA or LB. Note that we can efficiently determine the type of a given big square, if we have access to the \eol/ circuits $S$ and $P$.

\subsubsection{Embedding the \iter/ instance: The \pls/-Labyrinth}\label{sec:kkt-pls-labyrinth}

\paragraph{\bf PLS-Labyrinth}
We begin by describing the PLS-Labyrinth gadget for case A, i.e., $v$ has one incoming edge $(v_1,v)$ and one outgoing edge $(v,v_2)$ such that $v_1 < v$ and $v_2 < v$. In particular, $B(v,v)$ is of type LA. The PLS-Labyrinth gadget comprises $2^{m+2} \times 2^{m+2}$ small squares and is positioned in the big square $B(v,v)$ as shown in \cref{fig:LA-PLS-position}. Note, in particular, that the bottom side of the gadget is adjacent to the orange path, and the bottom-right corner of the gadget lies just above the point where the green and orange paths intersect (which occurs at the centre of $B(v,v)$). Finally, observe that since $B(v,v)$ has $2^{m+4} \times 2^{m+4}$ small squares, there is enough space for the PLS-Labyrinth gadget.

For convenience, we subdivide the PLS-Labyrinth gadget into $2^{m} \times 2^{m}$ medium squares. Thus, every medium square is made out of $4 \times 4$ small squares. We index the medium squares as follows: for $u_1,u_2 \in [2^m]$, let $M(u_1,u_2)$ denote the medium square that is the $u_2$th from the bottom and the $u_1$th from the \emph{right}. Thus, $M(1,1)$ corresponds to the medium square that lies at the bottom-right of the gadget (and is just above the intersection of the paths). Our construction will create the following paths inside the PLS-Labyrinth gadget:
\begin{itemize}
	\item For every $u \in [2^m]$ such that $C(u) > u$, there is an orange-blue path (namely, a path formed by both orange and blue points) starting at $M(u,1)$ and moving upwards until it reaches $M(u,u)$.
	\item For every $u \in [2^m]$ such that $C(u) > u$ and $C(C(u)) > C(u)$, there is a blue path starting at $M(u,u)$ and moving to the left until it reaches $M(C(u),u)$.
\end{itemize}
\cref{fig:labyrinth-A-highlevel} shows a high-level overview of how the \iter/ instance is embedded in the PLS-Labyrinth.
Note that if $C(u) > u$ and $C(C(u)) > C(u)$, then the blue path starting at $M(u,u)$ will move to the left until $M(C(u),u)$ where it will reach the orange-blue path moving up from $M(C(u),1)$ to $M(C(u),C(u))$ (which exists since $C(C(u)) > C(u)$). Thus, every blue path will always ``merge'' into some orange-blue path. On the other hand, some orange-blue paths will stop in the environment without merging into any other path. Consider any $u \in [2^m]$ such that $C(u) > u$. The orange-blue path for $u$ stops at $M(u,u)$. If $C(C(u)) > C(u)$, then there is a blue path starting there, so the orange-blue path ``merges'' into the blue path. However, if $C(C(u)) \leq C(u)$, i.e., $C(C(u))=C(u)$, there is no blue path starting at $M(u,u)$ and the orange-blue path just stops in the environment. Thus, the only place in the PLS-Labyrinth where a path can stop in the environment is in a medium square $M(u,u)$ such that $C(u)>u$ and $C(C(u))=C(u)$. This corresponds exactly to the solutions of the \iter/ instance $C$. In our construction, we will ensure that KKT points can indeed only occur at points where a path stops without merging into any other path.

\begin{figure}
	\centering
	\scalebox{0.8}{
		\includegraphics{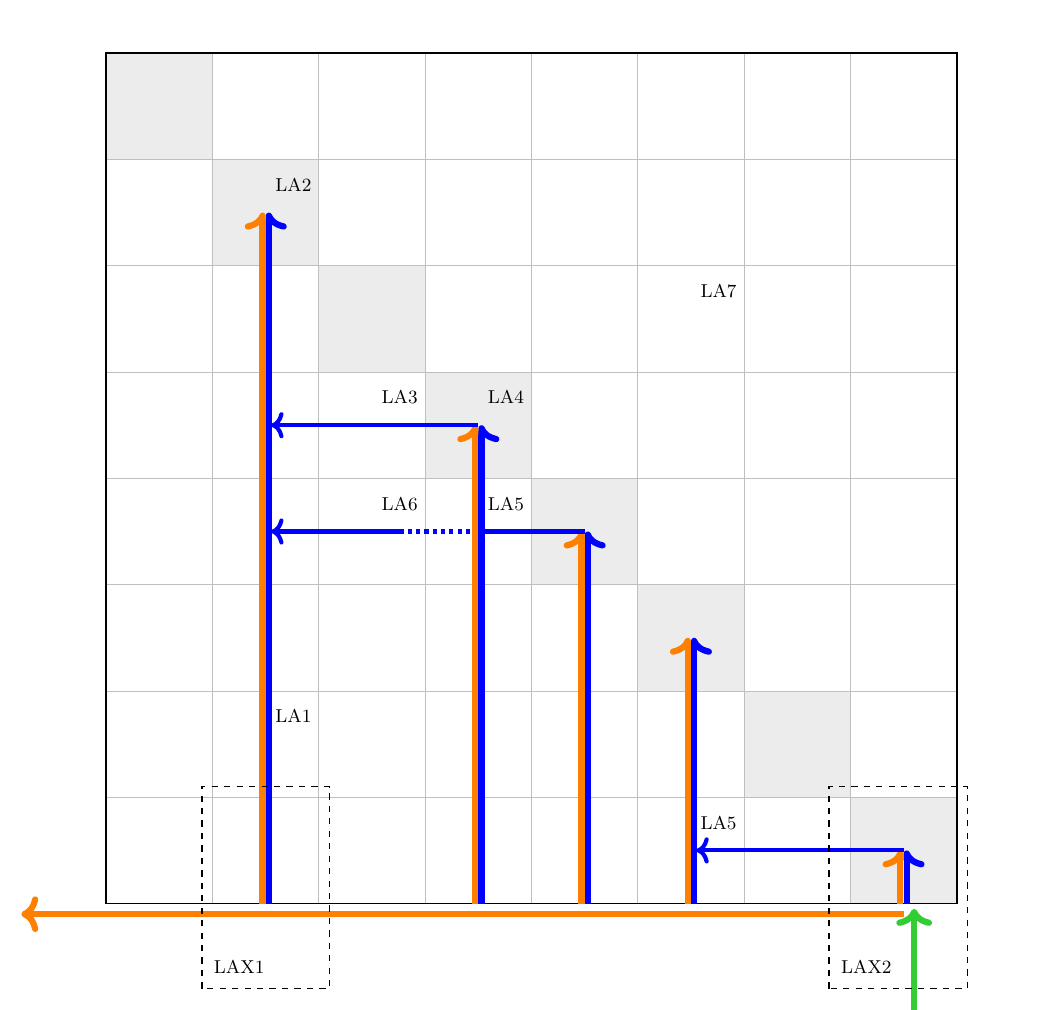}
	}
	\caption{Map of the \pls/-Labyrinth for case A corresponding to the \iter/ example of \cref{fig:ITER-example}. Shaded squares are the medium squares corresponding to the nodes of \iter/. The horizontal blue lines (pointing left) correspond to the 3 edges in \cref{fig:ITER-example} that go out from non-solutions, and we do not use similar lines going out from solutions (nodes 3 and 7). We have also indicated the parts LA1-LA6, and LAX1-LAX2, that are constructed in \cref{fig:LA-all-figures}.}\label{fig:labyrinth-A-highlevel}
\end{figure}

\paragraph{\bf Orange-blue paths} 
An orange-blue path moves from $M(u,1)$ upwards to $M(u,u)$ (for some $u \in [2^m]$ such that $C(u) > u$) and has a width of two small squares. The left-most point is coloured in orange and the two points on the right are blue. \cref{fig:LA1-orange-blue-path} shows a medium square that is being traversed by the orange-blue path, i.e., a medium square $M(u,w)$ where $w < u$. We say that such a medium square $M(u,w)$ is of type LA1. When the orange-blue path reaches $M(u,u)$, it either ``turns'' to the left and creates the beginning of a blue path (medium square of type LA4, \cref{fig:LA4-orange-blue-turn}), or it just stops there (medium square of type LA2, \cref{fig:LA2-orange-blue-sink}). The case where the orange-blue path just stops, occurs when there is no blue path starting at $M(u,u)$. Note that, in that case, $u$ is a solution of the \iter/ instance, and so it is acceptable for a medium square of type LA2 to contain a KKT point.

The orange-blue path begins in $M(u,1)$, which lies just above the orange path. In fact, the beginning of the orange-blue path is adjacent to the orange path as shown in \cref{fig:LAX1-orange-blue-start}. This is needed, since if the orange-blue path started in the environment, the point coloured orange would yield a local maximum and thus a KKT point.

The beginning of the orange-blue path for $u=1$ is special, since, in a certain sense, this path is created by the intersection of the green and orange paths. \cref{fig:LAX2-labyrinth-A-origin} shows how the intersection is implemented and how exactly it is adjacent to $M(1,1)$. Note that $M(1,1)$ is just a standard ``turn'', i.e., a medium square of type LA4.

\begin{figure}
	\begin{subfigure}[t]{0.3\textwidth}
		\centering
		\scalebox{0.8}{
			\includegraphics{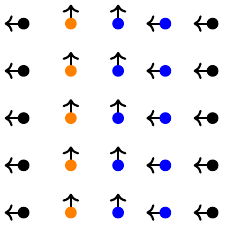}
		}
		\caption{[LA1] Orange-blue path traversing medium square.}\label{fig:LA1-orange-blue-path}
	\end{subfigure}\hfill
	\begin{subfigure}[t]{0.3\textwidth}
		\centering
		\scalebox{0.8}{
			\includegraphics{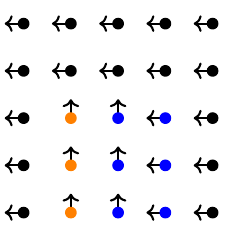}
		}
		\caption{[LA2] Orange-blue path ending in medium square.}\label{fig:LA2-orange-blue-sink}
	\end{subfigure}\hfill
	\begin{subfigure}[t]{0.3\textwidth}
		\centering
		\scalebox{0.8}{
			\includegraphics{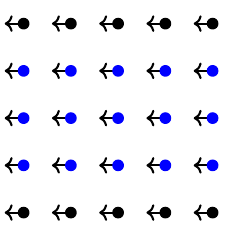}
		}
		\caption{[LA3] Blue path traversing medium square.}\label{fig:LA3-blue-path}
	\end{subfigure}
	
	\bigskip
	
	\begin{subfigure}[t]{0.3\textwidth}
		\centering
		\scalebox{0.8}{
			\includegraphics{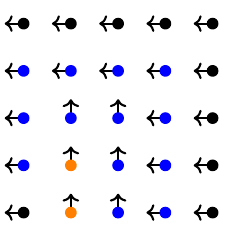}
		}
		\caption{[LA4] Orange-blue path turning in medium square and creating start of blue path.}\label{fig:LA4-orange-blue-turn}
	\end{subfigure}\hfill
	\begin{subfigure}[t]{0.65\textwidth}
		\centering
		\scalebox{0.8}{
			\includegraphics{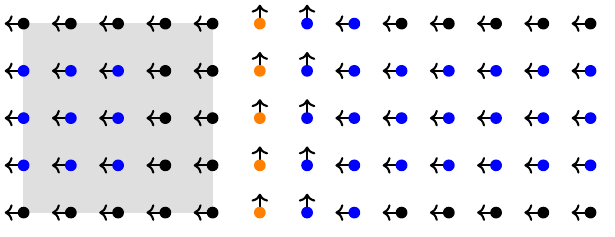}
		}
		\caption{[LA6] Blue path crossing over orange-blue path. Medium square of type LA6 indicated in grey.}\label{fig:LA6-crossing-over}
	\end{subfigure}
	
	\bigskip
	
	\begin{subfigure}[t]{0.42\textwidth}
		\centering
		\scalebox{0.8}{
			\includegraphics{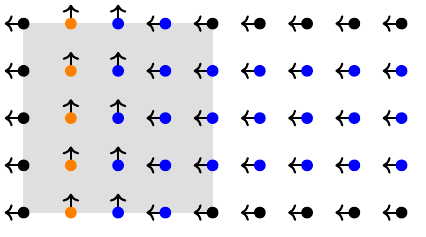}
		}
		\caption{[LA5] Blue path merging into orange-blue path. Medium square of type LA5 indicated in grey.}\label{fig:LA5-merge}
	\end{subfigure}\hfill
	\begin{subfigure}[t]{0.25\textwidth}
		\centering
		\scalebox{0.8}{
			\includegraphics{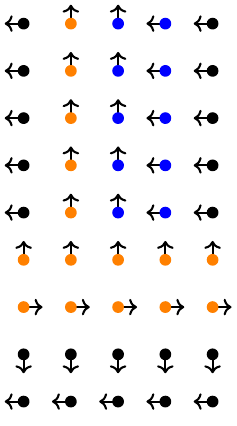}
		}
		\caption{[LAX1] Start of orange-blue path, adjacent to orange path.}\label{fig:LAX1-orange-blue-start}
	\end{subfigure}\hfill
	\begin{subfigure}[t]{0.27\textwidth}
		\centering
		\scalebox{0.8}{
			\includegraphics{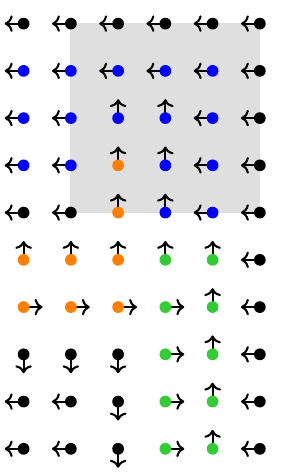}
		}
		\caption{[LAX2] Start of orange-blue path for $u=1$, adjacent to the intersection of green and orange path. $M(1,1)$ is indicated in grey.}\label{fig:LAX2-labyrinth-A-origin}
	\end{subfigure}
	
	\caption{Construction of blue and orange-blue paths in the PLS-Labyrinth gadget inside a big square of type LA.}\label{fig:LA-all-figures}
\end{figure}

\paragraph{\bf Blue paths} A blue path starts in $M(u,u)$ for some $u \in [2^m]$ such that $C(u) > u$ and $C(C(u)) > C(u)$. It moves from right to left and has a width of two small squares. All three points on the path are coloured blue and the direction of steepest descent points to the left. \cref{fig:LA3-blue-path} shows a medium square traversed by a blue path. Such a medium square is said to be of type LA3. As mentioned above, the blue path starts at $M(u,u)$ which is of type LA4 (a ``turn''). When the blue path reaches $M(C(u),u)$, it merges into the orange-blue path going from $M(C(u),1)$ to $M(C(u),C(u))$. This merging is straightforward and is implemented as shown in \cref{fig:LA5-merge}. The medium square $M(C(u),u)$ is then said to be of type LA5.

\paragraph{\bf Crossings} Note that two orange-blue paths cannot cross, and similarly two blue paths can also not cross. However, a blue path going from $M(u,u)$ to $M(C(u),u)$ can cross many other orange-blue paths, before it reaches and merges into its intended orange-blue path. Fortunately, these crossings are much easier to resolve than earlier. Indeed, when a blue path is supposed to cross an orange-blue path, it can simply merge into it and restart on the other side. The important thing to note here is that, while a blue path cannot stop in the environment (without creating a KKT point), it can \emph{start} in the environment. \cref{fig:LA6-crossing-over} shows how this is implemented. In particular, we use a medium square of type LA5 for the merging, and a medium square of type LA6 for the re-start of the blue path.

Note that if the blue path has to cross more than one orange-blue path in immediate succession, then it will simply merge into the first one it meets, and restart after the last one (i.e., as soon as it reaches a medium square that is not traversed by an orange-blue path).

\smallskip

Finally, we say that a medium square is of type LA7, if it does not contain any path at all. Medium squares of type LA7 are like the environment, i.e., all the grid points are coloured black and the arrows of steepest descent point to the left. In our construction, every medium square in the PLS-Labyrinth gadget is of type LA1, LA2, $\dots$, LA6, or LA7. It is easy to check that the type of a given medium square can be determined efficiently, given access to the \iter/ circuit $C$.

\smallskip

The PLS-Labyrinth gadget for case B is, in a certain sense, symmetric to the one presented above. Indeed, it suffices to perform a point reflection (in other words, a rotation by $180$ degrees) with respect to the centre of $B(v,v)$, and a very simple transformation of the colours. With regards to the final interpolated function, this corresponds to rotating $B(v,v)$ by $180$ degrees around its centre and multiplying the output of the function by $-1$. Let $\phi: B(v,v) \to B(v,v)$ denote rotation by $180$ degrees around the centre of $B(v,v)$. Then, the direction of steepest descent at some grid point $(x,y) \in B(v,v)$ in case B is simply the same as the direction of steepest descent at $\phi(x,y)$ in case A. The colour of $(x,y)$ in case B is obtained from the colour of $\phi(x,y)$ in case A as follows:
\begin{itemize}
	\item black remains black,
	\item green becomes orange, and vice-versa,
	\item blue becomes red, and vice-versa.
\end{itemize}
\cref{fig:labyrinth-B-highlevel} shows a high-level overview of the PLS-Labyrinth gadget for case B.
We obtain corresponding medium squares of type LB1, LB2, $\dots$, LB7. The analogous illustrations for case B are shown in \cref{fig:LB-all-figures}.

\begin{figure}
	\centering
	\scalebox{0.8}{
		\includegraphics{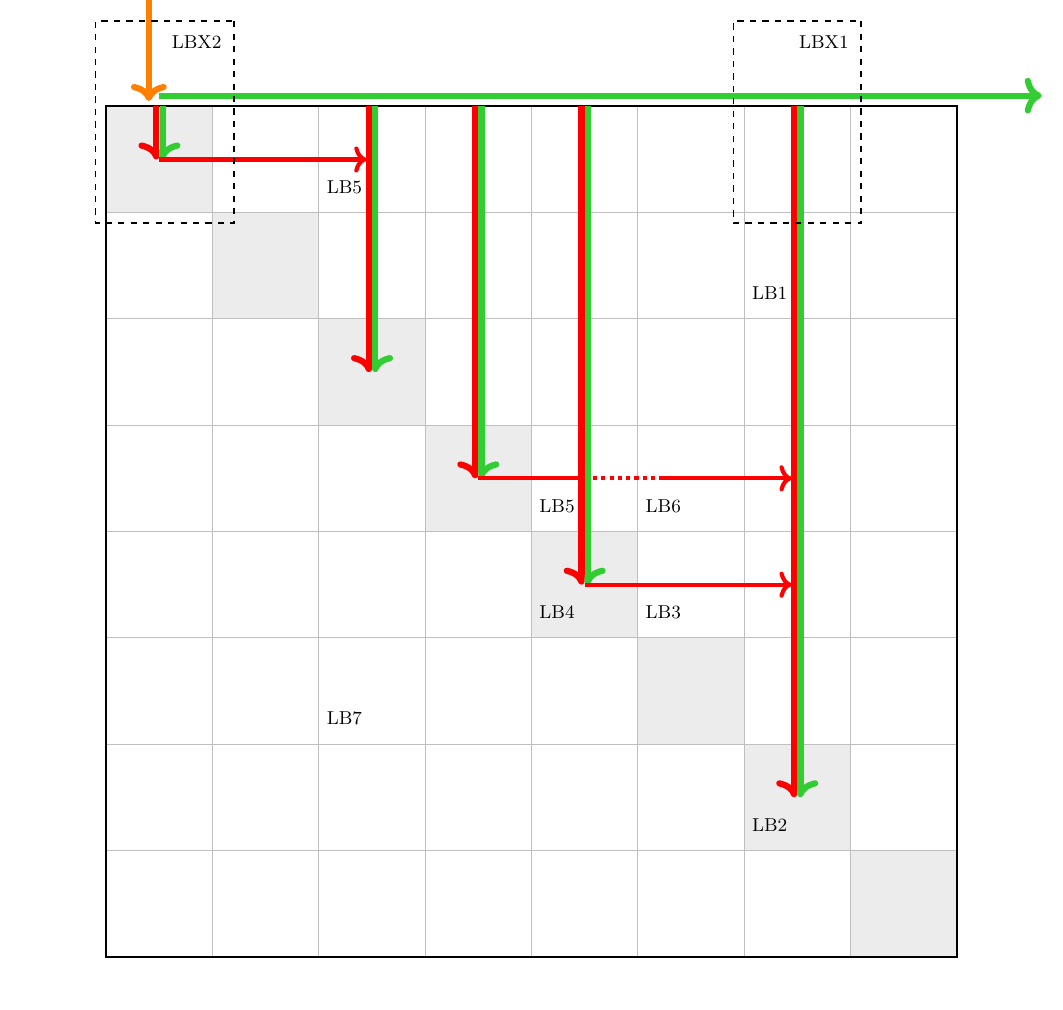}
	}
	\caption{Map of the \pls/-Labyrinth for case B corresponding to the \iter/ example of \cref{fig:ITER-example}. Shaded squares are the medium squares corresponding to the nodes of \iter/. We have also indicated the parts LB1-LB6, and LBX1-LBX2, that are constructed in \cref{fig:LB-all-figures}.}\label{fig:labyrinth-B-highlevel}
\end{figure}

\begin{figure}
	\begin{subfigure}[t]{0.3\textwidth}
		\centering
		\scalebox{0.8}{
			\includegraphics{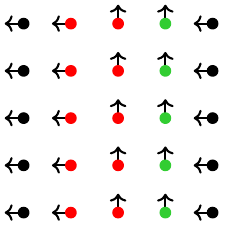}
		}
		\caption{[LB1] Red-green path traversing medium square.}\label{fig:LB1-red-green-path}
	\end{subfigure}\hfill
	\begin{subfigure}[t]{0.3\textwidth}
		\centering
		\scalebox{0.8}{
			\includegraphics{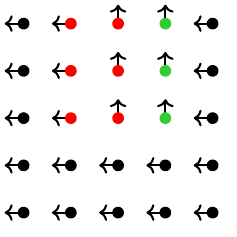}
		}
		\caption{[LB2] Red-green path ending in medium square.}\label{fig:LB2-red-green-sink}
	\end{subfigure}\hfill
	\begin{subfigure}[t]{0.3\textwidth}
		\centering
		\scalebox{0.8}{
			\includegraphics{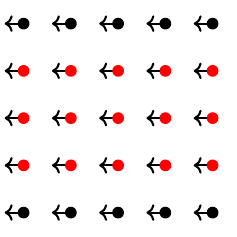}
		}
		\caption{[LB3] Red path traversing medium square.}\label{fig:LB3-red-path}
	\end{subfigure}
	
	\bigskip
	
	\begin{subfigure}[t]{0.3\textwidth}
		\centering
		\scalebox{0.8}{
			\includegraphics{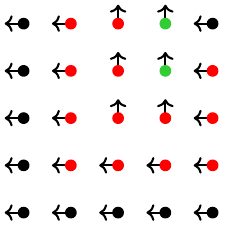}
		}
		\caption{[LB4] Red-green path turning in medium square and creating start of red path.}\label{fig:LB4-red-green-turn}
	\end{subfigure}\hfill
	\begin{subfigure}[t]{0.65\textwidth}
		\centering
		\scalebox{0.8}{
			\includegraphics{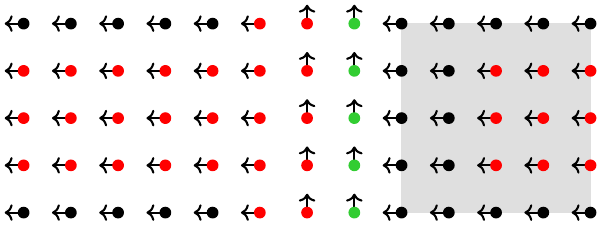}
		}
		\caption{[LB6] Red path crossing over red-green path. Medium square of type LB6 indicated in grey.}\label{fig:LB6-crossing-over}
	\end{subfigure}
	
	\bigskip
	
	\begin{subfigure}[t]{0.42\textwidth}
		\centering
		\scalebox{0.8}{
			\includegraphics{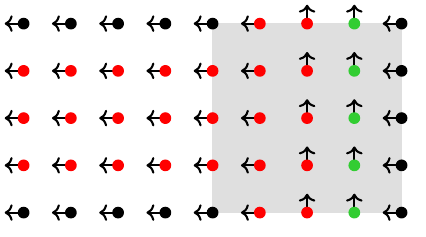}
		}
		\caption{[LB5] Red path merging into red-green path. Medium square of type LB5 indicated in grey.}\label{fig:LB5-merge}
	\end{subfigure}\hfill
	\begin{subfigure}[t]{0.25\textwidth}
		\centering
		\scalebox{0.8}{
			\includegraphics{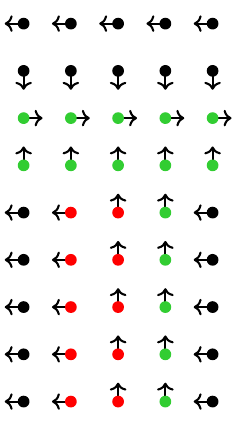}
		}
		\caption{[LBX1] Start of red-green path, adjacent to green path.}\label{fig:LBX1-red-green-start}
	\end{subfigure}\hfill
	\begin{subfigure}[t]{0.27\textwidth}
		\centering
		\scalebox{0.8}{
			\includegraphics{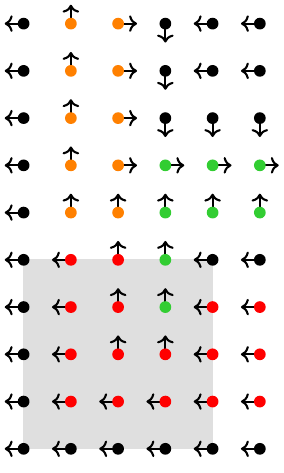}
		}
		\caption{[LBX2] Start of red-green path for $u=1$, adjacent to the intersection of orange and green path. $M(1,1)$ is indicated in grey.}\label{fig:LBX2-labyrinth-B-origin}
	\end{subfigure}
	
	\caption{Construction of red and red-green paths in the PLS-Labyrinth gadget inside a big square of type LB.}\label{fig:LB-all-figures}
\end{figure}

\subsection{Extending the function to the rest of the domain}

Up to this point we have defined the function $f$ and the direction of its gradient at all grid points of $G$. In order to extend $f$ to the whole domain $[0,N]^2$, we use \emph{bicubic interpolation} (see e.g.~\citep{Russell1995-interpolation} or the corresponding Wikipedia article\footnote{\url{https://en.wikipedia.org/wiki/Bicubic_interpolation}}). Note that the more standard and simpler \emph{bilinear} interpolation (used in particular by \citet{HubacekY2017-CLS}) yields a continuous function, but not necessarily a continuously differentiable function. On the other hand, bicubic interpolation ensures that the function will indeed be continuously differentiable over the whole domain $[0,N]^2$.

We use bicubic interpolation in every small square of the grid $G$. Consider any small square and let $(x,y) \in [0,1]^2$ denote the local coordinates of a point inside the square. Then, the bicubic interpolation inside this square will be a polynomial of the form:
\begin{equation}\label{eq:bicubic}
f(x,y) = \sum_{i=0}^3 \sum_{j=0}^3 a_{ij} x^i y^j
\end{equation}
where the coefficients $a_{ij}$ are computed as follows
\begin{align}\label{eq:bicubic-coefficients}
&\begin{bmatrix}
a_{00} & a_{01} & a_{02} & a_{03}\\
a_{10} & a_{11} & a_{12} & a_{13}\\
a_{20} & a_{21} & a_{22} & a_{23}\\
a_{30} & a_{31} & a_{32} & a_{33}\\
\end{bmatrix}\\
&=
\begin{bmatrix}
1 & 0 & 0 & 0\\
0 & 0 & 1 & 0\\
-3 & 3 & -2 & -1\\
2 & -2 & 1 & 1\\
\end{bmatrix} \cdot 
\begin{bmatrix}
f(0,0) & f(0,1) & f_y(0,0) & f_y(0,1)\\
f(1,0) & f(1,1) & f_y(1,0) & f_y(1,1)\\
f_x(0,0) & f_x(0,1) & f_{xy}(0,0) & f_{xy}(0,1)\\
f_x(1,0) & f_x(1,1) & f_{xy}(1,0) & f_{xy}(1,1)\\
\end{bmatrix}
\cdot
\begin{bmatrix}
1 & 0 & -3 & 2\\
0 & 0 & 3 & -2\\
0 & 1 & -2 & 1\\
0 & 0 & -1 & 1\\
\end{bmatrix}\nonumber
\end{align}
Here $f_x$ and $f_y$ denote the partial derivatives with respect to $x$ and $y$ respectively. Similarly, $f_{xy}$ denotes the second order partial derivative with respect to $x$ and $y$. It remains to explain how we set the values of $f, f_x, f_y$ and $f_{xy}$ at the four corners of the square:
\begin{itemize}
	\item The values $f(0,0), f(0,1), f(1,0)$ and $f(1,1)$ are set according to the value regimes in our construction.
	\item The values of $f_x(0,0), f_x(0,1), f_x(1,0), f_x(1,1), f_y(0,0), f_y(0,1), f_y(1,0)$ and $f_y(1,1)$ are set based on the direction of steepest descent ($- \nabla f$) in our construction, with a length multiplier of $\delta = 1/2$. For example, if the arrow of steepest descent at $(0,1)$ is pointing to the left, then we set $f_x(0,1)=\delta$ and $f_y(0,1)=0$. If it is pointing up, then we set $f_x(0,1)=0$ and $f_y(0,1)=-\delta$.
	\item We always set $f_{xy}(0,0) = f_{xy}(0,1) = f_{xy}(1,0) = f_{xy}(1,1) = 0$.
\end{itemize}
By using this interpolation procedure in each small square, we obtain a function $f: [0,N]^2 \to \mathbb{R}$. In fact, we can even extend the function to points $(x,y) \in \mathbb{R}^2 \setminus [0,N]^2$ by simply using the interpolated polynomial obtained for the small square that is closest to $(x,y)$. This will be done automatically by our construction of the arithmetic circuit computing $f$ and it will ensure that the gradient is well-defined even on the boundary of $[0,N]^2$.

\begin{lemma}\label{lem:kkt-function-properties}
	The function $f: \mathbb{R}^2 \to \mathbb{R}$ we obtain by bicubic interpolation has the following properties:
	\begin{itemize}
		\item It is continuously differentiable on $\mathbb{R}^2$;
		\item $f$ and its gradient $\nabla f$ are Lipschitz-continuous on $[0,N]^2$ with Lipschitz-constant $L = 2^{18} N$;
		\item Well-behaved arithmetic circuits computing $f$ and $\nabla f$ can be constructed in polynomial time (in the size of the circuits $S,P$ and $C$).
	\end{itemize}
\end{lemma}

\begin{proof}
	Regarding the first point, see, e.g., \citep{Russell1995-interpolation}.
	
	\paragraph{\bf Lipschitz-continuity} In order to prove the second point, we first show that $f$ and $\nabla f$ are $L$-Lipschitz-continuous in every small square of the grid. Consider any small square. In our construction, the values of $f,f_x,f_y,f_{xy}$ used in the computation of the coefficients $a_{ij}$ are clearly all upper bounded by $2^3 N$ in absolute value. Thus, using \cref{eq:bicubic-coefficients}, it is easy to check that $|a_{ij}| \leq 2^{10} N$ for all $i,j \in \{0,1,2,3\}$. Furthermore, note that the partial derivatives of $f$ inside the small square can be expressed as:
	\begin{equation}\label{eq:bicubic-gradient}
	\frac{\partial f}{\partial x}(x,y) = \sum_{i=1}^3 \sum_{j=0}^3 i \cdot a_{ij} x^{i-1} y^j
	\qquad
	\frac{\partial f}{\partial y}(x,y) = \sum_{i=0}^3 \sum_{j=1}^3 j \cdot a_{ij} x^i y^{j-1}
	\end{equation}
	using the local coordinates $(x,y) \in [0,1]^2$ inside the small square. Finally, it is easy to check that the monomials $x^i y^j$, $i,j \in \{0,1,2,3\}$, are all $6$-Lipschitz continuous over $[0,1]^2$. Putting everything together and using \cref{eq:bicubic} and \cref{eq:bicubic-gradient}, it follows that $f$ and $\nabla f$ are Lipschitz-continuous (w.r.t.\ the $\ell_2$-norm) with Lipschitz constant $L = 2^{18} N$ inside the small square. Note that the change from local coordinates to standard coordinates is just a very simple translation that does not impact the Lipschitzness of the functions.
	
	Since $f$ and $\nabla f$ are $L$-Lipschitz-continuous inside every small square and continuous over all of $[0,N]^2$, it follows that they are in fact $L$-Lipschitz-continuous over the whole domain $[0,N]^2$. Indeed, consider any points $z_1,z_2 \in [0,N]^2$. Then, there exists $\ell \in \mathbb{N}$ such that the segment $z_1z_2$ can be subdivided into $z_1w_1w_2 \dots w_\ell z_2$ so that each of the segments $z_1w_1$, $w_1w_2, \dots,$ $w_{\ell-1}w_\ell$, $w_\ell z_2$ lies within a small square. For ease of notation, we let $w_0:=z_1$ and $w_{\ell+1}:=z_2$. Then, we can write
	$$\|\nabla f(z_1) - \nabla f(z_2)\| \leq \sum_{i=0}^\ell \|\nabla f(w_i) - \nabla f(w_{i+1})\| \leq L \sum_{i=0}^\ell \|w_i-w_{i+1}\| = L \|z_1-z_2\|$$
	where we used the fact that $\sum_{i=0}^\ell \|w_i-w_{i+1}\| = \|z_1-z_2\|$, since $w_0w_1w_2 \dots w_\ell w_{\ell+1}$ is just a partition of the segment $z_1z_2$. The exact same argument also works for $f$.
	
	\paragraph{\bf Arithmetic circuits} Before showing how to construct the arithmetic circuits for $f$ and $\nabla f$, we first construct a Boolean circuit $B$ that will be used as a sub-routine. The Boolean circuit $B$ receives as input a point $(x,y)$ on the grid $G=\{0,1,\dots,N\}^2$ and outputs the colour (i.e., value regime) and steepest descent arrow at that point. It is not too hard to see that the circuit $B$ can be constructed in time that is polynomial in the sizes of the circuits $S,P$ and $C$. In more detail, it performs the following operations:
	\begin{enumerate}
		\item Compute \eol/-vertices $v_1,v_2 \in [2^n]$ such that $(x,y)$ lies in the big square $B(v_1,v_2)$.
		\item Using the \eol/ circuits $S$ and $P$ determine the exact type of $B(v_1,v_2)$, namely one of the following: G1-G7, O1-O7, E1, E2, S, LA or LB.
		\item If the type of $B(v_1,v_2)$ is not LA or LB, then we know the exact structure of $B(v_1,v_2)$ and can easily return the colour and arrow at $(x,y)$.
		\item If the type of $B(v_1,v_2)$ is LA or LB, then first determine whether $(x,y)$ lies in the PLS-Labyrinth inside $B(v_1,v_2)$ or not.
		\item If $(x,y)$ does not lie in the PLS-Labyrinth, then we can easily determine the colour and arrow at $(x,y)$, since we know the exact structure of $B(v_1,v_2)$ except the inside of the PLS-Labyrinth.
		\item If $(x,y)$ lies in the PLS-Labyrinth, then we can compute \iter/-vertices $u_1,u_2 \in [2^m]$ such that $(x,y)$ lies in the medium square $M(u_1,u_2)$ of the PLS-Labyrinth inside $B(v_1,v_2)$.
		\item Using the \iter/ circuit $C$ determine the type of $M(u_1,u_2)$, namely one of the following: LA1-LA7, LB1-LB7. Given the type of $M(u_1,u_2)$, we then know the exact structure of $M(u_1,u_2)$ and can in particular determine the colour and arrow at $(x,y)$.
	\end{enumerate}
	The arithmetic circuits for $f$ and $\nabla f$ are then constructed to perform the following operations on input $(x,y) \in [0,N]^2$:
	\begin{enumerate}
		\item Using the comparison gate $<$ and binary search, compute the bits representing $(\widehat{x},\widehat{y}) \in \{0,1,\dots,N-1\}^2$: a grid point such that $(x,y)$ lies in the small square that has $(\widehat{x},\widehat{y})$ as its bottom left corner.
		\item Simulate the Boolean circuit $B$ using arithmetic gates to compute (a bit representation) of the colour and arrow at the four corners of the small square, namely $(\widehat{x},\widehat{y})$, $(\widehat{x}+1,\widehat{y})$,$(\widehat{x},\widehat{y}+1)$ and $(\widehat{x}+1,\widehat{y}+1)$.
		\item Using this information and the formulas for the value regimes, compute the $16$ terms for $f,f_x,f_y$ and $f_{xy}$ needed to determine the bicubic interpolation. Then, compute the coefficients $a_{ij}$ by performing the matrix multiplication in \cref{eq:bicubic-coefficients}.
		\item In the arithmetic circuit for $f$, apply \cref{eq:bicubic} to compute the value of $f(x,y)$. In the arithmetic circuit for $\nabla f$, apply \cref{eq:bicubic-gradient} to compute the value of $\nabla f(x,y)$. Note that in the interpolation \cref{eq:bicubic,eq:bicubic-gradient}, we have to use the local coordinates $(x-\widehat{x},y-\widehat{y}) \in [0,1]^2$ instead of $(x,y)$.
	\end{enumerate}
	The two arithmetic circuits can be computed in polynomial time in $n,m$ and in the sizes of $S,P,C$. Since $n$ and $m$ are upper bounded by the sizes of $S$ and $C$ respectively, they can be constructed in polynomial time in the sizes of $S,P,C$. Furthermore, note that the two circuits are well-behaved. In fact, they only use a \emph{constant} number of true multiplication gates. To see this, note that true multiplication gates are only used for the matrix multiplication in step 3 and for step 4. In particular, steps 1 and 2 do not need to use any true multiplication gates at all (see, e.g., \citep{DGP09,CDT09}).
\end{proof}

\subsection{Correctness}

To show the correctness of the construction, we prove the following
lemma, which states that $0.01$-KKT points of $f$ only lie at solutions for the
\eol/ instance or the \iter/ instance.

\begin{lemma}\label{lem:kkt-correct-solutions}
	Let $\varepsilon = 0.01$. We have that $(x, y)$ is an $\varepsilon$-KKT point of
	$f$ on the domain $[0,N]^2$ only if $(x,y)$ lies in a ``solution region'', namely:
	\begin{itemize}
		\item $(x,y)$ lies in a big square $B(v,v)$, such that $v \in [2^n] \setminus \{1\}$ is a source or sink of the \eol/ instance $S,P$, or
		\item $(x,y)$ lies in a medium square $M(u,u)$ of some PLS-Labyrinth gadget, such that $u \in [2^m]$ is a solution to the \iter/ instance $C$, i.e., $C(u)>u$ and $C(C(u))=C(u)$.
	\end{itemize}
\end{lemma}

\begin{proof}
	Even though we have defined $\varepsilon$-KKT points in \cref{sec:nonlinearopt} with respect to the $\ell_2$-norm, here it is more convenient to consider the $\ell_\infty$-norm instead. Note that any $\varepsilon$-KKT point w.r.t.\ the $\ell_2$-norm is also an $\varepsilon$-KKT point w.r.t.\ the $\ell_\infty$-norm. Thus, if \cref{lem:kkt-correct-solutions} holds for $\varepsilon$-KKT points w.r.t.\ the $\ell_\infty$-norm, then it automatically also holds for $\varepsilon$-KKT points w.r.t.\ the $\ell_2$-norm.
	
	For the domain $[0,N]^2$, it is easy to see that a point $x \in [0,N]^2$ is an $\varepsilon$-KKT point (with respect to the $\ell_\infty$-norm) if and only if
	\begin{itemize}
		\item for all $i \in \{1,2\}$ with $x_i \neq 0$ : $[\nabla f(x)]_i \leq \varepsilon$
		\item for all $i \in \{1,2\}$ with $x_i \neq N$ : $-[\nabla f(x)]_i \leq \varepsilon$.
	\end{itemize}
	Intuitively, these conditions state that if $x$ is not on the boundary of $[0,N]^2$, then it must hold that $\|\nabla f(x)\|_\infty \leq \varepsilon$. If $x$ is on the boundary of $[0,1]^2$, then ``$-\nabla f(x)$ must point straight outside the domain, up to an error of $\varepsilon$''.

	In order to prove \cref{lem:kkt-correct-solutions}, we will show that any small square that does
	not lie in a solution region, does not contain any $\varepsilon$-KKT point.
	The behaviour of the function in a given small square depends on the information
	we have about the four corners, namely the colours and arrows at the four corners,
	but also on the position of the square in our instance, since the value defined
	by a colour depends on the position. For our proof, it is convenient to consider
	a square with the (colour and arrow) information about its four corners, but without any information
	about its position. Indeed, if we can show that a square does not contain any
	$\varepsilon$-KKT point using only this information, then this will always hold,
	wherever the square is positioned. As a result, we obtain a finite number of squares
	(with colour and arrow information) that we need to check.
	Conceptually, this is a straightforward task: for each small square we get a set
	of cubic polynomials that could be generated by bicubic interpolation for that
	square, and we must prove that no polynomial in that set has 
	an $\varepsilon$-KKT point within the square.

	However, this still leaves us with 101 distinct small squares to verify. In a earlier version of this paper, we presented a computer-assisted proof that used an SMT solver for this task. For each square we wrote down an SMT formula that encodes ``there exists an $\varepsilon$-KKT
	point within this square''. The SMT solver was then used to check this formula for
	satisfiability over the real numbers.\footnote{SMT solvers are capable
		of deciding satisfiability for the existential theory of the reals. There are no
		rounding issues or floating point errors to worry about.} We applied the Z3 SMT solver~\citep{MB08} to all 101 squares in our
	construction and we found that the formula is unsatisfiable for every square that does not lie
	directly at the end of a line of the \eol/ instance, or at a
	solution of the \iter/ instance, which proves \cref{lem:kkt-correct-solutions}. A detailed description of how the SMT formulas were constructed, as well as the full output of the solver is available online.\footnote{\url{https://github.com/jfearnley/PPADPLS/}}
	
	Below, we present a direct non-computer-assisted proof of \cref{lem:kkt-correct-solutions}. This relatively concise proof makes extensive use of symmetries to drastically reduce the number of cases that need to be verified. We view this proof not as a replacement of the computer-assisted proof, but instead as additional evidence that \cref{lem:kkt-correct-solutions} holds. Indeed, the two proofs complement each other nicely, in the sense that they cover each other's potential weaknesses. The new direct proof eliminates any concerns one might have about the implementation of the SMT solver, or about possible hardware failures. On the other hand, the original computer-assisted proof verified all 101 squares individually, thus avoiding any errors in the use of symmetries and also performing the additional ``sanity check'' of showing that ``bad'' squares (namely, squares that can only appear in a solution region) indeed introduce $\varepsilon$-KKT points.
	
	We begin by describing the transformations that we use to group ``symmetric'' squares together. Throughout, we consider a square where each of its four corners has an associated value in $\mathbb{Z}$ and an arrow pointing in a cardinal direction. The first transformation is \emph{reflection with respect to the axis $y=1/2$}. Applying this transformation to a square has the following effect: the two corners at the top of the square now find themselves at the bottom of the square (and vice-versa) and the sign of the $y$-coordinate of each arrow is flipped. Using \cref{eq:bicubic,eq:bicubic-coefficients} one can check  that taking the bicubic interpolation of this reflected square yields the same result as taking the interpolation of the original square and then applying the reflection to the interpolated function (which corresponds to considering $(x,y) \mapsto f(x,1-y)$). We can summarize this by saying that bicubic interpolation \emph{commutes} with this transformation.
	
	The second transformation we use is \emph{reflection with respect to the axis $y=x$}, i.e., the diagonal through the square. This corresponds to swapping the corners $(0,1)$ and $(1,0)$ of the square, and additionally also swapping the $x$- and $y$-coordinate of the arrows at all four corners. Again, using \cref{eq:bicubic,eq:bicubic-coefficients} one can directly verify that this transformation also commutes with bicubic interpolation (where applying the transformation to the interpolated function $f$ corresponds to considering $(x,y) \mapsto f(y,x)$). These two transformations are already enough to obtain any rotation or reflection of the square. We will only need one more transformation: \emph{negation}. This corresponds to negating the values and arrows at the four corners, where ``negating an arrow'' just means replacing it by an arrow in the opposite direction. Using \cref{eq:bicubic,eq:bicubic-coefficients}, it is immediately clear that negation commutes with bicubic interpolation.
	
	Since all three transformations commute with bicubic interpolation, this continues to hold for any more involved transformation that is constructed from these basic three. Furthermore, it is easy to see that the basic transformations do not introduce $\varepsilon$-KKT points. Indeed, if a function does not have any $\varepsilon$-stationary points, then applying any reflection or taking the negation cannot change that property. As a result, if two squares are ``symmetric'' (i.e., one can be obtained from the other by applying a combination of the three basic transformations), then it is enough to verify that just one of these two squares does not contain any $\varepsilon$-KKT points when we take the bicubic interpolation.
	
	Using this notion of symmetry, the squares that need to be verified can be grouped into just four different groups, namely Groups 1 to 4, as shown in \cref{fig:symmetry-grouping}. The remaining squares lie in Group 0, which contains all squares that always lie close to actual solutions and thus do not need to be verified. They are shown in \cref{fig:symmetry-group-0}. We now prove that the squares in Groups 1 to 4 do not contain $\varepsilon$-KKT points.

	\begin{figure}
		\centering
		{\large Group 1:}
		
		\medskip
		
		\resizebox{0.89\textwidth}{!}{
			\includegraphics{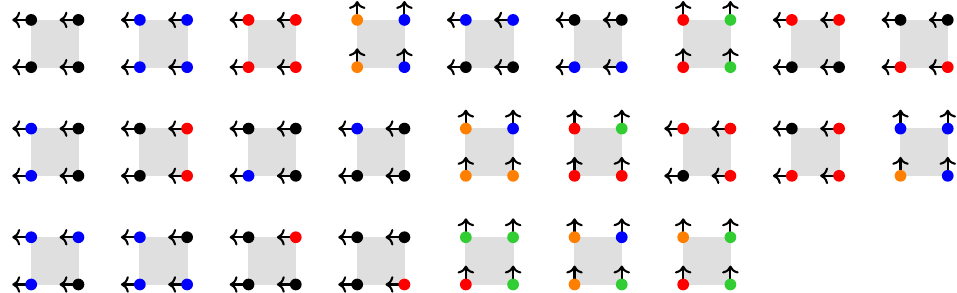}
		}

		\vspace{10pt}

		{\large Group 2:}
		
		\medskip
		
		\resizebox{0.89\textwidth}{!}{
			\includegraphics{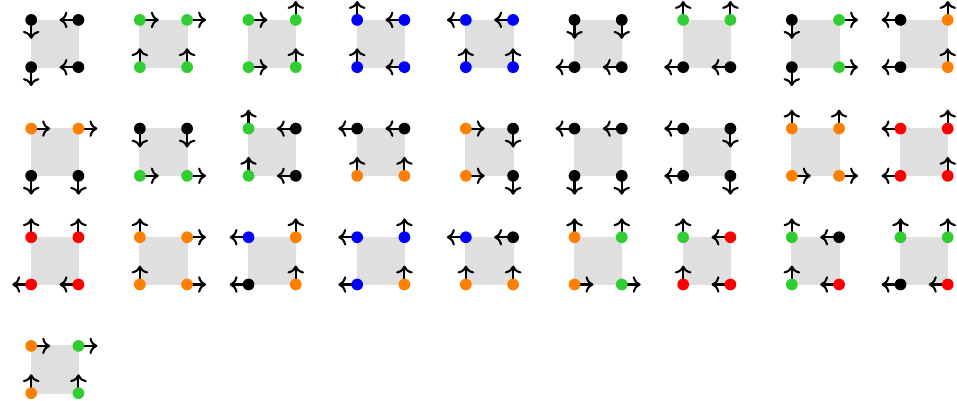}
		}

		\bigskip

		{\large Group 3:}
		
		\medskip
		
		\resizebox{0.89\textwidth}{!}{
			\includegraphics{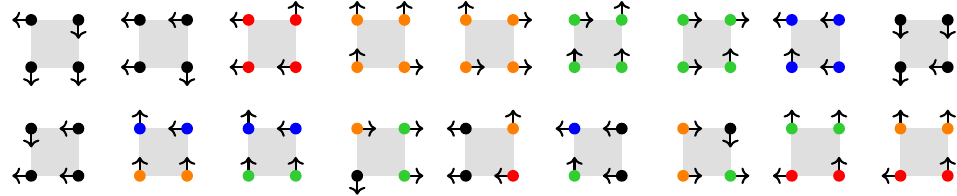}
		}

		\vspace{10pt}

		{\large Group 4:}
		
		\medskip
		
		\resizebox{0.89\textwidth}{!}{
			\includegraphics{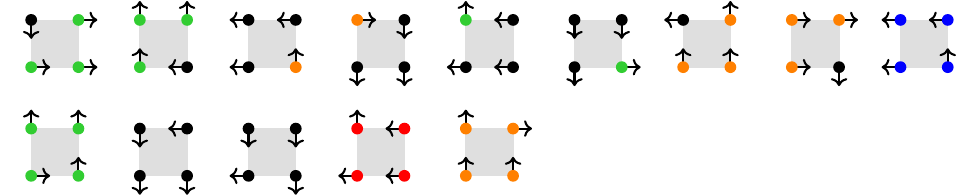}
		}
		
		\bigskip
		
		\caption{Complete list of all squares that need to be verified, partitioned into four groups.}\label{fig:symmetry-grouping}
	\end{figure}

	\begin{figure}
		\centering
		{\large Group 0:}
		
		\medskip
		
		\resizebox{0.9\textwidth}{!}{
			\includegraphics{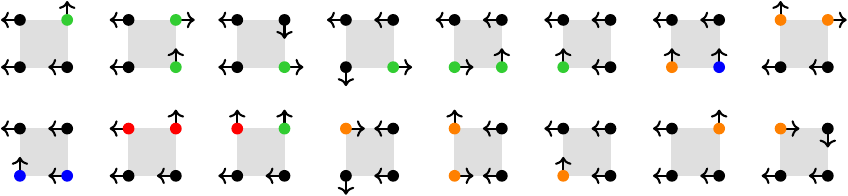}
		}
		
		\smallskip
		
		\caption{List of squares that lie close to actual solutions, and thus do \emph{not} need to be verified.}\label{fig:symmetry-group-0}
	\end{figure}

	\paragraph{\bf Group 1}
	
	This group contains the squares where all the arrows point in the same direction. These squares are all symmetric to a square of the following form:
	
	\begin{minipage}{0.6\textwidth}
		\begin{center}
			\resizebox{85pt}{!}{
				\begin{tikzpicture}[scale=5]
				\path (-0.07, -0.07) -- (0.27, 0.27);
				\draw[lightgray,dashed] (0.0,0.0) rectangle (0.2,0.2);
				\path[->,line width=1,black] (0.0, 0.0) edge (0.08, 0.0);
				\path[->,line width=1,black] (0.0, 0.2) edge (0.08, 0.2);
				\path[->,line width=1,black] (0.2, 0.0) edge (0.28, 0.0);
				\path[->,line width=1,black] (0.2, 0.2) edge (0.28, 0.2);
				\node [draw=black,fill=white,circle,inner sep=0,minimum size=9] at (0.0, 0.0)  {\scriptsize $c$};
				\node [draw=black,fill=white,circle,inner sep=0,minimum size=9] at (0.0, 0.2)  {\scriptsize $a$};
				\node [draw=black,fill=white,circle,inner sep=0,minimum size=9] at (0.2, 0.0)  {\scriptsize $d$};
				\node [draw=black,fill=white,circle,inner sep=0,minimum size=9] at (0.2, 0.2)  {\scriptsize $b$};
				\end{tikzpicture}
			}
		\end{center}
	\end{minipage} \hfill
	\begin{minipage}{0.4\textwidth}
		\begin{tabular}{ll}
			Conditions:\\
			~$a \geq b+1$\\
			~$c \geq d+1$\\
		\end{tabular}
	\end{minipage}
	where $a,b,c,d$ are the values at the four corners of the square as shown.  We now prove that bicubic interpolation does not introduce any $\varepsilon$-KKT point in such a square.
	
	Recall that $\varepsilon=0.01$ and that we use $\delta=1/2$ as the length multiplier for the arrows. Also recall that the arrows point in the \emph{opposite} direction of the gradient. Thus, our goal here will be to prove that the bicubic interpolation $f$ satisfies $\frac{\partial f}{\partial x} (x,y) < -\varepsilon$ for all points $(x,y) \in [0,1]^2$ in the square, which immediately implies that the square does not contain any $\varepsilon$-KKT point.
	
	We can compute the coefficients of the interpolation using \cref{eq:bicubic-coefficients} by setting $f(0,0) := c$, $f(1,0) := d$, $f(0,1) := a$, $f(1,1) := b$, $f_x(0,0) = f_x(1,0) = f_x(0,1) = f_x(1,1) := - \delta = -1/2$, and the remaining terms to zero. Then, substituting these coefficients into \cref{eq:bicubic-gradient} and simplifying, we obtain that for all $x,y \in [0,1]$:
	\begin{align*}
	\frac{\partial f}{\partial x} (x,y) &= -\frac{1}{2}-3x(1-x)\Big(1+2(1-y^2(3-2y))(c-d-1) + 2y^2(3-2y)(a-b-1)\Big)\\
	&\leq -\frac{1}{2} < -\varepsilon
	\end{align*}
	where we used the fact that $y^2(3-2y) \in [0,1]$ for all $y \in [0,1]$, as well as the conditions $c-d-1 \geq 0$ and $a-b-1 \geq 0$.

	\paragraph{\bf Group 2}
	
	This group contains all the squares that are symmetric to a square of the following form:
	
	\begin{minipage}{0.6\textwidth}
		\begin{center}
			\resizebox{85pt}{!}{
				\begin{tikzpicture}[scale=5]
				\path (-0.07, -0.07) -- (0.27, 0.27);
				\draw[darkgray,dotted] (0.0,0.134) -- (0.2,0.134);
				\draw[lightgray,dashed] (0.0,0.0) rectangle (0.2,0.2);
				\path[->,line width=1,black] (0.0, 0.0) edge (0.0, 0.08);
				\path[->,line width=1,black] (0.0, 0.2) edge (0.08, 0.2);
				\path[->,line width=1,black] (0.2, 0.0) edge (0.2, 0.08);
				\path[->,line width=1,black] (0.2, 0.2) edge (0.28, 0.2);
				\node [draw=black,fill=white,circle,inner sep=0,minimum size=9] at (0.0, 0.0)  {\scriptsize $c$};
				\node [draw=black,fill=white,circle,inner sep=0,minimum size=9] at (0.0, 0.2)  {\scriptsize $a$};
				\node [draw=black,fill=white,circle,inner sep=0,minimum size=9] at (0.2, 0.0)  {\scriptsize $d$};
				\node [draw=black,fill=white,circle,inner sep=0,minimum size=9] at (0.2, 0.2)  {\scriptsize $b$};
				\end{tikzpicture}
			}
		\end{center}
	\end{minipage} \hfill
	\begin{minipage}{0.4\textwidth}
		\begin{tabular}{ll}
			\\
			Conditions:\\
			~$a \geq b+1$\\
			~$c \geq a+1$\\
			~$d \geq b+1$\\
			~$c \geq d-1$\\
			\\
		\end{tabular}
	\end{minipage}
	In order to prove that such squares do not contain any $\varepsilon$-KKT points, we distinguish two cases depending on the value of $y$: when $y \in [0,2/3]$ (i.e., below the dotted line) we show that $\frac{\partial f}{\partial y} (x,y) < -\varepsilon$, and when $y \in [2/3,1]$ (i.e., above the dotted line) we show that $\frac{\partial f}{\partial x} (x,y) < -\varepsilon$. Consider first the case where $y \leq 2/3$. Using \cref{eq:bicubic-coefficients,eq:bicubic-gradient} and simplifying we obtain
	\begin{align*}
	\frac{\partial f}{\partial y} (x,y) &= \begin{aligned}[t] -\frac{1}{2} (1-y) \Bigg( 1 + 3y \bigg( 1 + 2x + 4\bigg(c - a &- \frac{1}{2}\bigg) \Big(1-x^2(3-2x)\Big)\\
	&+ 2(2d-2b-2)x^2(3-2x) \bigg) \Bigg)\end{aligned}\\
	&\leq -\frac{1}{2} (1-y) \leq -\frac{1}{6} < -\varepsilon
	\end{align*}
	where we used the fact that $x^2(3-2x) \in [0,1]$ for all $x \in [0,1]$, as well as $c-a -1/2 \geq 0$ and $2d-2b-2 \geq 0$ to obtain the first inequality, and then $y \leq 2/3$. Next, consider the case where $y \geq 2/3$. Using \cref{eq:bicubic-coefficients,eq:bicubic-gradient} and simplifying we obtain
	\begin{align*}
	\frac{\partial f}{\partial x} (x,y) &= \begin{aligned}[t] -\frac{1}{2} y^2 (3-2y) - 3x(1-x)\Big( (2a-2b-1)y^2(3 &- 2y)\\
	&+ 2(c-d)\left(1-y^2(3-2y)\right) \Big)\end{aligned}\\
	&\leq -\frac{1}{2} y^2 (3-2y) - 3x(1-x)\left(3y^2(3-2y)-2\right)\\
	&\leq -\frac{1}{2} y^2 (3-2y) \leq -\frac{1}{3} < -\varepsilon
	\end{align*}
	where we used $2a-2b-1 \geq 1$ and $c-d \geq -1$ to obtain the first inequality (as well as the fact that $y^2(3-2y) \in [0,1]$ for all $y \in [0,1]$ as before), and for the second and third inequality the fact that for $y \in [2/3,1]$ we have $y^2(3-2y) \geq 2/3$.

	\paragraph{\bf Group 3}
	
	This group contains all the squares that are symmetric to a square of the following form:
	
	\begin{minipage}{0.6\textwidth}
		\begin{center}
			\resizebox{85pt}{!}{
				\begin{tikzpicture}[scale=5]
				\path (-0.07, -0.07) -- (0.27, 0.27);
				\draw[darkgray,dotted] (0.0,0.134) -- (0.1,0.134) -- (0.1,0.2);
				\draw[lightgray,dashed] (0.0,0.0) rectangle (0.2,0.2);
				\path[->,line width=1,black] (0.0, 0.0) edge (0.0, 0.08);
				\path[->,line width=1,black] (0.0, 0.2) edge (0.08, 0.2);
				\path[->,line width=1,black] (0.2, 0.0) edge (0.2, 0.08);
				\path[->,line width=1,black] (0.2, 0.2) edge (0.2, 0.28);
				\node [draw=black,fill=white,circle,inner sep=0,minimum size=9] at (0.0, 0.0)  {\scriptsize $c$};
				\node [draw=black,fill=white,circle,inner sep=0,minimum size=9] at (0.0, 0.2)  {\scriptsize $a$};
				\node [draw=black,fill=white,circle,inner sep=0,minimum size=9] at (0.2, 0.0)  {\scriptsize $d$};
				\node [draw=black,fill=white,circle,inner sep=0,minimum size=9] at (0.2, 0.2)  {\scriptsize $b$};
				\end{tikzpicture}
			}
		\end{center}
	\end{minipage} \hfill
	\begin{minipage}{0.4\textwidth}
		\begin{tabular}{ll}
			\\
			Conditions:\\
			~$a \geq b+1$\\
			~$c \geq a+1$\\
			~$d \geq b+1$\\
			~$c \geq d-1$\\
			\\
		\end{tabular}
	\end{minipage}
	Consider first the case where $(x,y) \in [0,1/2] \times [2/3,1]$, i.e., the point lies in the top-left corner region, as delimited by the dotted line. Using \cref{eq:bicubic-coefficients,eq:bicubic-gradient} and simplifying we obtain
	\begin{align*}
	\frac{\partial f}{\partial x} (x,y) &= \begin{aligned}[t] -\frac{1}{2} (1-x) \bigg( y^2(3 -2y) + 3x\Big( y^2 -2y^2(3-2y) &+4(a-b)y^2(3-2y)\\
	&+4(c-d)\left(1-y^2(3-2y)\right) \Big) \bigg)\end{aligned}\\
	&\leq -\frac{1}{2} (1-x) \Big( y^2(3 -2y) + 3x(y^2 + 6y^2(3-2y) -4) \Big)\\
	&\leq -\frac{1}{2} (1-x) y^2(3 -2y) \leq -\frac{1}{6} < -\varepsilon
	\end{align*}
	where we used $a-b \geq 1$, $c-d \geq -1$ and $y^2(3-2y) \in [0,1]$ for the first inequality, and then the fact that $y^2(3-2y) \geq 2/3$ for $y \in [2/3,1]$, as well as $x \leq 1/2$, for the second and third inequality.

	Next, consider the case where $(x,y) \notin [0,1/2] \times [2/3,1]$, i.e., $x > 1/2$ or $y < 2/3$. Using \cref{eq:bicubic-coefficients,eq:bicubic-gradient} and simplifying we obtain
	\begin{align*}
	\frac{\partial f}{\partial y} (x,y) &= \begin{aligned}[t] -\frac{1}{2} + \frac{y}{2} \bigg( y\big(1 &-x^2(3-2x)\big) - (1-y)\Big( -4 +3x(2-x) -5x^2(3-2x)\\
	&+12(d-b-1)x^2(3-2x) +12(c-a-1)\left(1-x^2(3-2x)\right) \Big) \bigg)\end{aligned}\\
	&\leq -\frac{1}{2} + \frac{y}{2} \Big( y\left(1-x^2(3-2x)\right) - (1-y)\big(8 +3x(2-x) -5x^2(3-2x)\big) \Big)\\
	&\leq -\frac{1}{2} + \frac{y^2}{2}\left(1-x^2(3-2x)\right)\\
	&\leq -\frac{1}{2} + \max\left\{\frac{2}{9},\frac{3}{8}\right\} < -\varepsilon
	\end{align*}
	where we used $d-b-1 \geq 0$, $c-a-1 \geq 0$ and $x^2(3-2x) \in [0,1]$ for the first and second inequalities. For the third inequality we used the fact that $x>1/2$ or $y<2/3$, where we note that $1-x^2(3-2x) \in [0,1]$ when $x \in [0,1]$, and $1-x^2(3-2x) \leq 3/4$ when $x \in [1/2,1]$.

	\paragraph{\bf Group 4}
	
	This group contains all the squares that are symmetric to a square of the following form:
	
	\begin{minipage}{0.6\textwidth}
		\begin{center}
			\resizebox{85pt}{!}{
				\begin{tikzpicture}[scale=5]
				\path (-0.07, -0.07) -- (0.27, 0.27);
				\draw[darkgray,dotted] (0.0,0.067) -- (0.1,0.067) -- (0.1,0);
				\draw[lightgray,dashed] (0.0,0.0) rectangle (0.2,0.2);
				\path[->,line width=1,black] (0.0, 0.0) edge (0.08, 0.0);
				\path[->,line width=1,black] (0.0, 0.2) edge (0.0, 0.28);
				\path[->,line width=1,black] (0.2, 0.0) edge (0.2, 0.08);
				\path[->,line width=1,black] (0.2, 0.2) edge (0.2, 0.28);
				\node [draw=black,fill=white,circle,inner sep=0,minimum size=9] at (0.0, 0.0)  {\scriptsize $c$};
				\node [draw=black,fill=white,circle,inner sep=0,minimum size=9] at (0.0, 0.2)  {\scriptsize $a$};
				\node [draw=black,fill=white,circle,inner sep=0,minimum size=9] at (0.2, 0.0)  {\scriptsize $d$};
				\node [draw=black,fill=white,circle,inner sep=0,minimum size=9] at (0.2, 0.2)  {\scriptsize $b$};
				\end{tikzpicture}
			}
		\end{center}
	\end{minipage} \hfill
	\begin{minipage}{0.4\textwidth}
		\begin{tabular}{ll}
			\\
			Conditions:\\
			~$c \geq a+1$\\
			~$d \geq b+1$\\
			~$c \geq d+1$\\
			~$a \geq b-1$\\
			\\
		\end{tabular}
	\end{minipage}
	Consider first the case where $(x,y) \in [0,1/2] \times [0,1/3]$, i.e., the point lies in the bottom-left corner region, as delimited by the dotted line. Using \cref{eq:bicubic-coefficients,eq:bicubic-gradient} and simplifying we obtain
	\begin{align*}
	\frac{\partial f}{\partial x} (x,y) &= \begin{aligned}[t] -\frac{1}{2}(1-x) \bigg( 1 -y^2(3-2y) + 3x \Big( -1 +y(2 &-y) + 4(a-b)y^2(3-2y)\\
	&+ 4(c-d)\left(1 -y^2(3+2y)\right) \Big) \bigg)\end{aligned}\\
	&\leq -\frac{1}{2}(1-x) \Big( 1 -y^2(3-2y) + 3x \big( 3 + y(2-y) -8y^2(3-2y) \big) \Big)\\
	&\leq -\frac{1}{2}(1-x) \left( 1 -y^2(3-2y)\right) \leq -\frac{1}{6} < -\varepsilon
	\end{align*}
	where we used $a-b \geq -1$, $c-d \geq 1$ and $y^2(3-2y) \in [0,1]$ for the first inequality, and $y^2(3-2y) \leq 1/3$ for $y \in [0,1/3]$ for the second inequality. For the third inequality, we used the fact that $x \leq 1/2$ and $1 -y^2(3-2y) \geq 2/3$ for $y \in [0,1/3]$.

	Next, consider the case where $(x,y) \notin [0,1/2] \times [0,1/3]$, i.e., $x > 1/2$ or $y > 1/3$. Using \cref{eq:bicubic-coefficients,eq:bicubic-gradient} and simplifying we obtain
	\begin{align*}
	\frac{\partial f}{\partial y} (x,y) &= \begin{aligned}[t] -\frac{1}{2} +\frac{1}{2}(1-y)\big(1 &-x^2(3-2x)\big) -\frac{3}{2}y(1-y)\Big( 3 - x(2-x)\\
	&+4(d-b-1)x^2(3 -2x) +4(c-a-1)\left(1 -x^2(3-2x)\right) \Big)\end{aligned}\\
	&\leq -\frac{1}{2} +\frac{1}{2}(1-y)\left(1-x^2(3-2x)\right) -\frac{3}{2}y(1-y)\big(3 - x(2-x)\big)\\
	&\leq -\frac{1}{2} +\frac{1}{2}(1-y)\left(1-x^2(3-2x)\right)\\
	&\leq -\frac{1}{2} + \max\left\{\frac{3}{8},\frac{1}{3}\right\} < -\varepsilon
	\end{align*}
	where we used $d-b-1 \geq 0$, $c-a-1 \geq 0$ and $x^2(3-2x) \in [0,1]$ for the first inequality, and the fact that $3-x(2-x) \geq 0$ for the second inequality. For the third inequality we used the fact that $x>1/2$ or $y>1/3$, where we again note that $1-x^2(3-2x) \in [0,1]$ when $x \in [0,1]$, and $1-x^2(3-2x) \leq 3/4$ when $x \in [1/2,1]$.

	\paragraph{\bf Boundary}
	
	Up to this point we have shown that there are no unintended $\varepsilon$-KKT points in the interior of the domain $[0,N]^2$ of our instance. It remains to show that no $\varepsilon$-KKT points appear on the boundary of the domain. Intuitively, this corresponds to showing that $-\nabla f$ never points ``straight outside the domain'' when we are on the boundary.
	
	First of all, it is easy to see that there is no $\varepsilon$-KKT point at any of the four corners of the domain $[0,N]^2$. This follows from the fact that at any such corner the arrow never points outside the domain (see \cref{fig:X-full-boundary}). Since these corners of the domain are also corners of their respective squares, $- \nabla f$ (where $f$ is obtained by bicubic interpolation in every square) will automatically be equal to the arrow at that corner (scaled by $\delta=1/2$), which ensures that it is not an $\varepsilon$-KKT point.
	
	It remains to consider points that lie on the boundary of the domain, except at the corners. Using the first two transformations introduced earlier it is easy to check that any square that lies on the boundary (see \cref{fig:X-full-boundary}) is symmetric to one of the three following cases:

	\begin{minipage}{0.3\textwidth}
		\begin{center}
			\resizebox{80pt}{!}{
				\begin{tikzpicture}[scale=5]
				\path (-0.04, 0.0) -- (0.24, 0.28);
				\draw[lightgray,dashed] (0.0,0.0) rectangle (0.2,0.2);
				\path[lightgray,line width=0.7] (0.0,0.0) edge (0.0,0.2);
				\path[->,line width=1,black] (0.0, 0.0) edge (0.08, 0.0);
				\path[->,line width=1,black] (0.0, 0.2) edge (0.08, 0.2);
				\node [draw=black,fill=white,circle,inner sep=0,minimum size=9] at (0.0, 0.0)  {\scriptsize $c$};
				\node [draw=black,fill=white,circle,inner sep=0,minimum size=9] at (0.0, 0.2)  {\scriptsize $a$};
				\end{tikzpicture}
			}
		\end{center}
		
		\begin{center}
			Case (i)	
		\end{center}
	\end{minipage}
	\hfill
	\begin{minipage}{0.3\textwidth}
		\begin{center}
			\resizebox{80pt}{!}{
				\begin{tikzpicture}[scale=5]
				\path (-0.04, 0.0) -- (0.24, 0.28);
				\draw[lightgray,dashed] (0.0,0.0) rectangle (0.2,0.2);
				\path[lightgray,line width=0.7] (0.0,0.0) edge (0.0,0.2);
				\path[->,line width=1,black] (0.0, 0.0) edge (0.0, 0.08);
				\path[->,line width=1,black] (0.0, 0.2) edge (0.0, 0.28);
				\node [draw=black,fill=white,circle,inner sep=0,minimum size=9] at (0.0, 0.0)  {\scriptsize $c$};
				\node [draw=black,fill=white,circle,inner sep=0,minimum size=9] at (0.0, 0.2)  {\scriptsize $a$};
				\end{tikzpicture}
			}
		\end{center}
		
		\begin{center}
			Case (ii)	
		\end{center}
	\end{minipage}
	\hfill
	\begin{minipage}{0.3\textwidth}
		\begin{center}
			\resizebox{80pt}{!}{
				\begin{tikzpicture}[scale=5]
				\path (-0.04, 0.0) -- (0.24, 0.28);
				\draw[darkgray,dotted] (0.0,0.1) -- (0.2,0.1);
				\draw[lightgray,dashed] (0.0,0.0) rectangle (0.2,0.2);
				\path[lightgray,line width=0.7] (0.0,0.0) edge (0.0,0.2);
				\path[->,line width=1,black] (0.0, 0.0) edge (0.0, 0.08);
				\path[->,line width=1,black] (0.0, 0.2) edge (0.08, 0.2);
				\node [draw=black,fill=white,circle,inner sep=0,minimum size=9] at (0.0, 0.0)  {\scriptsize $c$};
				\node [draw=black,fill=white,circle,inner sep=0,minimum size=9] at (0.0, 0.2)  {\scriptsize $a$};
				\end{tikzpicture}
			}
		\end{center}
		
		\begin{center}
			Case (iii)
		\end{center}
	\end{minipage}
	
	\smallskip
	
	\noindent where for cases (ii) and (iii) we also have that $c \geq a+1$. In each of these three illustrations the square touches the boundary of the domain with its left edge (which is not dashed). In order to show that there are no $\varepsilon$-KKT points on that left edge, it suffices to show that for any $(x,y)$ on that edge (i.e., $x=0$, $y \in [0,1]$) we have $\frac{\partial f}{\partial y} (x,y) \notin [-\varepsilon,\varepsilon]$ or $\frac{\partial f}{\partial x} (x,y) < -\varepsilon$.
	
	\begin{itemize}
		\item For case (i), we have
		\begin{align*}
		\frac{\partial f}{\partial x} (0,y) &= -\frac{1}{2}y^2(3-2y) -\frac{1}{2}\left(1-y^2(3-2y)\right) = -\frac{1}{2} < -\varepsilon.
		\end{align*}
		\item For case (ii), we have
		\begin{align*}
		\frac{\partial f}{\partial y} (0,y) &= -\frac{1}{2} -3(2c-2a-1)(1-y)y \leq -\frac{1}{2} < -\varepsilon
		\end{align*}
		where we used $2c-2a-1 \geq 0$.
		\item For case (iii), if $y \geq 1/2$ then we have
		\begin{align*}
		\frac{\partial f}{\partial x} (0,y) &= -\frac{1}{2}y^2(3-2y) \leq -\frac{1}{4} < -\varepsilon
		\end{align*}
		since $y^2(3-2y) \geq 1/2$ when $y \in [1/2,1]$. On the other hand, if $y \leq 1/2$ then we have
		\begin{align*}
		\frac{\partial f}{\partial y} (0,y) &= -\frac{1}{2}(1-y)\left(1+3(4c-4a-1)y\right) \leq -\frac{1}{2}(1-y) \leq - \frac{1}{4} < -\varepsilon
		\end{align*}
		where we used $4c-4a-1 \geq 0$.
	\end{itemize}
	It follows that there are no $\varepsilon$-KKT points on the boundary of the domain. This completes the proof of \cref{lem:kkt-correct-solutions}.
\end{proof}

\subsection{Re-scaling}\label{sec:kkt-scaling}

The last step of the reduction is to re-scale the function $f$ so that it is defined on $[0,1]^2$ instead of $[0,N]^2$. Thus, the final function, which we denote $\widehat{f}$ here, is defined by
$$\widehat{f}(x,y) = (1/N) \cdot f(N \cdot x,N \cdot y).$$
The properties of $f$ proved in \cref{lem:kkt-function-properties} naturally also hold for $\widehat{f}$, in the following sense. Clearly, $\widehat{f}$ is also continuously differentiable. Furthermore, it holds that $\nabla \widehat{f}(x,y) = \nabla f(N \cdot x,N \cdot y)$. Thus, we can easily construct well-behaved arithmetic circuits for $\widehat{f}$ and $\nabla \widehat{f}$ in polynomial time given well-behaved circuits for $f$ and $\nabla f$, which, in turn, can be efficiently constructed according to \cref{lem:kkt-function-properties}. Furthermore, since $\nabla f$ is $L$-Lipschitz-continuous, it is easy to see that $\nabla \widehat{f}$ is $\widehat{L}$-Lipschitz-continuous with $\widehat{L} = N \cdot L = 2^{18} N^2 = 2^{2n+2m+26}$. Finally, note that since $f$ is $L$-Lipschitz-continuous, $\widehat{f}$ is too, and in particular it is also $\widehat{L}$-Lipschitz-continuous.

All these properties imply that the instance of \textsc{KKT} we construct does
not admit any violation solutions. In other words, it satisfies all the expected
promises. Finally, note that any $\varepsilon$-KKT point of $\widehat{f}$ on
$[0,1]^2$ immediately yields an $\varepsilon$-KKT point of $f$ on $[0,N]^2$.
Thus, the correctness of the reduction follows from
\cref{lem:kkt-correct-solutions}.

Note that we can re-scale the instance depending on the parameter regime we are interested in. The instance $(\varepsilon,\widehat{f},\nabla \widehat{f},\widehat{L})$ we have constructed is clearly equivalent to the instance $(\alpha\varepsilon,\alpha\widehat{f},\nabla (\alpha\widehat{f}),\alpha\widehat{L})$ for any $\alpha > 0$. For example, by letting $\alpha = 1/\widehat{L}$, we obtain hard instances with Lipschitz-constant $1$, and with inversely exponential precision parameter.

\section{Gradient Descent and KKT are \texorpdfstring{$\boldsymbol{\ppadpls/}$}{PPAD ∩ PLS}-complete}\label{sec:all-ppadpls-complete}

In this section, we show how the \ppadpls/-hardness of \kkt/ (\cref{thm:main-kkt-hard}) implies that our other problems of interest, including our Gradient Descent problems, are \ppadpls/-complete. Namely, we prove:

\begin{theorem}\label{thm:all-ppadpls-complete}
	The problems \kkt/, \gdls/, \gdfp/ and \gclo/ are $\ppad/~\cap$ \pls/-complete, even when the domain is fixed to be the unit square $[0,1]^2$. This hardness result continues to hold even if one considers the promise-versions of these problems, i.e., only instances without violations.
\end{theorem}

\noindent The hardness results in this theorem are the ``best possible'', in the following sense:
\begin{itemize}
	\item Promise-problem: as mentioned in the theorem, the hardness holds even for the promise-versions of these problems. In other words, the hard instances that we construct are not pathological: they satisfy all the conditions that we would expect from the input, e.g., $\nabla f$ is indeed the gradient of $f$, $\nabla f$ and $f$ are indeed $L$-Lipschitz-continuous, etc.
	\item Domain: the problems remain hard even if we fix the domain to be the unit square $[0,1]^2$, which is arguably the simplest two-dimensional bounded domain. All the problems become polynomial-time solvable if the domain is one-dimensional (\cref{lem:one-dimensional-easy}).
	\item Exponential parameters: in all of our problems, the parameters, such as $\varepsilon$ and $L$, are provided in the input in \emph{binary} representation. This means that the parameters are allowed to be exponentially small or large with respect to the length of the input. Our hardness results make use of this, since the proof of \cref{thm:main-kkt-hard} constructs an instance of \kkt/ where $\varepsilon$ is some constant, but $L$ is exponential in the input length. By a simple transformation, this instance can be transformed into one where $\varepsilon$ is exponentially small and $L$ is constant (see \cref{sec:kkt-scaling}). It is easy to see that at least one of $\varepsilon$ or $L$ must be exponentially large/small, for the problem to be hard on the domain $[0,1]^2$. However, this continues to hold even in high dimension, i.e., when the domain is $[0,1]^n$ (\cref{lem:kkt-poly-params}). In other words, if the parameters are given in unary, the problem is easy, even in high dimension. This is in contrast with the problem of finding a Brouwer fixed point, where moving to domain $[0,1]^n$ makes it possible to prove \ppad/-hardness even when the parameters are given in unary.
\end{itemize}
\cref{thm:all-ppadpls-complete} follows from \cref{thm:main-kkt-hard}, proved in \cref{sec:kkt-ppadpls}, and a set of domain- and promise-preserving reductions as pictured in \cref{fig:reductions}, which are presented in the rest of this section as follows. In \cref{sec:kkt-equiv-GD} we show that the problems \kkt/, \gdls/ and \gdfp/ are equivalent. Then, in \cref{sec:gdls2gclo2ppadpls} we reduce \gdls/ to \gclo/, and finally we show that \gclo/ lies in \ppadpls/.

\subsection{\kkt/ and the Gradient Descent problems are equivalent}\label{sec:kkt-equiv-GD}

The equivalence between \kkt/, \gdls/ and \gdfp/ is proved by providing a ``triangle'' of reductions as shown in \cref{fig:reductions}. Namely, we show that \gdls/ reduces to \gdfp/ (\cref{prop:gdls2gdfp}), \gdfp/ reduces to \kkt/ (\cref{prop:gdfp2kkt}), and \kkt/ reduces to \gdls/ (\cref{prop:kkt2gdls}). All the reductions are domain- and promise-preserving.

\begin{proposition}\label{prop:gdls2gdfp}
	\gdls/ reduces to \gdfp/ using a domain- and promise-preserving reduction.
\end{proposition}

\begin{proof}
	Let $(\varepsilon, \eta, A, b, f, \nabla f, L)$ be an instance of \gdls/. The reduction simply constructs the instance $(\varepsilon', \eta, A, b, f, \nabla f, L)$ of \gdfp/, where $\varepsilon' = \varepsilon/L$. This reduction is trivially domain-preserving and it is also promise-preserving, because any violation of the constructed instance is immediately also a violation of the original instance. Clearly, the reduction can be computed in polynomial time, so it remains to show that any (non-violation) solution of the constructed instance can be mapped back to a solution or violation of the original instance.
	
	Consider any solution $x \in D$ of the \gdfp/ instance, i.e.,
	$$\|x-y\| = \|x - \Pi_D(x - \eta \nabla f(x))\| \leq \varepsilon'.$$
	where $y = \Pi_D(x - \eta \nabla f(x))$. If $x,y$ do not satisfy the $L$-Lipschitzness of $f$, then we have obtained a violation. Otherwise, it must be that
	$$|f(x)-f(y)| \leq L \|x-y\| \leq L \varepsilon' = \varepsilon.$$
	In particular, it follows that
	$$f(y) \geq f(x) - \varepsilon$$
	which means that $x$ is a solution of the original \gdls/ instance.
\end{proof}

\begin{proposition}\label{prop:gdfp2kkt}
	\gdfp/ reduces to \kkt/ using a domain- and promise-preserving reduction.
\end{proposition}

\begin{proof}
	Let $(\varepsilon, \eta, A, b, f, \nabla f, L)$ be an instance of \gdfp/. The reduction simply constructs the instance $(\varepsilon'$, $A$, $b$, $f$, $\nabla f$, $L)$ of \kkt/, where $\varepsilon' = \varepsilon/\eta$. This reduction is trivially domain-preserving and it is also promise-preserving, because any violation of the constructed instance is immediately also a violation of the original instance. Clearly, the reduction can be computed in polynomial time, so it remains to show that any (non-violation) solution of the constructed instance can be mapped back to a solution or violation of the original instance.
	
	In more detail, we will show that any $\varepsilon'$-KKT point must be an $\varepsilon$-approximate fixed point of the gradient descent dynamics. Consider any $\varepsilon'$-KKT point of the \kkt/ instance, i.e., a point $x \in D$ such that there exists $\mu \geq 0$ with $\langle \mu, Ax-b \rangle = 0$ and $\|\nabla f(x) + A^T\mu\| \leq \varepsilon'$.
	
	Let $y = \Pi_D(x - \eta \nabla f(x))$. We want to show that $\|x-y\| \leq \varepsilon$. Since $y$ is the projection of $x - \eta \nabla f(x)$ onto $D$, by \cref{lem:projection-lemma} it follows that for all $z \in D$
	$$\langle x - \eta \nabla f(x) - y, z - y \rangle \leq 0.$$
	Letting $z := x$, this implies that $\langle x-y, x-y \rangle \leq \eta \langle \nabla f(x), x-y \rangle$ and thus
	\begin{align*}
	\|x-y\|^2 = \langle x-y, x-y \rangle \leq \eta \langle \nabla f(x), x-y \rangle &= \eta \langle \nabla f(x) + A^T\mu, x-y \rangle - \eta \langle A^T\mu, x-y \rangle\\
	&\leq \eta \langle \nabla f(x) + A^T\mu, x-y \rangle\\
	&\leq \eta \|\nabla f(x) + A^T\mu\| \cdot \|x-y\|
	\end{align*}
	where we used the Cauchy-Schwarz inequality and the fact that $\langle A^T\mu, x-y \rangle \geq 0$, which follows from
	$$\langle A^T\mu, x-y \rangle = \langle \mu, A(x-y) \rangle = \langle \mu, Ax-b \rangle - \langle \mu, Ay-b \rangle \geq 0$$
	since $\langle \mu, Ax-b \rangle = 0$, $\mu \geq 0$ and $Ay-b \leq 0$ (because $y \in D$).
	
	We can now show that $\|x-y\| \leq \varepsilon$. If $\|x-y\| = 0$, this trivially holds. Otherwise, divide both sides of the inequality obtained above by $\|x-y\|$, which yields
	\[\|x-y\| \leq \eta \|\nabla f(x) + A^T\mu\| \leq \eta \cdot \varepsilon' = \varepsilon. \qedhere \]
\end{proof}

\begin{proposition}\label{prop:kkt2gdls}
	\kkt/ reduces to \gdls/ using a domain- and promise-preserving reduction.
\end{proposition}

\begin{proof}
	Let $(\varepsilon, A, b, f, \nabla f, L)$ be an instance of \kkt/. The reduction simply constructs the instance $(\varepsilon', \eta, A, b, f, \nabla f, L)$ of \gdls/, where $\varepsilon' = \frac{\varepsilon^2}{8L}$ and $\eta = \frac{1}{L}$. This reduction is trivially domain-preserving and it is also promise-preserving, because any violation of the constructed instance is immediately also a violation of the original instance. Clearly, the reduction can be computed in polynomial time, so it remains to show that any (non-violation) solution of the constructed instance can be mapped back to a solution or violation of the original instance.
	
	Consider any $x \in D$ that is a solution of the \gdls/ instance and let $y = \Pi_D(x - \eta \nabla f(x))$. Then, it must be that $f(y) \geq f(x) - \varepsilon'$. We begin by showing that this implies that $\|x-y\| \leq \frac{\varepsilon}{2L}$, or we can find a violation of the \kkt/ instance.
	
	\paragraph{\bf Step 1: Bounding $\boldsymbol{\|x-y\|}$}
	If $x$ and $y$ do not satisfy Taylor's theorem (\cref{lem:taylor}), then we immediately obtain a violation. If they do satisfy Taylor's theorem, it holds that
	$$\langle \nabla f(x), x-y \rangle - \frac{L}{2} \|y-x\|^2 \leq f(x) - f(y) \leq \varepsilon'.$$
	Now, since $y$ is the projection of $x - \eta \nabla f(x)$ onto $D$, by \cref{lem:projection-lemma} it follows that $\langle x - \eta \nabla f(x) - y, z - y \rangle \leq 0$ for all $z \in D$. In particular, by letting $z := x$, we obtain that
	$$\langle \nabla f(x), x - y \rangle \geq \frac{1}{\eta} \langle x - y, x - y \rangle = L \|y-x\|^2$$
	where we used the fact that $\eta = 1/L$. Putting the two expressions together we obtain that
	$$\frac{L}{2} \|y-x\|^2 = L \|y-x\|^2 - \frac{L}{2} \|y-x\|^2 \leq \varepsilon'$$
	which yields that $\|x-y\| \leq \sqrt{2\varepsilon'/L} = \frac{\varepsilon}{2L}$.
	
	\paragraph{\bf Step 2: Obtaining an $\boldsymbol{\varepsilon}$-KKT point}
	Next, we show how to obtain an $\varepsilon$-KKT point or a violation of the \kkt/ instance.
	Note that if $y-x = - \eta \nabla f(x)$, then we immediately have that $\|\nabla f(x)\| = \|x-y\| / \eta \leq \varepsilon/2$, i.e., $x$ is an $\varepsilon$-KKT point. However, because of the projection $\Pi_D$ used in the computation of $y$, in general we might not have that $y-x = - \eta \nabla f(x)$ and, most importantly, $x$ might not be an $\varepsilon$-KKT point. Nevertheless, we show that $y$ will necessarily be an $\varepsilon$-KKT point.
	
	Since $y$ is the projection of $x - \eta \nabla f(x)$ onto $D$, by \cref{lem:projection-lemma} it follows that for all $z \in D$
	$$\langle x - \eta \nabla f(x) - y, z - y \rangle \leq 0.$$
	From this it follows that for all $z \in D$
	$$\langle -\nabla f(x), z - y \rangle \leq \frac{1}{\eta} \langle y - x, z - y \rangle \leq \frac{1}{\eta} \|x-y\| \cdot \|z-y\| \leq \frac{\varepsilon}{2} \|z-y\|$$
	where we used the Cauchy-Schwarz inequality, $\eta = 1/L$ and $\|x-y\| \leq \varepsilon/2L$. Next, unless $x$ and $y$ yield a violation to the $L$-Lipschitzness of $\nabla f$, it must hold that $\|\nabla f(x) - \nabla f(y)\| \leq L \|x - y\| \leq \varepsilon/2$. Thus, we obtain that for all $z \in D$
	\begin{equation}\label{eq:kkt2gdls}
	\begin{split}
	\langle -\nabla f(y), z - y \rangle &= \langle -\nabla f(x), z - y \rangle + \langle \nabla f(x) - \nabla f(y), z - y \rangle\\
	&\leq \frac{\varepsilon}{2} \|z-y\| + \|\nabla f(x) - \nabla f(y)\| \cdot \|z-y\|\\
	&\leq \varepsilon \|z-y\|
	\end{split}
	\end{equation}
	where we used the Cauchy-Schwarz inequality.
	
	Let $I = \{i \in [m]: [Ay-b]_i = 0\}$, i.e., the indices of the constraints that are tight at $y$. Denote by $A_I \in \mathbb{R}^{(m-|I|) \times n}$ the matrix obtained by only keeping the rows of $A$ that correspond to indices in $I$. Assume for now that the statement
	\begin{equation}\label{eq:for-farkas}
	\exists p \in \mathbb{R}^n: A_I p \leq 0, \langle -\nabla f(y), p \rangle > \varepsilon \|p\|
	\end{equation}
	does \emph{not} hold. Then, by a stronger version of Farkas' Lemma, which we prove in the appendix (\cref{lem:new-farkas}), it follows that there exists $\nu \in \mathbb{R}^{|I|}_{\geq 0}$ such that $\|A_I^T\nu + \nabla f(y)\| \leq \varepsilon$. Let $\mu \in \mathbb{R}^m_{\geq 0}$ be such that $\mu$ agrees with $\nu$ on indices $i \in I$, i.e., $\mu_I = \nu$, and $\mu_i=0$ for $i \notin I$. Then we immediately obtain that $A^T\mu = A_I^T\nu$ and thus $\|A^T\mu + \nabla f(y)\| = \|A_I^T\nu + \nabla f(y)\| \leq \varepsilon$. Since we also have that $\langle \mu, Ay - b \rangle = \langle \mu_I, [Ay-b]_I \rangle = 0$ (because $[Ay-b]_I = 0$), it follows that $y$ indeed is an $\varepsilon$-KKT point of $f$ on domain $D$.
	
	It remains to show that the statement \eqref{eq:for-farkas} indeed does not hold. Consider any $p \in \mathbb{R}^n$ such that $A_I p \leq 0$. Then, there exists a sufficiently small $\alpha > 0$ such that $z = y + \alpha p \in D$. Indeed, note that $[Az-b]_i = [Ay - b]_i + \alpha [Ap]_i$ and thus
	\begin{itemize}
		\item for $i \in I$, we get that $[Az-b]_i \leq 0$, since $[Ay - b]_i = 0$ and $[Ap]_i \leq 0$,
		\item for $i \notin I$, we have that $[Ay - b]_i < 0$. If $[Ap]_i \leq 0$, then we obtain $[Az-b]_i \leq 0$ as above. If $[Ap]_i > 0$, then it also holds that $[Az-b]_i \leq 0$, as long as $\alpha \leq - \frac{[Ay - b]_i}{[Ap]_i}$.
	\end{itemize}
	Thus, it suffices to pick $\alpha = \min \left\{- \frac{[Ay - b]_i}{[Ap]_i} : i \notin I, [Ap]_i > 0\right\}$. Note that this indeed ensures that $\alpha > 0$.
	
	Since $z = y + \alpha p \in D$, using \eqref{eq:kkt2gdls} we get that
	$$\langle -\nabla f(y), p \rangle = \frac{1}{\alpha} \langle -\nabla f(y), z - y \rangle \leq \frac{\varepsilon}{\alpha} \|z-y\| = \varepsilon \|p\|.$$
	Thus, we have shown that the statement \eqref{eq:for-farkas} indeed does not hold, and this finishes the proof.
\end{proof}

\subsection{From \gdls/ to \texorpdfstring{$\boldsymbol{\ppadpls/}$}{PPAD ∩ PLS}}\label{sec:gdls2gclo2ppadpls}

In this section we show that \gdls/ reduces to \gclo/ (\cref{prop:gdls2gclo}), and then that \gclo/ lies in \ppadpls/ (\cref{prop:gclo2ppadpls}).

\begin{proposition}\label{prop:gdls2gclo}
	\gdls/ reduces to \gclo/ using a domain- and promise-preserving reduction.
\end{proposition}

\begin{proof}
	This essentially follows from the fact that the local search version of Gradient Descent is a special case of continuous local search, which is captured by the \gclo/ problem.
	Let $(\varepsilon, A, b, \eta, f, \nabla f, L)$ be an instance of \gdls/. The reduction simply constructs the instance $(\varepsilon, A, b, p, g, L')$ of \gclo/, where $p(x) = f(x)$, $g(x) = x - \eta \nabla f(x)$ and $L' = \max \{\eta L + 1, L\}$. We can easily construct an arithmetic circuit computing $g$, given the arithmetic circuit computing $\nabla f$. It follows that the reduction can be computed in polynomial time. In particular, since we extend $\nabla f$ by using only the gates $-$ and $\times \zeta$, the circuit for $g$ is also well-behaved.
	
	Let us now show that any solution to the \gclo/ instance yields a solution to the \gdls/ instance. First of all, by construction of $g$, it immediately follows that any local optimum solution of the \gclo/ instance is also a non-violation solution to the \gdls/ instance.
	
	Next, we show that any pair of points $x,y \in D$ that violate the $(\eta L + 1)$-Lipschitzness of $g$, also violate the $L$-Lipschitzness of $\nabla f$. Indeed, if $x,y$ do not violate the $L$-Lipschitzness of $\nabla f$, then
	$$\|g(x)-g(y)\| \leq \|x-y\| + \eta \|\nabla f(x) -\nabla f(y)\| \leq (\eta L + 1) \|x - y\|.$$
	In particular, any violation to the $L'$-Lipschitzness of $g$ yields a violation to the $L$-Lipschitzness of $\nabla f$. 
	
	Finally, any violation to the $L'$-Lipschitzness of $p$ immediately yields a violation to the $L$-Lipschitzness of $f$. Since any violation to \gclo/ yields a violation to \gdls/, the reduction is also promise-preserving.
\end{proof}

\begin{proposition}\label{prop:gclo2ppadpls}
	\gclo/ lies in \ppadpls/.
\end{proposition}

\begin{proof}
	This essentially follows by the same arguments that were used by \citet{DaskalakisP2011-CLS} to show that \cls/ lies in \ppadpls/. The only difference is that here the domain is allowed to be more general. Consider any instance $(\varepsilon, A, b, p, g, L)$ of \gclo/.
	
	The containment of \gclo/ in \ppad/ follows from a reduction to the problem of finding a fixed point guaranteed by Brouwer's fixed point theorem, which is notoriously \ppad/-complete. Indeed, let $x^* \in D$ be any $\varepsilon/L$-approximate fixed point of the function $x \mapsto \Pi_D(g(x))$, i.e., such that $\|\Pi_D(g(x^*)) - x^*\| \leq \varepsilon/L$. Then, unless $x^*$ and $\Pi_D(g(x^*))$ yield a violation of $L$-Lipschitzness of $p$, it follows that $p(\Pi_D(g(x^*))) \geq p(x^*) - \varepsilon$, i.e., $x^*$ is a solution of the \gclo/ instance. Formally, the reduction works by constructing the instance $(\varepsilon', A, b, g, L)$ of \gbrouwer/, where $\varepsilon' = \varepsilon/L$. The formal definition of \gbrouwer/ can be found in \cref{sec:general-brouwer-rlo}, where it is also proved that the problem lies in \ppad/.
	
	The containment of \cls/ in \pls/ was proved by \citet{DaskalakisP2011-CLS} by reducing \clo/ to a problem called \textsc{Real-Localopt}, which they show to lie in \pls/. \textsc{Real-Localopt} is defined exactly as \clo/, except that the function $g$ is not required to be continuous. In order to show the containment of \gclo/ in \pls/, we reduce to the appropriate generalisation of \textsc{Real-Localopt}, which we simply call \grlo/. Formally, the reduction is completely trivial, since any instance of \gclo/ is also an instance of \grlo/, and solutions can be mapped back as is. The formal definition of \grlo/ can be found in \cref{sec:general-brouwer-rlo}, where it is also proved that the problem lies in \pls/.
\end{proof}

\section{Consequences for Continuous Local Search}\label{sec:cls}

In this section, we explore the consequences of \cref{thm:main-kkt-hard} (and \cref{thm:all-ppadpls-complete}) for the class \cls/, defined by \citet{DaskalakisP2011-CLS} to capture problems that can be solved by ``continuous local search'' methods. In \cref{sec:linear-cls} we also consider a seemingly weaker version of \cls/, which we call \linearcls/, and show that it is in fact the same as \cls/. Finally, we define a Gradient Descent problem where we do not have access to the gradient of the function (which might, in fact, not even be differentiable) and instead use ``finite differences'' to compute an approximate gradient. We show that this problem remains \ppadpls/-complete.

\subsection{Consequences for \cls/}

The class \cls/ was defined by \citet{DaskalakisP2011-CLS} as a more natural counterpart to \ppadpls/. Indeed, Daskalakis and Papadimitriou noted that all the known \ppadpls/-complete problems were unnatural, namely uninteresting combinations of a \ppad/-complete and a \pls/-complete problem. As a result, they defined \cls/, a subclass of \ppadpls/, which is a more natural combination of \ppad/ and \pls/, and conjectured that \cls/ is a \emph{strict} subclass of \ppadpls/. They were able to prove that various interesting problems lie in \cls/, thus further strengthening the conjecture that \cls/ is a more natural subclass of \ppadpls/, and more likely to capture the complexity of interesting problems.

It follows from our results that, surprisingly, \cls/ is actually equal to \ppadpls/.
\begin{theorem}\label{thm:cls-equal-ppadpls}
	\cls/ $=$ \ppadpls/.
\end{theorem}
\noindent Recall that in \cref{thm:all-ppadpls-complete}, we have shown that \gclo/ with domain $[0,1]^2$ is \ppadpls/-complete. \cref{thm:cls-equal-ppadpls} follows from the fact that this problem lies in \cls/, almost by definition. Before proving this in \cref{prop:2d-gclo2cls} below, we explore some further consequences of our results for \cls/.
An immediate consequence is that the two previously known \cls/-complete problems are in fact \ppadpls/-complete.
\begin{corollary}\label{cor:banach-metametric}
	\textup{\textsc{Banach}} and \textup{\textsc{MetametricContraction}} are \ppadpls/-complete.
\end{corollary}
\noindent For the definitions of these problems, which are computational versions of Banach's fixed point theorem, see \citep{DaskTZ18} and \citep{FGMS17}, respectively.

\smallskip
\noindent Furthermore, our results imply that the definition of \cls/ is ``robust'' in the following sense:
\begin{itemize}
	\item Dimension: the class \cls/ was defined by \citet{DaskalakisP2011-CLS} as the set of all TFNP problems that reduce to 3D-\clo/, i.e., \clo/ with $n=3$. Even though it is easy to see that $k$D-\clo/ reduces to $(k+1)$D-\clo/ (\cref{lem:clo-dimension}), it is unclear how to construct a reduction in the other direction. Indeed, similar reductions exist for the Brouwer problem, but they require using a discrete equivalent of Brouwer, namely \eol/, as an intermediate step. Since no such discrete problem was known for \cls/, this left open the possibility of a hierarchy of versions of \cls/, depending on the dimension, i.e., 2D-\cls/ $\subset$ 3D-\cls/ $\subset$ 4D-\cls/ $\dots$. We show that even the two-dimensional version is \ppadpls/-hard, and thus the definition of \cls/ is indeed independent of the dimension used. In other words,
	$$\text{2D-\cls/ = \cls/ = $n$D-\cls/}.$$
	Note that this is tight, since 1D-\clo/ can be solved in polynomial time (\cref{lem:one-dimensional-easy}), i.e., 1D-\cls/ $=$ FP.\footnote{With the slight abuse of notation that FP $\subseteq$ TFNP, as explained in \cref{sec:def-tfnp}.}
	\item Domain: some interesting problems can be shown to lie in \cls/, but the reduction produces a polytopal domain, instead of the standard hypercube $[0,1]^n$. In other words, they reduce to \gclo/, which we have defined as a generalization of \clo/. Since \gclo/ is \ppadpls/-complete (\cref{thm:all-ppadpls-complete}), it follows that \cls/ can equivalently be defined as the set of all TFNP problems that reduce to \gclo/.
	\item Promise: the problem \clo/, which defines \cls/, is a problem with violation solutions. One can instead consider promise-\cls/, which is defined as the set of all TFNP problems that reduce to a promise version of \clo/. In the promise version of \clo/, we restrict our attention to instances that satisfy the promise, i.e., where the functions $p$ and $g$ are indeed $L$-Lipschitz-continuous. The class promise-\cls/ could possibly be weaker than \cls/, since the reduction is required to always map to instances of \clo/ without violations. However, it follows from our results that promise-\cls/\,=\,\cls/, since the promise version of \clo/ is shown to be \ppadpls/-hard, even on domain $[0,1]^2$ (\cref{thm:all-ppadpls-complete}).
	\item Turing reductions: since \ppad/ and \pls/ are closed under Turing reductions \citep{buss2012propositional}, it is easy to see that this also holds for \ppadpls/, and thus by our result also for \cls/.
	\item Circuits: \cls/ is defined using the problem \clo/ where the functions are represented by general arithmetic circuits. If one restricts the type of arithmetic circuit that is used, this might yield a weaker version of \cls/. Linear arithmetic circuits are a natural class of circuits that arise when reducing from various natural problems. We define \linearcls/ as the set of problems that reduce to \clo/ with linear arithmetic circuits. In \cref{sec:linear-cls} we show that \linearcls/ = \cls/.
\end{itemize}
Before moving on to \cref{sec:linear-cls} and \linearcls/, we provide the last reduction in the chain of reductions proving \cref{thm:cls-equal-ppadpls}.

\begin{proposition}\label{prop:2d-gclo2cls}
	\gclo/ with fixed domain $[0,1]^2$ reduces to \textup{2D}-\clo/ using a promise-preserving reduction. In particular, the problem lies in \cls/.
\end{proposition}

\begin{proof}
	Given an instance $(\varepsilon, p, g, L)$ of \gclo/ with fixed domain $[0,1]^2$, we construct the instance $(\varepsilon, p, g', L)$ of 2D-\clo/, where $g'(x) = \Pi_D(g(x))$. Note that since $D=[0,1]^2$, the projection $\Pi_D$ can easily be computed as $[\Pi_D(x)]_i = \min \{1, \max\{0,x_i\}\}$ for all $x \in \mathbb{R}^2$ and $i \in [2]$. In particular, since we extend $g$ by using only the gates $-$, $\times \zeta$, $\min$, $\max$ and rational constants, the circuit for $g'$ is also well-behaved.
	
	Any non-violation solution of the constructed instance is also a solution of the original instance.
	Any violation of the constructed instance is immediately mapped back to a violation of the original instance. In particular, it holds that $\|g'(x)-g'(y)\| \leq \|g(x) - g(y)\|$ for all $x,y \in [0,1]^2$, since projecting two points cannot increase the distance between them. This implies that any violation of the $L$-Lipschitzness of $g'$ is also a violation of the $L$-Lipschitzness of $g$. Note that by \cref{lem:clo-codomain} we do not need to ensure that the codomain of $p$ is in $[0,1]$. Finally, it is easy to see that 2D-\clo/ lies in \cls/, since it immediately reduces to 3D-\clo/ (\cref{lem:clo-dimension}).
\end{proof}

\subsection{\linearcls/ and Gradient Descent with finite differences}\label{sec:linear-cls}

The class \cls/ was defined by \citet{DaskalakisP2011-CLS} using the \clo/ problem which uses arithmetic circuits with gates in $\{+,-,\min,\max,\times,<\}$ and rational constants. In this section we show that even if we restrict ourselves to linear arithmetic circuits (i.e., only the gates in $\{+,-,\min,\max, \times \zeta\}$ and rational constants are allowed), the $\clo/$ problem and \cls/ remain just as hard as the original versions. Note that even though these circuits are called \emph{linear}, they in fact correspond to \emph{piecewise linear} functions.

\begin{tcolorbox}[breakable,enhanced]
	\begin{definition}
		\linclo/:
		
		\noindent\textbf{Input}:
		\begin{itemize}
			\item precision/stopping parameter $\varepsilon > 0$,
			\item linear arithmetic circuits $p: [0,1]^n \to [0,1]$ and $g: [0,1]^n \to [0,1]^n$.
		\end{itemize}
		
		\noindent\textbf{Goal}: Compute an approximate local optimum of $p$ with respect to $g$. Formally, find $x \in [0,1]^n$ such that
		$$p(g(x)) \geq p(x) - \varepsilon.$$
	\end{definition}
\end{tcolorbox}

\noindent For $k \in \mathbb{N}$, we let $k$D-\linclo/ denote the problem \linclo/ where $n$ is fixed to be equal to $k$. Note that the definition of \linclo/ does not require violation solutions, since every linear arithmetic circuit is automatically Lipschitz-continuous with a Lipschitz-constant that can be represented with a polynomial number of bits (\cref{lem:linear-lipschitz}). In particular, \linclo/ reduces to \clo/ and thus to \gclo/.

We define the class 2D-\linearcls/ as the set of all TFNP problems that reduce to 2D-\linclo/. We show that:

\begin{theorem}\label{thm:linear-cls}
	$\textup{2D-}\linearcls/ = \ppadpls/$.
\end{theorem}

\noindent Note that, just as for \cls/, the one-dimensional version can be solved in polynomial time, i.e., 1D-\linearcls/ = FP.\footnote{With the slight abuse of notation that FP $\subseteq$ TFNP, as explained in \cref{sec:def-tfnp}.} The containment 2D-\linearcls/ $\subseteq$ \ppadpls/ immediately follows from the fact that 2D-\linearcls/ $\subseteq$ \cls/ $\subseteq$ \ppadpls/. The other, more interesting, containment in \cref{thm:linear-cls} can be proved by directly reducing 2D-\clo/ to 2D-\linclo/. This reduction mainly relies on a more general result which says that any arithmetic circuit can be arbitrarily well approximated by a linear arithmetic circuit on a bounded domain. This approximation theorem (\cref{thm:lin-circuit-simple}) is stated and proved in \cref{sec:approx-linear-circuit}. The proof uses known techniques developed in the study of the complexity of Nash equilibria \citep{DGP09,CDT09}, but replaces the usual averaging step by a median step, which ensures that we obtain the desired accuracy of approximation.

Instead of reducing 2D-\clo/ to 2D-\linclo/, we prove \cref{thm:linear-cls} by a different route that also allows us to introduce a problem which might be of independent interest. To capture the cases where the gradient is not available or perhaps too expensive to compute, we consider a version of Gradient Descent where the \emph{finite differences} approach is used to compute an approximate gradient, which is then used as usual to obtain the next iterate. Formally, given a finite difference spacing parameter $h > 0$, the approximate gradient $\widetilde{\nabla}_h f(x)$ at some point $x \in [0,1]^n$ is computed as
$$\left[\widetilde{\nabla}_h f(x)\right]_i = \frac{f(x + h \cdot e_i) - f(x - h \cdot e_i)}{2h}$$
for all $i \in [n]$. The computational problem is defined as follows. Note that even though we define the problem on the domain $[0,1]^n$, it can be defined on more general domains as in our other problems.

\begin{tcolorbox}[breakable,enhanced]
	\begin{definition}
		\gdfd/:
		
		\noindent\textbf{Input}:
		\begin{itemize}
			\item precision/stopping parameter $\varepsilon > 0$,
			\item step size $\eta > 0$,
			\item finite difference spacing parameter $h > 0$,
			\item linear arithmetic circuit $f: \mathbb{R}^n \to \mathbb{R}$.
		\end{itemize}
		
		\noindent\textbf{Goal}: Compute any point where (projected) gradient descent for $f$ on domain $D=[0,1]^n$ using finite differences to approximate the gradient and fixed step size $\eta$ terminates. Formally, find $x \in [0,1]^n$ such that
		$$f(\Pi_D(x - \eta \widetilde{\nabla}_h f(x))) \geq f(x) - \varepsilon$$
		where for all $i \in [n]$
		$$\left[\widetilde{\nabla}_h f(x)\right]_i = \frac{f(x + h \cdot e_i) - f(x - h \cdot e_i)}{2h}.$$
	\end{definition}
\end{tcolorbox}

\noindent \gdfd/ immediately reduces to \linclo/ by setting $p := f$ and $g := \Pi_D(x - \eta \widetilde{\nabla}_h f(x))$. It is easy to construct a linear arithmetic circuit computing $g$, given a linear arithmetic circuit computing $f$. Note, in particular, that the projection $\Pi_D$ can be computed by a linear circuit since $D=[0,1]^n$. Indeed, $[\Pi_D(x)]_i = \min \{1, \max\{0,x_i\}\}$ for all $i \in [n]$ and $x \in \mathbb{R}^n$. Finally, the restriction of the codomain of $p$ to $[0,1]$ can be handled exactly as in the proof of \cref{lem:clo-codomain}.

In particular, the reduction from \gdfd/ to \linclo/ is domain-preserving and thus \cref{thm:linear-cls} immediately follows from the following theorem.

\begin{theorem}\label{thm:gd-finite-diff}
	\gdfd/ is \ppadpls/-complete, even with fixed domain $[0,1]^2$.
\end{theorem}

\noindent This result is interesting by itself, because the problem \gdfd/ is arguably quite natural, but also because it is the first problem that is complete for \ppadpls/ (and \cls/) that has a \emph{single} arithmetic circuit in the input. Note that our other problems which we prove to be \ppadpls/-complete, as well as the previously known \cls/-complete problems, all have two arithmetic circuits in the input.

\begin{proof}
	As explained above, \gdfd/ immediately reduces to \linclo/ and thus to \gclo/, which lies in \ppadpls/ by \cref{prop:gclo2ppadpls}. Thus, it remains to show that \gdfd/ is \ppadpls/-hard when we fix $n=2$. This is achieved by reducing from \gdls/ on domain $[0,1]^2$, which is \ppadpls/-hard by \cref{thm:all-ppadpls-complete}. In fact, we can even simplify the reduction by only considering \gdls/ instances that have some additional structure, but remain \ppadpls/-hard. Namely, consider an instance $(\varepsilon, \eta, f, \nabla f, L)$ of \gdls/ on domain $D=[0,1]^2$ such that:
	\begin{itemize}
		\item $\nabla f$ is the gradient of $f$,
		\item $f$ and $\nabla f$ are $L$-Lipschitz-continuous on $[-1,2]^2$.
	\end{itemize}
	To see that the problem remains \ppadpls/-hard even with these restrictions, note that the restrictions are satisfied by the hard instances constructed for the \kkt/ problem in the proof of \cref{thm:main-kkt-hard}, and that the reduction from \kkt/ to \gdls/ in \cref{prop:kkt2gdls} also trivially preserves them. In particular, even though the proof of \cref{thm:main-kkt-hard} only mentions that $f$ and $\nabla f$ are $L$-Lipschitz-continuous on $[0,1]^2$, the same arguments also show that they are $L$-Lipschitz-continuous on $[-1,2]^2$ (where $L$ has been scaled by some fixed constant).
	
	Let us now reduce from the instance $(\varepsilon, \eta, f, \nabla f, L)$ of \gdls/ to \gdfd/. We construct the instance $(\varepsilon', \eta, h, F)$ of \gdfd/ where $\varepsilon' = \varepsilon/4$, $h = \min\{1, \frac{\varepsilon}{8 \eta L^2}\}$ and $F$ is a linear arithmetic circuit that is obtained as follows. Let $\delta = \min\{\varepsilon/4, Lh^2/2\}$. By \cref{thm:lin-circuit-simple} and \cref{rem:approx-linear-domain}, we can construct a linear arithmetic circuit $F: \mathbb{R}^2 \to \mathbb{R}$ in polynomial time in $\sz(f)$, $\log L$ and $\log (1/\delta)$ such that $|f(x)-F(x)| \leq \delta$ for all $x \in [-1,2]^2$. Note that the second possibility in \cref{thm:lin-circuit-simple} cannot occur, since $f$ is guaranteed to be $L$-Lipschitz-continuous on $[-1,2]^2$.
	
	Consider any solution of that instance of \gdfd/, i.e., a point $x \in [0,1]^2$ such that $F(\Pi_D(x - \eta \widetilde{\nabla}_h F(x))) \geq F(x) - \varepsilon/4$. Let us show that $x$ is a solution to the original \gdls/ instance, i.e., that $f(\Pi_D(x - \eta \nabla f(x))) \geq f(x) - \varepsilon$.
	
	We have that for $i \in \{1,2\}$
	\begin{align*}
	&\quad \biggl| \bigl[\widetilde{\nabla}_h f(x)\bigr]_i - \bigl[\nabla f(x)\bigr]_i \biggr|\\
	&= \left|\frac{f(x + h \cdot e_i) - f(x - h \cdot e_i)}{2h} - \bigl[\nabla f(x)\bigr]_i\right|\\
	&\leq \frac{1}{2h} \biggl( \Bigl|f(x + h \cdot e_i) - f(x) - h\bigl[\nabla f(x)\bigr]_i\Bigr| + \Bigl|- f(x - h \cdot e_i) + f(x) - h\bigl[\nabla f(x)\bigr]_i\Bigr| \biggr)\\
	&= \frac{1}{2h} \biggl( \Bigl|f(x + h \cdot e_i) - f(x) - \bigl\langle \nabla f(x), (x + h \cdot e_i) - x \bigr\rangle\Bigr|\\
	&\quad + \Bigl|- f(x - h \cdot e_i) + f(x) + \bigl\langle \nabla f(x), (x - h \cdot e_i) - x \bigr\rangle\Bigr| \biggr)\\
	&\leq \frac{1}{2h} \left(\frac{L}{2} \bigl\|h \cdot e_i\bigr\|^2 + \frac{L}{2} \bigl\|-h \cdot e_i\bigr\|^2\right) = \frac{Lh}{2}
	\end{align*}
	where we used Taylor's theorem (\cref{lem:taylor}). Note that $x \pm h \cdot e_i \in [-1,2]^2$, since $h \leq 1$. Furthermore, it is easy to see that $\bigl|[\widetilde{\nabla}_h F(x)]_i - [\widetilde{\nabla}_h f(x)]_i\bigr| \leq \delta/h$, since $F$ approximates $f$ up to error $\delta$ on all of $[-1,2]^2$. It follows that $\bigl\|\widetilde{\nabla}_h F(x) - \nabla f(x)\bigr\| \leq \sqrt{2} (\delta/h + Lh/2) \leq 2Lh$. From this it follows that
	\begin{align*}
	&\quad \biggl|f\Bigl(\Pi_D\bigl(x - \eta \nabla f(x)\bigr)\Bigr) - f\Bigl(\Pi_D\bigl(x - \eta \widetilde{\nabla}_h F(x)\bigr)\Bigr)\biggr|\\
	&\leq L \cdot \Bigl\|\Pi_D\bigl(x - \eta \nabla f(x)\bigr) - \Pi_D\bigl(x - \eta \widetilde{\nabla}_h F(x)\bigr)\Bigr\|\\
	&\leq L \cdot \Bigl\|\bigl(x - \eta \nabla f(x)\bigr) - \bigl(x - \eta \widetilde{\nabla}_h F(x)\bigr)\Bigr\|\\
	&\leq \eta L \cdot \bigl\|\widetilde{\nabla}_h F(x) - \nabla f(x)\bigr\|\\
	&\leq 2 \eta L^2 h \leq \varepsilon/4.
	\end{align*}
	Finally, note that $|f(x) - F(x)| \leq \delta \leq \varepsilon/4$ and
	$$\biggl|f\Bigl(\Pi_D\bigl(x - \eta \widetilde{\nabla}_h F(x)\bigr)\Bigr) - F\Bigl(\Pi_D\bigl(x - \eta \widetilde{\nabla}_h F(x)\bigr)\Bigr)\biggr| \leq \delta \leq \varepsilon/4.$$
	Thus, since $F\Bigl(\Pi_D\bigl(x - \eta \widetilde{\nabla}_h F(x)\bigr)\Bigr) \geq F(x) - \varepsilon/4$, it follows that
	$$f\Bigl(\Pi_D\bigl(x - \eta \nabla f(x)\bigr)\Bigr) \geq f(x) - 4\varepsilon/4$$
	i.e., $x$ is a solution to the original \gdls/ instance.
\end{proof}

\section{Future directions}\label{sec:future-directions}

Our results may help to identify the complexity of the following problems that are known to lie in \ppadpls/:
\begin{itemize}
	\item \textsc{Mixed-Congestion}: The problem of finding a \emph{mixed} Nash
	equilibrium of a congestion game. It is known that finding a \emph{pure} Nash
	equilibrium is \pls/-complete \citep{FPT04}. As mentioned in \cref{sec:intro-results}, \citet{BabR21} have recently applied our main result to obtain \ppadpls/-completeness for \textsc{Mixed-Congestion}. It would be interesting to extend this to \emph{network} congestion games, where the strategies are represented implicitly.
	\item \textsc{polynomial-KKT}: The special case of the \kkt/ problem where
	the function is a polynomial, provided explicitly in the input (exponents in
	unary). In particular, note that this version of the problem does not require the introduction of violation-solutions for Lipschitzness and smoothness. A consequence of the above-mentioned reduction by
	\citet{BabR21} is that the problem is \ppadpls/-complete for polynomials of
	degree 5. It is an interesting open problem to extend this hardness result to lower degree polynomials.
	\item \textsc{Contraction}: Find a fixed point of a function that is contracting with respect to some $\ell_p$-norm.
	\item \textsc{Tarski}: Find a fixed point of an order-preserving function, as guaranteed by Tarski's theorem \citep{EtessamiPRY20,FearnleyPS22-Tarski-algo,DangQY20}.
	\item \textsc{ColorfulCarath{\'e}odory}: A problem based on a theorem in convex geometry \citep{MMSS2017caratheodory}.
\end{itemize}
The first three problems on this list were known to lie in \cls/ \citep{DaskalakisP2011-CLS}, while the other two were only known to lie in \ppadpls/.

The collapse between \cls/ and \ppadpls/ raises the question of whether the class EOPL (for End of Potential Line), a subclass of \cls/, is also equal to \ppadpls/. The class EOPL, or more precisely its subclass UEOPL (with U for unique), is known to contain various problems of interest that have unique solutions such as Unique Sink Orientation (USO), the P-matrix Linear Complementarity Problem (P-LCP), Simple Stochastic Games (SSG) and Parity Games \citep{FGMS2020-UEOPL}.

In an earlier version of this paper, we conjectured that EOPL\,$\neq$\,\ppadpls/, but, in another surprising turn of events, subsequent work by \citet{GoosHJMPRT22-collapses} has shown that in fact EOPL\,$=$\,\ppadpls/. This collapse provides further evidence that \ppadpls/ is a natural and robust class. Furthermore, together with the work of \citet{HubacekY2017-CLS}, who prove that EOPL\,$\subseteq$\,\cls/, it also yields an alternative proof of $\cls/ = \ppadpls/$. Finally, regarding our results about the gradient descent problem, the proof of \ppadpls/-hardness can now be significantly simplified by reducing from \textsc{End-of-Potential-Line}, the standard EOPL-complete problem, instead of \textsc{Either-Solution(\eol/,\iter/)}. In particular, ``green paths'' are enough to embed an \textsc{End-of-Potential-Line} instance, and the ``orange paths'' and the \pls/-Labyrinth are no longer needed.

The collapse EOPL\,$=$\,\ppadpls/ leaves open the question of whether UEOPL is also equal to \ppadpls/. We conjecture that UEOPL\,$\neq$\,\ppadpls/. The canonical complete problem for UEOPL, \textsc{Unique-End-of-Potential-Line}, is the same as \textsc{End-of-Potential-Line}, except that it also allows an additional kind of violation solution. We currently see no way of reducing \textsc{End-of-Potential-Line} to \textsc{Unique-End-of-Potential-Line}, because this would require a way to extract a standard solution from one of those new violation solutions. In order to provide evidence that the classes are not equal, one could try to obtain an oracle separation between UEOPL and \ppadpls/, in the sense of \citet{BCEIP98}.

\citet{Ishizuka21} has recently shown that the definition of EOPL is robust with respect to some modifications (similarly to \ppad/ \citep{GH19}), and has provided a somewhat artificial problem that is complete for $\textup{PPA} \cap \textup{PLS}$. This raises the interesting question of whether $\textup{PPA} \cap \textup{PLS}$, and other intersections of well-studied classes, also admit natural complete problems, or if \ppadpls/ is in fact an isolated case.

\appendix

\section{More on arithmetic circuits}\label{app:more-circuits}

\subsection{Evaluation of well-behaved arithmetic circuits (Proof of \texorpdfstring{\cref{lem:well-behaved-efficient}}{Lemma~\ref*{lem:well-behaved-efficient}})}

We restate the Lemma here for convenience.

\theoremstyle{plain}
\newtheorem*{lem:appendix:well-behaved-efficient}{Lemma \ref*{lem:well-behaved-efficient}}
\begin{lem:appendix:well-behaved-efficient}
	Let $f$ be a well-behaved arithmetic circuit with $n$ inputs. Then, for any rational $x \in \mathbb{R}^n$, $f(x)$ can be computed in time $\textup{poly}(\sz(f),\sz(x))$.
\end{lem:appendix:well-behaved-efficient}

\begin{proof}
	Recall that an arithmetic circuit $f$ is well-behaved if, on any directed path that leads to an output, there are at most $\log (\sz(f))$ true multiplication gates. Without loss of generality, we can assume that the circuit $f$ only contains gates that are used to compute at least one of the outputs.
	
	Let $x$ denote the input to circuit $f$ and for any gate $g$ of $f$ let $v(g)$ denote the value computed by gate $g$ when $x$ is provided as input to the circuit. For any gate $g$ that is not an input gate or a constant gate, let $g_1$ and $g_2$ denote the two gates it uses as inputs. Clearly, if $g$ is one of $\{+,-,\times,\max,\min,>\}$, $v(g)$ can be computed in polynomial time in $\sz(v(g_1)) + \sz(v(g_2))$, including transforming it into an irreducible fraction. Thus, in order to show that the circuit can be evaluated in polynomial time, it suffices to show that for all gates $g$ of $f$, it holds that $\sz(v(g)) \leq p(\sz(f) + \sz(x))$, where $p$ is some fixed polynomial (independent of $f$ and $x$). In the rest of this proof, we show that
	$$\sz(v(g)) \leq 6 \cdot \sz(f)^3 \cdot \sz(x).$$
	
	It is convenient to partition the gates of the circuit depending on their depth. For any gate $g$ in $f$, we let $d(g)$ denote the depth of the gate in $f$. The input gates and the constant gates are at depth $1$. For any other gate $g$, we define its depth inductively as $d(g) = 1 + \max\{d(g_1),d(g_2)\}$, where $g_1$ and $g_2$ are the two input gates of $g$. Note that $d(g) \leq \sz(f)$ for all gates $g$ in the circuit.
	
	We also define a notion of ``multiplication-depth'' $md(g)$. The gates $g$ at depth $1$ all have $md(g)=0$. For the rest of the gates, the multiplication-depth is defined inductively. For a gate $g$ whose inputs are $g_1$ and $g_2$, we let $md(g) = 1 + \max \{md(g_1), md(g_2)\}$ if $g$ is a true multiplication gate, and $md(g) = \max \{md(g_1), md(g_2)\}$ otherwise. Since $f$ is well-behaved, it immediately follows that $md(g) \leq \log (\sz(f))$ for all gates $g$ of the circuit.
	
	We begin by showing that for any gate $g$ of $f$, it holds that $|v(g)| \leq 2^{\sz(f)^2(\sz(x)+\sz(f))}$.
	This follows from the stronger statement that
	$$|v(g)| \leq 2^{d(g) \cdot 2^{md(g)} \cdot (\sz(x)+\sz(f))},$$
	which we prove by induction as follows. First of all, note that any gate at depth $1$ satisfies the statement, since any input or constant of the circuit is bounded by $2^{\sz(x)}$ or $2^{\sz(f)}$ respectively. Next, assume that the statement holds for all gates with depth $\leq k-1$ and consider some gate $g$ at depth $k$. Let $g_1$ and $g_2$ denote its two inputs, which must satisfy that $d(g_1) \leq k-1$ and $d(g_2) \leq k-1$. If $g$ is one of $\{\min, \max, <\}$, then the statement immediately also holds for $g$. If $g$ is an addition or subtraction gate, then $|v(g)| \leq |v(g_1)| + |v(g_2)| \leq 2 \max \{|v(g_1)|,|v(g_2)|\}$, which implies that the statement also hold for $g$, since $d(g_1), d(g_2) \leq k-1$ and $d(g) = k$. If $g$ is a multiplication by a constant, then $|v(g)| \leq 2^{\sz(f)} |v(g_1)|$ (wlog $g_2$ is the constant), and the statement holds for $g$ too. Finally, if $g$ is a true multiplication gate, then $|v(g)| = |v(g_1)| |v(g_2)| \leq (\max \{|v(g_1)|,|v(g_2)|\})^2$. Since $md(g) = 1 + \max\{md(g_1), md(g_2)\}$, it follows that the statement also holds for $g$.
	
	Let $den(g)$ denote the absolute value of the denominator of $v(g)$ (written as an irreducible fraction). We show that for all gates $g$, it holds that $den(g) \leq 2^{\sz(f)^2(\sz(x)+\sz(f))}$. This is enough to conclude our proof. Indeed, since we also have that $|v(g)| \leq 2^{\sz(f)^2(\sz(x)+\sz(f))}$, it follows that the absolute value of the numerator of $v(g)$ is
	$$|v(g)| \cdot den(g) \leq 2^{2 \cdot \sz(f)^2(\sz(x)+\sz(f))}.$$
	As a result, it follows that
	$$\sz(v(g)) \leq 2 \cdot \sz(f)^2(\sz(x)+\sz(f)) + \sz(f)^2(\sz(x)+\sz(f)) \leq 6 \cdot \sz(f)^3 \cdot \sz(x).$$
	It remains to show that $den(g) \leq 2^{\sz(f)^2(\sz(x)+\sz(f))}$, which we prove by showing that
	\[den(g) \leq 2^{d(g) \cdot 2^{md(g)} \cdot (\sz(x)+\sz(f))}.\]
	Let $M$ denote (the absolute value of) the product of all denominators appearing in the input $x$ and the description of $f$, i.e., the denominators of the coordinates of the input $x$, and the denominators of the constants used by $f$. Note that $M \leq 2^{\sz(x)+\sz(f)}$. We prove by induction that for all gates $g$,
	$$\text{$den(g)$ is a factor of $M^{d(g) \cdot 2^{md(g)}}$}$$
	which in particular implies the bound on $den(g)$ above. First of all, note that any gate at depth $1$ is an input or a constant, and thus satisfies the statement. Next, assume that the statement holds for all gates with depth $\leq k-1$ and consider some gate $g$ at depth $k$. Let $g_1$ and $g_2$ denote its two inputs, which must satisfy that $d(g_1) \leq k-1$ and $d(g_2) \leq k-1$. If $g$ is one of $\{\min, \max, <\}$, then it is easy to see that the statement immediately also holds for $g$. If $g$ is an addition or subtraction gate, then, since $v(g_1)$ and $v(g_2)$ can both be expressed as fractions with denominator $M^{d(g) \cdot 2^{md(g)}}$, so can $v(g)$, and the statement also holds for $g$. If $g$ is a multiplication by a constant, then $den(g)$ is a factor of $M^{d(g) \cdot 2^{md(g)}}$, since $den(g_1)$ is a factor of $M^{(d(g)-1) \cdot 2^{md(g)}}$ and the denominator of the constant is a factor of $M$ (wlog assume that $g_2$ is the constant). Finally, if $g$ is a true multiplication gate, then $den(g_1)$ and $den(g_2)$ are factors of $M^{d(g) \cdot 2^{md(g)-1}}$, and thus $den(g)$ is a factor of $(M^{d(g) \cdot 2^{md(g)-1}})^2 = M^{d(g) \cdot 2^{md(g)}}$ as desired.
\end{proof}

\subsection{Linear arithmetic circuits are Lipschitz-continuous}

Linear arithmetic circuits are only allowed to use the operations $\{+,-,\max,\min, \times \zeta\}$ and rational constants. The operation $\times \zeta$ denotes multiplication by a constant (which is part of the description of the circuit). Every linear arithmetic circuit is in particular a well-behaved arithmetic circuit, and so, by \cref{lem:well-behaved-efficient}, can be evaluated in polynomial time. Furthermore, every linear arithmetic circuit represents a Lipschitz-continuous function such that the Lipschitz constant has polynomial bit-size with respect to the size of the circuit.

\begin{lemma}\label{lem:linear-lipschitz}
	Any linear arithmetic circuit $f: \mathbb{R}^n \to \mathbb{R}^m$ is $2^{\sz(f)^2}$-Lipschitz-continuous (w.r.t. the $\ell_\infty$-norm) over $\mathbb{R}^n$.
\end{lemma}

\begin{proof}
	For any gate $g$ of the circuit $f$, let $L(g)$ denote the Lipschitz-constant of the function which outputs the value of $g$, given the input $x$ to the circuit. As in the proof of \cref{lem:well-behaved-efficient}, it is convenient to partition the gates of $f$ according to their depth. Note that for all the gates $g$ at depth $1$, i.e., the input gates and the constant gates, it holds that $L(g) \leq 1$. We show that any gate $g$ at depth $k$ satisfies $L(g) \leq 2^{k \cdot \sz(f)}$. It immediately follows from this that $f$ is $2^{\sz(f)^2}$-Lipschitz-continuous (w.r.t. the $\ell_\infty$-norm) over $\mathbb{R}^n$.
	
	Consider a gate $g$ at depth $k$ with inputs $g_1$ and $g_2$ (which lie at a lower depth). If $g$ is $+$ or $-$, then $L(g) \leq L(g_1) + L(g_2) \leq 2 \max\{L(g_1),L(g_2)\} \leq 2 \cdot 2^{(k-1) \cdot \sz(f)} \leq 2^{k \cdot \sz(f)}$. If $g$ is $\max$ or $\min$, then it is easy to see that $L(g) \leq \max\{L(g_1),L(g_2)\} \leq 2^{k \cdot \sz(f)}$. Finally, if $g$ is $\times \zeta$, then $L(g) \leq |\zeta| \cdot L(g_1) \leq 2^{\sz(f)} 2^{(k-1) \cdot \sz(f)} = 2^{k \cdot \sz(f)}$, where we used the fact that $|\zeta| \leq 2^{\sz(f)}$.
\end{proof}

\section{Mathematical tools}

\subsection{Tools from convex analysis and a generalization of Farkas' Lemma}\label{app:convex}

Let $D \subseteq \mathbb{R}^n$ be a non-empty closed convex set. Recall that the projection $\Pi_D: \mathbb{R}^n \to D$ is defined by $\Pi_D(x) = \argmin_{y \in D} \|x-y\|$, where $\|\cdot\|$ denotes the Euclidean norm. It is known that $\Pi_D(x)$ always exists and is unique. The following two results are standard tools in convex analysis, see, e.g., \citep{bertsekas1998nonlinear}.

\begin{lemma}\label{lem:projection-lemma}
	Let $D$ be a non-empty closed convex set in $\mathbb{R}^n$ and let $y \in \mathbb{R}^n$. Then for all $x \in D$ it holds that
	$$\langle y - \Pi_D(y), x - \Pi_D(y) \rangle \leq 0.$$
\end{lemma}

\begin{proposition}\label{prop:strict-separation}
	Let $D_1$ and $D_2$ be two disjoint non-empty closed convex sets in $\mathbb{R}^n$ and such that $D_2$ is bounded. Then, there exist $c \in \mathbb{R}^n \setminus \{0\}$ and $d \in \mathbb{R}$ such that $\langle c, x \rangle < d$ for all $x \in D_1$, and $\langle c, x \rangle > d$ for all $x \in D_2$.
\end{proposition}

\noindent We will need the following generalization of Farkas' Lemma, which we prove below. For $\varepsilon = 0$, we recover the usual statement of Farkas' Lemma.

\begin{lemma}\label{lem:new-farkas}
	Let $A \in \mathbb{R}^{m \times n}$, $b \in \mathbb{R}^n$ and $\varepsilon \geq 0$. Then exactly one of the following two statements holds:
	\begin{enumerate}
		\item $\exists x \in \mathbb{R}^n: Ax \leq 0$, $\langle b, x \rangle > \varepsilon \|x\|$,
		\item $\exists y \in \mathbb{R}^m: \|A^Ty - b\| \leq \varepsilon$, $y \geq 0$.
	\end{enumerate}
\end{lemma}

\begin{proof}
	Let us first check that both statements cannot hold at the same time. Indeed, if this were the case, then we would obtain the following contradiction
	$$\varepsilon \|x\| < \langle b, x \rangle = \langle A^T y, x \rangle + \langle b - A^Ty, x \rangle \leq \langle y, Ax \rangle + \|b - A^Ty\| \|x\| \leq \varepsilon \|x\|$$
	where we used the fact that $\langle y, Ax \rangle \leq 0$ and the Cauchy-Schwarz inequality.
	
	Now, let us show that if statement 2 does not hold, then statement 1 must necessarily hold. Let $D_1 = \{A^Ty : y \geq 0\}$ and $D_2 = \{x : \|x-b\| \leq \varepsilon\}$. Note that since statement 2 does not hold, it follows that $D_1$ and $D_2$ are disjoint. Furthermore, it is easy to check that $D_1$ and $D_2$ satisfy the conditions of \cref{prop:strict-separation}. Thus, there exist $c \in \mathbb{R}^n \setminus \{0\}$ and $d \in \mathbb{R}$ such that $\langle c, x \rangle < d$ for all $x \in D_1$, and $\langle c, x \rangle > d$ for all $x \in D_2$. In particular, we have that for all $y \geq 0$, $\langle Ac, y \rangle = \langle c, A^Ty \rangle < d$. From this it follows that $Ac \leq 0$, since if $[Ac]_i > 0$ for some $i$, then $\langle Ac, y \rangle \geq d$ for $y = \frac{|d|}{[Ac]_i} e_i$.
	
	In order to show that $x := c$ satisfies the first statement, it remains to prove that $\langle b, c \rangle > \varepsilon \|c\|$. Note that by setting $y=0$, we get that $0 = \langle c, A^T0 \rangle < d$. Let $z = b - \varepsilon \frac{c}{\|c\|}$. Since $z \in D_2$, it follows that $\langle c, z \rangle > d > 0$. Since $\langle c, z \rangle = \langle c, b \rangle - \varepsilon \|x\|$, statement 1 indeed holds.
\end{proof}

\subsection{Proof of \texorpdfstring{\cref{lem:taylor}}{Lemma~\ref*{lem:taylor}} (Taylor's Theorem)}\label{app:taylor}

We restate the Lemma here for convenience.

\theoremstyle{plain}
\newtheorem*{lem:appendix:taylor}{Lemma \ref*{lem:taylor}}
\begin{lem:appendix:taylor}[Taylor's theorem]
	Let $f: \mathbb{R}^n \to \mathbb{R}$ be continuously differentiable and let $D \subseteq \mathbb{R}^n$ be convex. If $\nabla f$ is $L$-Lipschitz-continuous (w.r.t.\ the $\ell_2$-norm) on $D$, then for all $x,y \in D$ we have
	$$\bigl|f(y) - f(x) - \langle \nabla f(x), y-x \rangle \bigr| \leq \frac{L}{2} \|y-x\|^2.$$
\end{lem:appendix:taylor}

\begin{proof}
	Let $g : [0,1] \to \mathbb{R}$ be defined by $g(t) = f(x + t(y-x))$. Then, $g$ is continuously differentiable on $[0,1]$ and $g'(t) = \langle \nabla f(x + t(y-x)), y-x \rangle$. Furthermore, $g'$ is $(L \|x-y\|^2)$-Lipschitz-continuous on $[0,1]$, since
	\begin{equation*}
	\begin{split}
	|g'(t_1) - g'(t_2)| &= \left| \langle \nabla f(x + t_1(y-x)) - \nabla f(x + t_2(y-x)), y-x \rangle \right|\\
	&\leq \|\nabla f(x + t_1(y-x)) - \nabla f(x + t_2(y-x))\| \cdot \|y-x\|\\
	&\leq L \cdot \|t_1(y-x) - t_2(y-x)\| \cdot \|y-x\|\\
	&\leq L \cdot |t_1-t_2| \cdot \|y-x\|^2
	\end{split}
	\end{equation*}
	where we used the Cauchy-Schwarz inequality. We also used the fact that $\nabla f$ is $L$-Lipschitz on $D$, and $x + t(y-x) \in D$ for all $t \in [0,1]$.
	Now, we can write
	$$f(y) - f(x) - \langle \nabla f(x), y-x \rangle = g(1) - g(0) - g'(0) = \int_0^1 (g'(t)-g'(0)) \, \mathrm{d} t$$
	and thus
	\[\left|f(y) - f(x) - \langle \nabla f(x), y-x \rangle\right| \leq \int_0^1 |g'(t)-g'(0)| \, \mathrm{d} t \leq \int_0^1 L \cdot \|x-y\|^2 \cdot |t| \, \mathrm{d} t = \frac{L}{2} \|y-x\|^2. \qedhere\]
\end{proof}

\section{Minor observations on \clo/ and \kkt/}

\begin{lemma}\label{lem:clo-dimension}
	For all integers $k_2 > k_1 > 0$, $k_1$\textup{D}-\clo/ reduces to $k_2$\textup{D}-\clo/ using a promise-preserving reduction.
\end{lemma}

\begin{proof}
	For $x \in \mathbb{R}^{k_2}$, we write $x = (x_1,x_2)$, where $x_1 \in \mathbb{R}^{k_1}$ and $x_2 \in \mathbb{R}^{k_2-k_1}$. Let $(\varepsilon, p, g, L)$ be an instance of $k_1$D-\clo/. The reduction constructs the instance $(\varepsilon, p', g', L)$ of $k_2$D-\clo/, where
	$$p'(x) = p'(x_1,x_2) = p(x_1) \qquad \text{and} \qquad g'(x) = g'(x_1,x_2) = (g(x_1),0).$$
	Clearly, the arithmetic circuits for $p'$ and $g'$ can be constructed in polynomial time and are well-behaved.
	
	Since $|p'(x) - p'(y)| = |p(x_1) - p(y_1)| \leq L \|x_1 - y_1\| \leq L \|x-y\|$, it is clear that any violation $x,y \in [0,1]^{k_2}$ of $L$-Lipschitzness for $p'$ also yields a violation $x_1,y_1 \in [0,1]^{k_1}$ for $p$. Similarly, since
	$$\|g'(x) - g'(y)\| = \|(g(x_1),0) - (g(y_1),0)\| = \|g(x_1) - g(y_1)\| \leq L \|x_1 - y_1\| \leq L \|x-y\|,$$
	any violation $x,y$ of $L$-Lipschitzness for $g'$ also yields a violation $x_1,y_1$ for $g$. Thus, any violation of the constructed instance is always mapped back to a violation of the original instance, and the reduction is indeed promise-preserving.
	
	Finally, note that any proper solution $x \in [0,1]^{k_2}$ of the constructed instance, i.e., such that $p'(g'(x)) \geq p'(x) - \varepsilon$, immediately yields a solution $x_1 \in [0,1]^{k_1}$ of the original instance.
\end{proof}

\begin{lemma}\label{lem:clo-codomain}
	\clo/ with codomain $[0,1]$ for function $p$ is equivalent to \clo/ without this restriction.
\end{lemma}

\begin{proof}
	It is clear that the version with the restriction trivially reduces to the version without the restriction. Thus, it remains to show the other direction, namely that \clo/ without the codomain restriction reduces to the restricted version.
	
	Let $(\varepsilon, p, g, L)$ be an instance of \clo/ with domain $[0,1]^n$ and without a codomain restriction for $p$. The reduction constructs the instance $(\varepsilon', p', g, L')$ of \clo/ with domain $[0,1]^n$, where $\varepsilon' = \frac{\varepsilon}{2nL}$, $L' = \max \{L, \frac{1}{2n}\}$ and
	$$p'(x) = \min \left\{1, \max \left\{0, \frac{1}{2} + \frac{p(x) - p(z_c)}{2nL} \right\} \right\}$$
	where $z_c = (1/2, 1/2, \dots, 1/2)$ is the centre of $[0,1]^n$. Note that the arithmetic circuit computing $p'$ can be computed in polynomial time given the circuit for $p$, and that the modification of $p$ will require using gates $\times \zeta$, but no general multiplication gates. Thus, the circuit for $p'$ is also well-behaved. Note, in particular, that the value $p(z_c)$ can be computed in polynomial time in the size of the circuit for $p$. It follows that the reduction can be computed in polynomial time.
	
	First of all, let us show that any point $x \in [0,1]^n$ such that $p'(x) \neq \frac{1}{2} + \frac{p(x) - p(z_c)}{2nL}$ will immediately yield a violation of the $L$-Lipschitzness of $p$. Indeed, if $x$ and $z_c$ satisfy the $L$-Lipschitzness of $p$, then this means that
	$$|p(x) - p(z_c)| \leq L \|x - z_c\| \leq n L$$
	since $x, z_c \in [0,1]^n$. As a result, it follows that $\frac{1}{2} + \frac{p(x) - p(z_c)}{2nL} \in [0,1]$ and thus $p'(x) = \frac{1}{2} + \frac{p(x) - p(z_c)}{2nL}$.
	
	In the rest of this proof we assume that we always have $p'(x) = \frac{1}{2} + \frac{p(x) - p(z_c)}{2nL}$, since we can immediately extract a violation if we ever come across a point $x$ where this does not hold. Let us now show that any solution of the constructed instance immediately yields a solution of the original instance. Clearly, any violation of the $L'$-Lipschitzness of $g$ is trivially also a violation of $L$-Lipschitzness.
	
	Next assume that $x,y \in [0,1]^n$ are a violation of $L'$-Lipschitzness of $p'$. Let us show by contradiction that $x,y$ must be a violation of $L$-Lipschitzness for $p$. Indeed, assume that $x,y$ satisfy the $L$-Lipschitzness for $p$, then
	$$|p'(x)-p'(y)| = \left| \frac{p(x) - p(y)}{2nL} \right| \leq \frac{1}{2n} \|x-y\|$$
	which is a contradiction to $x,y$ being a violation of $L'$-Lipschitzness of $p'$.
	
	Finally, consider any proper solution of the constructed instance, i.e., $x \in [0,1]^n$ such that $p'(g(x)) \geq p'(x) - \varepsilon'$. Then it follows straightforwardly that $p(g(x)) \geq p(x) - 2nL \varepsilon'$, which implies that $x$ is a solution to the original instance, since $2nL \varepsilon' = \varepsilon$. Note that the reduction is also promise-preserving, since we always map violations of the constructed instance back to violations of the original instance.
\end{proof}

\begin{lemma}\label{lem:one-dimensional-easy}
	\gclo/ with fixed dimension $n=1$ can be solved in polynomial time. As a result, this also holds for \kkt/, \gdls/ and \gdfp/.
\end{lemma}

\begin{proof}
	This is a straightforward consequence of the fact that finding Brouwer fixed points in one dimension is easy. Consider any instance $(\varepsilon,A,b,p,g,L)$ of \gclo/ with $n=1$. It is easy to see that any $\varepsilon/L$-approximate fixed point of $x \mapsto \Pi_D(g(x))$ immediately yields a solution to the \gclo/ instance.
	
	Thus, we proceed as follows. First of all, from $A$ and $b$ we can directly determine $t_1,t_2 \in \mathbb{R}$ such that $D = [t_1,t_2]$. Note that the bit-size of $t_1$ and $t_2$ is polynomial in the input size. Then define a grid of points on the interval $[t_1,t_2]$ such that the distance between consecutive points is $\varepsilon/L^2$. Finally, using binary search, find two consecutive points $x_1$ and $x_2$ such that $g(x_1) \geq x_1$ and $g(x_2) \leq x_2$. One of these two points has to be an $\varepsilon/L$-approximate fixed point of $x \mapsto \Pi_D(g(x))$ (or we obtain a violation of Lipschitz-continuity). Binary search takes polynomial time, because the number of points is at most exponential in the input size.
	
	Since the other three problems reduce to \gclo/ using domain-preserving reductions (see \cref{sec:all-ppadpls-complete}), it follows that they can also be solved in polynomial time when $n=1$.
\end{proof}

\begin{lemma}\label{lem:kkt-poly-params}
	\kkt/ on domain $[0,1]^n$ can be solved in polynomial time in $1/\varepsilon$, $L$ and the sizes of the circuits for $f$ and $\nabla f$.
\end{lemma}

\begin{proof}
	This follows from the fact that the problem can be solved by Gradient Descent in polynomial time in those parameters.
	Let $(\varepsilon,f,\nabla f,L)$ be an instance of \kkt/ with domain $[0,1]^n$. First, compute $f(0)$ in polynomial time in $\sz(f)$. If $f$ is indeed $L$-Lipschitz-continuous, then it follows that $f(x) \in I = [f(0)-\sqrt{n}L, f(0)+\sqrt{n}L]$ for all $x \in [0,1]^n$. If we ever come across a point where this does not hold, we immediately obtain a violation of $L$-Lipschitz-continuity of $f$. So, for the rest of this proof we simply assume that $f(x) \in I$ for all $x \in [0,1]^n$.
	
	Note that the length of interval $I$ is $2\sqrt{n}L$, which is polynomial in $L$ and $n$. By using the reduction in the proof of \cref{prop:kkt2gdls}, we can solve our instance by solving the instance $(\varepsilon',\eta,f,\nabla f,L)$ of \gdls/, where $\varepsilon' = \frac{\varepsilon^2}{8L}$ and $\eta = \frac{1}{L}$. The important observation here is that this instance of \gdls/ can be solved by applying Gradient Descent with step size $\eta$ and with any starting point, in at most $\frac{|I|}{\varepsilon'}=\frac{16\sqrt{n}L^2}{\varepsilon^2}$ steps. Indeed, every step must improve the value of $f$ by $\varepsilon'$, otherwise we have found a solution. It is easy to see that each step of Gradient Descent can be done in polynomial time in $\sz(\nabla f)$, $n$ and $\log L$. Since the number of steps is polynomial in $1/\varepsilon$, $L$ and $n$, the problem can be solved in polynomial time in $1/\varepsilon$, $L$, $n$, $\sz(f)$ and $\sz(\nabla f)$. Finally note that $n \leq \sz(f)$ (because $f$ has $n$ input gates).
\end{proof}

\section{\gbrouwer/ and \grlo/}\label{sec:general-brouwer-rlo}

In this section we define the computational problems \gbrouwer/ and \grlo/ and prove that they are \ppad/- and \pls/-complete respectively. These two completeness results follow straightforwardly from prior work. The membership of \gbrouwer/ in \ppad/ and of \grlo/ in \pls/ are used in this paper to show that our problems of interest lie in \ppadpls/.

\begin{tcolorbox}[breakable,enhanced]
	\begin{definition}
		\gbrouwer/:
		
		\noindent\textbf{Input}:
		\begin{itemize}
			\item precision parameter $\varepsilon > 0$,
			\item $(A,b) \in \mathbb{R}^{m \times n} \times \mathbb{R}^m$ defining a bounded non-empty domain $D = \{x \in \mathbb{R}^n: Ax \leq b\}$,
			\item well-behaved arithmetic circuit $g: \mathbb{R}^n \to \mathbb{R}^n$,
			\item Lipschitz constant $L > 0$.
		\end{itemize}
		
		\noindent\textbf{Goal}: Compute an approximate fixed point of $g$ on domain $D$. Formally, find $x \in D$ such that
		$$\|\Pi_D(g(x)) - x\| \leq \varepsilon.$$
		Alternatively, we also accept a violation of $L$-Lipschitzness of $g$ as a solution. Namely, $x,y \in D$ such that $\|g(x) - g(y)\| > L \|x-y\|$.
	\end{definition}
\end{tcolorbox}

\begin{proposition}\label{prop:general-brouwer}
	\gbrouwer/ is \ppad/-complete.
\end{proposition}

\begin{proof}
	Various formulations and special cases of the problem of finding a Brouwer fixed point are known to be \ppad/-complete \citep{Pap94,ChenD2009-2D,EY10-Nash-FIXP}. The \ppad/-hardness of our \gbrouwer/ problem immediately follows from the \ppad/-hardness of the problem on the domain $[0,1]^2$ and when $g$ is a linear arithmetic circuit, which is known from \citep{Mehta2018-constant-rank}.
	
	The containment in \ppad/ essentially follows from Proposition~2 in \citep{EY10-Nash-FIXP}, where it is shown that finding an approximate fixed point of a Brouwer function that is efficiently computable and continuous, when the domain is a bounded polytope, is in \ppad/. In \gbrouwer/, the function is not guaranteed to be continuous, but instead we allow violations of Lipschitz-continuity as solutions. However, it can easily be seen that the proof in \citep{EY10-Nash-FIXP} also applies to this case. Alternatively, we can also use our \cref{thm:lin-circuit-simple} to approximate the circuit $g$ by a linear arithmetic circuit (which is necessarily Lipschitz-continuous with a polynomially representable Lipschitz-constant, see \cref{lem:well-behaved-efficient}) and then use \citep[Proposition 2]{EY10-Nash-FIXP} directly. Note that since $D$ is bounded, we can easily compute $M > 0$ such that $D \subseteq [-M,M]^n$ (using linear programming). Then, using \cref{thm:lin-circuit-simple} and \cref{rem:approx-linear-domain}, we can approximate $g$ by a linear arithmetic circuit on the domain $D$.
\end{proof}

\begin{tcolorbox}[breakable,enhanced]
	\begin{definition}
		\grlo/:
		
		\noindent\textbf{Input}:
		\begin{itemize}
			\item precision/stopping parameter $\varepsilon > 0$,
			\item $(A,b) \in \mathbb{R}^{m \times n} \times \mathbb{R}^m$ defining a bounded non-empty domain $D = \{x \in \mathbb{R}^n: Ax \leq b\}$,
			\item well-behaved arithmetic circuits $p: \mathbb{R}^n \to \mathbb{R}$ and $g: \mathbb{R}^n \to \mathbb{R}^n$,
			\item Lipschitz constant $L > 0$.
		\end{itemize}
		
		\noindent\textbf{Goal}: Compute an approximate local optimum of $p$ with respect to $g$ on domain $D$. Formally, find $x \in D$ such that
		$$p(\Pi_D(g(x))) \geq p(x) - \varepsilon.$$
		Alternatively, we also accept a violation of $L$-Lipschitzness of $p$ as a solution. Namely, $x,y \in D$ such that $|p(x) - p(y)| > L \|x-y\|$.
	\end{definition}
\end{tcolorbox}

\begin{proposition}\label{prop:general-rlo}
	\grlo/ is \pls/-complete.
\end{proposition}

\begin{proof}
	The \pls/-hardness of \grlo/ immediately follows from Theorem~2.1 in \citep{DaskalakisP2011-CLS}, where it is shown that the problem is \pls/-complete in the special case where the domain is $[0,1]^3$. The proof of membership in \pls/ for the domain $[0,1]^3$ immediately generalizes to $[0,1]^n$, even for non-fixed $n$. Thus, it remains to show that we can reduce \grlo/ to the special case where the domain is $[0,1]^n$.
	
	Note that since $D$ is bounded, we can easily compute $M > 0$ such that $D \subseteq [-M,M]^n$ (using linear programming). We extend $p$ and $g$ to the whole hypercube $[-M,M]^n$ by using the projection onto $D$, namely $\widehat{p}(x) = p(\Pi_D(x))$ and $\widehat{g}(x) = g(\Pi_D(x))$. Since $\|x-y\| \geq \|\Pi_D(x) - \Pi_D(y)\|$ for all $x,y \in \mathbb{R}^n$, it follows that any violation of $L$-Lipschitzness for $\widehat{p}$ immediately yields a violation for $p$. If $x \in [-M,M]^n$ is an $\varepsilon$-approximate local optimum of $\widehat{p}$ with respect to $\widehat{g}$, i.e., $\widehat{p}(\Pi_D(\widehat{g}(x))) \geq \widehat{p}(x) - \varepsilon$, then it immediately follows that $\Pi_D(x) \in D$ is an $\varepsilon$-approximate local optimum of $p$ with respect to $g$. Thus, we have reduced the problem to the case where the domain is a hypercube $[-M,M]^n$.
	
	The final step is to change the domain from $[-M,M]^n$ to $[0,1]^n$, which can easily be achieved by letting $\tilde{p}(x) = \widehat{p}(2M \cdot x - M \cdot \mathbf{e})$ and $\tilde{g}(x) = (\widehat{g}(2M \cdot x - M \cdot \mathbf{e}) + M \cdot \mathbf{e})/2M$. Here $\mathbf{e} \in \mathbb{R}^n$ denotes the all-ones vector. A violation of $2ML$-Lipschitzness for $\tilde{p}$ immediately yields a violation of $L$-Lipschitzness for $\widehat{p}$. Furthermore, if $x \in [0,1]^n$ is an $\varepsilon$-approximate local optimum of $\tilde{p}$ with respect to $\tilde{g}$, then it is easy to see that $(2M \cdot x - M \cdot \mathbf{e}) \in [-M,M]^n$ is an $\varepsilon$-approximate local optimum of $\widehat{p}$ with respect to $\widehat{g}$. We thus obtain an instance on the domain $[0,1]^n$ with the functions $\tilde{p}$ and $\tilde{g}$ and $L'=2ML$ instead of $L$. By using the same arguments as in \citep[Theorem 2.1]{DaskalakisP2011-CLS}, it follows that the problem lies in \pls/. Note that we do not actually need to construct arithmetic circuits that compute $\tilde{p}$ and $\tilde{g}$ (from the given circuits for $p$ and $g$), because it suffices to be able to compute the functions in polynomial time for the arguments in \citep{DaskalakisP2011-CLS} to go through.
\end{proof}

\section{Approximation by linear circuits}\label{sec:approx-linear-circuit}

In this section, we show that functions computed by arithmetic circuits can be approximated by \emph{linear} arithmetic circuits with a very small error. In linear arithmetic circuits we are only allowed to use the gates $+$, $-$, $\max$, $\min$, $\times \zeta$ and rational constants. In particular, we cannot use general multiplication gates or comparison gates.

\begin{theorem}\label{thm:lin-circuit-simple}
	Given a well-behaved arithmetic circuit $f: [0,1]^n \to \mathbb{R}^d$, a purported Lipschitz constant $L > 0$, and a precision parameter $\varepsilon > 0$, in polynomial time in $\sz(f)$, $\log L$ and $\log (1/\varepsilon)$, we can construct a linear arithmetic circuit $F: [0,1]^n \to \mathbb{R}^d$ such that for any $x \in [0,1]^n$ it holds that:
	\begin{itemize}
		\item $\|f(x)-F(x)\|_\infty \leq \varepsilon$, or
		\item given $x$, we can efficiently compute $y \in [0,1]^n$ such that
		$$\|f(x)-f(y)\|_\infty > L \|x-y\|_\infty.$$
	\end{itemize}
	Here ``efficiently'' means in polynomial time in $\sz(x)$, $\sz(f)$, $\log L$ and $\log(1/\varepsilon)$.
\end{theorem}

\noindent Our proof of this result relies on existing techniques introduced by \citet{DGP09} and \citet{CDT09}, but with a modification that ensures that we only get a very small error. Indeed, using the usual so-called sampling trick with averaging does not work here. We modify the sampling trick to output the \emph{median} instead of the average.

Since we believe that this tool will be useful in future work, we prove a more general version of \cref{thm:lin-circuit-simple}. This more general version is \cref{thm:lin-circuit-general} and it is presented and proved in the next subsection, where we also explain how \cref{thm:lin-circuit-simple} is easily obtained from \cref{thm:lin-circuit-general}.

\begin{remark}\label{rem:approx-linear-domain}
	Note that in \cref{thm:lin-circuit-simple} the domain $[0,1]^n$ can be replaced by $[-M,M]^n$ for any $M > 0$ (in which case the running time is polynomial in the same quantities and in $\log M$). This is easy to show by using a simple bijection between $[0,1]^n$ and $[-M,M]^n$. This also holds for the more general statement in \cref{thm:lin-circuit-general}. In fact, the result holds for any convex set $S \subseteq [-M,M]^n$, as long as we can efficiently compute the projection onto $S$. Furthermore, the choice of the $\ell_\infty$-norm in the statement is not important, and it can be replaced by any other $\ell_p$-norm, if $f$ is $L$-Lipschitz-continuous with respect to that norm.
\end{remark}

\subsection{General statement and proof}

In order to make the statement of the result as general as possible, we consider a class of functions $\mathcal{F}$. Every function $f \in \mathcal{F}$ has an associated representation, and we let $\sz(f)$ denote the length of the representation of $f$. For example, if $\mathcal{F}$ is the class of functions represented using a certain type of circuit, then $\sz(f)$ is the size of the circuit corresponding to $f$. The following definition is inspired by a similar notion in \citep{EY10-Nash-FIXP}.

\begin{definition}
	A class $\mathcal{F}$ of functions is said to be \emph{polynomially-approximately-computable} if there exists a polynomial $q$ such that for any function $f \in \mathcal{F}$ where $f: [0,1]^n \to \mathbb{R}^d$, any point $x \in [0,1]^n$, and any precision parameter $\delta > 0$, a value $v \in \mathbb{R}^d$ such that $\|f(x)-v\|_\infty \leq \delta$ can be computed in time $q(\sz(f) + \sz(x) + \log(1/\delta))$.
\end{definition}

\noindent The next theorem basically says that any polynomially-approximately-computable class can be approximated by linear arithmetic circuits, as long as the functions are Lipschitz-continuous.

\begin{theorem}\label{thm:lin-circuit-general}
	Let $\mathcal{F}$ be a polynomially-approximately-computable class of functions. Given $f \in \mathcal{F}$ where $f: [0,1]^n \to \mathbb{R}^d$, $L > 0$ and $\varepsilon > 0$, in polynomial time in $\sz(f)$, $\log L$ and $\log(1/\varepsilon)$, we can construct a linear arithmetic circuit $F: [0,1]^n \to \mathbb{R}^d$ such that for any $x \in [0,1]^n$ it holds that:
	\begin{itemize}
		\item $\|f(x)-F(x)\|_\infty \leq \varepsilon$, or
		\item given $x$, we can efficiently compute $y \in [0,1]^n$ such that
		$$\|f(x)-f(y)\|_\infty > L \|x-y\|_\infty + \varepsilon/2.$$
	\end{itemize}
	Here ``efficiently'' means in polynomial time in $\sz(x)$, $\sz(f)$, $\log L$ and $\log(1/\varepsilon)$.
\end{theorem}

\noindent Note that \cref{thm:lin-circuit-general} immediately implies \cref{thm:lin-circuit-simple}, since the class of all well-behaved arithmetic circuits mapping $[0,1]^n$ to $\mathbb{R}^d$ is polynomially-approximately-computable. (In fact, it is even exactly computable.) Note that since $L \|x-y\|_\infty + \varepsilon/2 \geq L \|x-y\|_\infty$ for any $\varepsilon>0$, we indeed immediately obtain \cref{thm:lin-circuit-simple}.

\begin{proof}[Proof of \cref{thm:lin-circuit-general}]
	First of all, note that we can assume that $d=1$. Indeed, if $f: [0,1]^n \to \mathbb{R}^d$, then we can consider $f_1, \dots, f_d : [0,1]^n \to \mathbb{R}$ where $f_i(x) = [f(x)]_i$, and construct linear arithmetic circuits $F_1, \dots, F_d$ approximating $f_1, \dots, f_d$ (as in the statement of the theorem). By constructing $F(x) = (F_1(x), \dots, F_d(x))$, we have then obtained a linear arithmetic that satisfies the statement of the theorem. Indeed, if for some $x \in [0,1]^n$ we have $\|f(x)-F(x)\|_\infty > \varepsilon$, then it follows that there exists $i \in [n]$ such that $|f_i(x)-F_i(x)| > \varepsilon$. From this it follows that we can efficiently compute $y$ with $|f_i(x)-f_i(y)|> L \|x-y\|_\infty + \varepsilon/2$, which implies that $\|f(x)-f(y)\|_\infty > L \|x-y\|_\infty + \varepsilon/2$. Note that $d \leq \sz(f)$, so this construction remains polynomial-time with respect to $\sz(f)$, $\log L$ and $\log(1/\varepsilon)$.
	
	Consider any $L > 0$, $\varepsilon > 0$ and $f \in \mathcal{F}$ where $f: [0,1]^n \to \mathbb{R}$. Pick $k \in \mathbb{N}$ such that $N := 2^k \geq 4L/\varepsilon$. We consider the partition of $[0,1]^n$ into $N^n$ subcubes of side-length $1/N$. Every $p \in [N]^n$ then represents one subcube of the partition, and we let $\widehat{p} \in [0,1]^n$ denote the centre of that subcube. Formally, for all $p \in [N]^n$, $\widehat{p} \in [0,1]^n$ is given by
	$$\left[\,\widehat{p}\,\right]_i = \frac{2p_i-1}{2N} \qquad \text{for all $i \in [n]$}.$$
	For any $p \in [N]^n$, let $\tilde{f}(\widehat{p})$ denote the approximation of $f(\widehat{p})$ with error at most $\varepsilon/16$. Note that $\tilde{f}(\widehat{p})$ can be computed in time $q(\sz(f) + \sz(\widehat{p}) + \log(16/\varepsilon))$ (where $q$ is the polynomial associated to $\mathcal{F}$). Since $\sz(\widehat{p})$ is polynomial in $\log L$ and $\log(1/\varepsilon)$, we can compute a rational number $M > 0$ such that $\sz(M)$ is polynomial in $\sz(f)$, $\log L$ and $\log(1/\varepsilon)$, and it holds that $|\tilde{f}(\widehat{p})| \leq M$ for all $p \in [N]^n$. We then define
	$$C(p) := \left\lfloor\left(\tilde{f}(\widehat{p})+M\right)\frac{16}{\varepsilon} \right\rfloor + 1$$
	for all $p \in [N]^n$. Note that $C(p) \in [1,32M/\varepsilon+1] \cap \mathbb{N}$. Pick $m \in \mathbb{N}$ such that $2^m \geq 32M/\varepsilon+1$. Then, $C: [N]^n \to [2^m]$ and we construct a Boolean circuit $\{0,1\}^{kn} \to \{0,1\}^m$ that computes $C$. Importantly, the Boolean circuit can be constructed in polynomial time in $\sz(f)$, $\log L$ and $\log(1/\varepsilon)$. Before we move on, note that for all $p \in [N]^n$, letting $V(p) := (C(p)-1)\frac{\varepsilon}{16}-M$, it holds that
	\begin{equation}\label{eq:lin-circuit-trafo}
	\left|f(\widehat{p}) - V(p)\right| \leq \left|f(\widehat{p}) - \tilde{f}(\widehat{p})\right| + \left|\tilde{f}(\widehat{p}) - V(p)\right| \leq \frac{\varepsilon}{8}
	\end{equation}
	since $|f(\widehat{p}) - \tilde{f}(\widehat{p})| \leq \varepsilon/16$ and
	\begin{align*}
	&\quad \left|\tilde{f}(\widehat{p}) - V(p)\right|\\
	&= \left|\tilde{f}(\widehat{p}) + M - \left\lfloor\left(\tilde{f}(\widehat{p})+M\right)\frac{16}{\varepsilon} \right\rfloor\frac{\varepsilon}{16}\right|\\
	&\leq \left|\tilde{f}(\widehat{p}) + M - \left(\tilde{f}(\widehat{p})+M\right)\frac{16}{\varepsilon} \frac{\varepsilon}{16}\right| + \left|\left(\tilde{f}(\widehat{p})+M\right)\frac{16}{\varepsilon} - \left\lfloor\left(\tilde{f}(\widehat{p})+M\right)\frac{16}{\varepsilon} \right\rfloor \right| \frac{\varepsilon}{16}\\
	&\leq \frac{\varepsilon}{16}.
	\end{align*}
	
	Using \cref{lem:lin-circuit-approx}, which is our key lemma here and is stated and proved in the next subsection, we can construct a linear arithmetic circuit $F: [0,1]^n \to \mathbb{R}$ in polynomial time in $\sz(C)$ (and thus in $\sz(f)$, $\log L$ and $\log(1/\varepsilon)$), such that for all $x \in [0,1]^n$
	$$F(x) \in \left[ \min_{p \in S(x)} C(p), \max_{p \in S(x)} C(p) \right]$$
	where $S(x) \subseteq [N]^n$ is such that
	\begin{enumerate}
		\item $|S(x)| \leq n+1$,
		\item $\|x - \widehat{p}\|_\infty \leq 1/N$ for all $p \in S(x)$, and
		\item $S(x)$ can be computed in polynomial time in $\sz(x)$ and $\log N$.
	\end{enumerate}
	We modify the linear circuit so that instead of outputting $F(x)$, it outputs $(F(x)-1)\frac{\varepsilon}{16} - M$. Note that this is straightforward to do using the arithmetic gates at our disposal. Since $V(p) = (C(p)-1)\frac{\varepsilon}{16}-M$, we obtain that for all $x \in [0,1]^n$
	$$F(x) \in \left[ \min_{p \in S(x)} V(p), \max_{p \in S(x)} V(p) \right].$$
	
	We are now ready to complete the proof. For this it suffices to show that if $|f(x)-F(x)| > \varepsilon$, then the second point in the statement of the theorem must hold. Assume that $x \in [0,1]^n$ is such that $|f(x)-F(x)| > \varepsilon$. It immediately follows that there exists $p^* \in S(x)$ such that $|f(x)-V(p^*)| > \varepsilon$. By \cref{eq:lin-circuit-trafo}, it follows that $|f(x)-f(\widehat{p^*})| > \varepsilon - \varepsilon/8 = 7\varepsilon/8$. Note that we might not be able to identify $p^*$, since we can only approximately compute $f$. Thus, we instead proceed as follows. We compute $p' = \argmax_{p \in S(x)} |f'(x)-f'(\widehat{p})|$, where $f'$ denotes computation of $f$ with error at most $\varepsilon/32$. Note that $p'$ can be computed in polynomial time in $\sz(x)$, $\sz(f)$, $\log L$ and $\log(1/\varepsilon)$, since $f'$ and $S(x)$ can be computed efficiently.
	
	We now show that $y=\widehat{p'} \in [0,1]^n$ satisfies the second point in the statement of the theorem. First of all, note that $|f'(x)-f'(\widehat{p^*})| > 7\varepsilon/8 - 2\varepsilon/32 = 13\varepsilon/16$. By the choice of $p'$, it must be that $|f'(x)-f'(\widehat{p'})| \geq |f'(x)-f'(\widehat{p^*})| > 13\varepsilon/16$, which implies that $|f(x)-f(\widehat{p'})| > 13\varepsilon/16 - 2\varepsilon/32 > 3\varepsilon/4$. On the other hand, since we have that $\|x - \widehat{p'}\|_\infty \leq 1/N \leq \varepsilon/4L$ (because $p' \in S(x)$), it follows that $L \|x - \widehat{p'}\|_\infty \leq \varepsilon/4$. Thus, we indeed have that $|f(x)-f(y)|> L \|x-y\|_\infty + \varepsilon/2$, as desired. Since $p'$ can be computed efficiently, so can $y=\widehat{p'}$.
\end{proof}

\subsection{Key Lemma}

Let us recall some notation introduced in the proof of \cref{thm:lin-circuit-general}.
For $N \in \mathbb{N}$, consider the partition of $[0,1]^n$ into $N^n$ subcubes of side-length $1/N$. Every $p \in [N]^n$ then represents one subcube of the partition, and we let $\widehat{p} \in [0,1]^n$ denote the centre of that subcube. Formally, for all $p \in [N]^n$, $\widehat{p} \in [0,1]^n$ is given by
$$\left[\,\widehat{p}\,\right]_i = \frac{2p_i-1}{2N}$$
for all $i \in [n]$.

\begin{lemma}\label{lem:lin-circuit-approx}
	Assume that we are given a Boolean circuit $C : \{0,1\}^{kn} \to \{0,1\}^m$, interpreted as a function $C: [N]^n \to [2^m]$, where $N=2^k$. Then, in polynomial time in $\sz(C)$, we can construct a linear arithmetic circuit $F: [0,1]^n \to \mathbb{R}$ such that for all $x \in [0,1]^n$
	$$F(x) \in \left[ \min_{p \in S(x)} C(p), \max_{p \in S(x)} C(p) \right]$$
	where $S(x) \subseteq [N]^n$ is such that
	\begin{enumerate}
		\item $|S(x)| \leq n+1$,
		\item $\|x - \widehat{p}\|_\infty \leq 1/N$ for all $p \in S(x)$, and
		\item $S(x)$ can be computed in polynomial time in $\sz(x)$ and $\log N$.
	\end{enumerate}
\end{lemma}

\begin{proof}
	We begin by a formal definition of $S(x)$ and prove that it has the three properties mentioned in the statement of the Lemma. We then proceed with the construction of the linear arithmetic circuit.
	
	For $N \in \mathbb{N}$, consider the partition of $[0,1]$ into $N$ subintervals of length $1/N$. Let $I_N : [0,1] \to [N]$ denote the function that maps any point in $[0,1]$ to the index of the subinterval that contains it. In the case where a point lies on the boundary between two subintervals, i.e., $x \in B = \{1/N,2/N, \dots, (N-1)/N\}$, the tie is broken in favour of the smaller index. Formally,
	$$I_N(x) = \min \left\{\ell \in [N] \, \left| \, x \in \left[\frac{\ell-1}{N}, \frac{\ell}{N}\right]\right\}\right..$$
	We abuse notation and let $I_N: [0,1]^n \to [N]^n$ denote the natural extension of the function to $[0,1]^n$, where it is simply applied on each coordinate separately. Thus, if we consider the partition of $[0,1]^n$ into $N^n$ subcubes of side-length $1/N$, then, for any point $x \in [0,1]^n$, $p = I_N(x) \in [N]^n$ is the index of the subcube containing $x$. For $x \in \mathbb{R}^n \setminus [0,1]^n$, we let $I_N(x) := I_N(y)$ where $y$ is obtained by projecting every coordinate of $x$ onto $[0,1]$.
	
	Letting $\mathbf{e} \in \mathbb{R}^n$ denote the all-ones vector, we define
	$$S(x) = \left\{I_N(x+\alpha \cdot \mathbf{e}) \, \left| \, \alpha \in \left[0, \frac{1}{2N}\right] \right\}\right..$$
	In other words, we consider a small segment starting at $x$ and moving up simultaneously in all dimensions, and we let $S(x)$ be the set of subcubes-indices of all the points on the segment. We can now prove the three properties of $S(x)$:
	\begin{enumerate}
		\item Note that for any $i \in [n]$, there exists at most one value $\alpha \in [0,1/2N]$ such that $[x+\alpha \cdot \mathbf{e}]_i \in B = \{1/N,2/N, \dots, (N-1)/N\}$. We let $\alpha_i$ denote that value of $\alpha$ if it exists, and otherwise we let $\alpha_i=1/2N$. Thus, we obtain $\alpha_1, \alpha_2, \dots, \alpha_n \in [0,1/2N]$ and we rename them $\beta_i$ so that they are ordered, i.e., $\beta_1 \leq \beta_2 \leq \dots \leq \beta_n$ and $\{\beta_i \, | \, i \in [n]\} = \{\alpha_i \, | \, i \in [n]\}$. By the definition of $I_N$, it is then easy to see that $\alpha \mapsto I_N(x + \alpha \cdot \mathbf{e})$ is constant on each of the intervals $[0,\beta_1]$, $(\beta_1,\beta_2]$, $(\beta_2,\beta_3]$, $\dots$, $(\beta_{n-1},\beta_n]$ and $(\beta_n,1/2N]$. Since these $n+1$ intervals (some of which are possibly empty) cover the entirety of $[0,1/2N]$, it follows that $|S(x)| \leq n+1$.
		\item Consider any $p \in S(x)$. Let $\alpha \in [0,1/2N]$ be such that $p = I_N(y)$ where $y = x+\alpha \cdot \mathbf{e}$. For any $i \in [n]$, it holds that $x_i \leq y_i \leq x_i + 1/2N$. There are two cases to consider. If $I_N(y_i) = I_N(x_i)$, then this means that $x_i$ lies in the subinterval of length $1/N$ centred at $[\widehat{p}]_i$, and thus, in particular, $|x_i-[\widehat{p}]_i| \leq 1/2N \leq 1/N$. The only other possibility is that $I_N(y_i) = I_N(x_i) + 1$, since $\alpha \in [0,1/2N]$. But for this to happen, it must be that $(2 I_N(x_i)-1)/2N \leq x_i \leq 2 I_N(x_i)/2N$, since $\alpha \leq 1/2N$. By definition, $[\widehat{p}]_i = (2 I_N(y_i)-1)/2N = (2 I_N(x_i)+1)/2N$, and so we again obtain that $|x_i-[\widehat{p}]_i| \leq 1/N$. Since this holds for all $i \in [n]$, it follows that $\|x - \widehat{p}\|_\infty \leq 1/N$.
		\item Given $x \in [0,1]^n$, the values $\alpha_1, \dots, \alpha_n \in [0,1/2N]$, defined in the first point above, can be computed in polynomial time in $\sz(x)$, $n$ and $\log N$. Then, $S(x)$ can be computed by simply evaluating $I_N(x+\alpha \cdot \mathbf{e})$ for all $\alpha \in \{\alpha_1, \dots, \alpha_n, 1/2N\}$, which can also be done in polynomial time in $\sz(x)$, $n$ and $\log N$. Note that since $n \leq \sz(x)$, the computation of $S(x)$ runs in polynomial time in $\sz(x)$ and $\log N$.
	\end{enumerate}
We can now describe how the linear arithmetic circuit $F: [0,1]^n \to \mathbb{R}$ is constructed. Let $C : \{0,1\}^{kn} \to \{0,1\}^m$ be the Boolean circuit that is provided. It is interpreted as a function $C: [N]^n \to [2^m]$, where $N=2^k$. Let $x \in [0,1]^n$ be the input to the linear arithmetic circuit. $F$ is constructed to perform the following steps.
	
	\paragraph{\bf Step 1: Sampling trick} In the first step, we create a sample $T$ of points close to $x$. This is a standard trick that was introduced in the study of the complexity of computing Nash equilibria \citep{DGP09,CDT09}. Here we use the so-called equi-angle sampling trick introduced by \citet{CDT09}. The sample $T$ consists of $2n+1$ points:
	$$T = \left.\left\{x + \frac{\ell}{4nN} \cdot \mathbf{e} \, \right| \, \ell \in \{0, 1, 2, \dots,2n\} \right\}.$$
	Note that these $2n+1$ points can easily be computed by $F$ given the input $x$. The following two observations are important:
	\begin{enumerate}
		\item for all $y \in T$, $I_N(y) \in S(x)$ (by definition of $S(x)$),
		\item Let $T_b = \{y \in T \, | \, \exists i \in [n]: \textup{dist}(y_i,B) < \frac{1}{8nN} \}$, where $B = \{1/N,2/N, \dots, (N-1)/N\}$ and $\textup{dist}(y_i,B) = \min_{t \in B} |y_i-t|$. We call these the \emph{bad} samples, because they are too close to a boundary between two subcubes. The points in $T_g = T \setminus T_b$ are the \emph{good} samples. It holds that $|T_b| \leq n$. This is easy to see by fixing some coordinate $i \in [n]$, and noting that there exists at most one point $y \in T$ such that $\textup{dist}(y_i,B) < \frac{1}{8nN}$. Indeed, since the samples are successively $1/4nN$ apart, at most one can be sufficiently close to any given boundary. Furthermore, since the samples are all $1/2N$ close, at most one boundary can be sufficiently close to any of them (for every coordinate). Thus, since there is at most one bad sample for each coordinate, there are at most $n$ bad samples overall.
	\end{enumerate}
	
	\paragraph{\bf Step 2: Bit extraction} In the second step, we want to compute $I_N(y)$ for all $y \in T$. This corresponds to extracting the first $k$ bits of each coordinate of each point $y \in T$, because $N=2^k$. Unfortunately, bit extraction is not a continuous function and thus it is impossible to always perform it correctly with a linear arithmetic circuit. Fortunately, we will show that we can perform it correctly for \emph{most} points in $T$, namely all the good points in $T_g$.
	
	Consider any $y \in T$ and any coordinate $i \in [n]$. In order to extract the first bit of $y_i$, the arithmetic circuit computes
	$$b_1 = \min\{1, \max\{0, 8nN(y_i-1/2)\} =: \phi(y_i-1/2).$$
	Note that if $y_i \geq 1/2 + 1/8nN$, then $8nN(y_i-1/2) \geq 1$ and thus $b_1 = 1$. On the other hand, if $y_i \leq 1/2 - 1/8nN$, then $8nN(y_i-1/2) \leq -1$ and thus $b_1 = 0$. This means that if $\textup{dist}(y_i,B) \geq 1/8nN$, the first bit of $y_i$ is extracted correctly. Note that $B = \{1/N,2/N, \dots, (N-1)/N\} = \{1/2^k,2/2^k, \dots, (2^k-1)/2^k\}$.
	
	To extract the second bit, the arithmetic circuit computes $t := y_i - b_1/2$ and
	$$b_2 = \phi(t-1/4).$$
	By the same argument as above, $b_2$ is the correct second bit of $y_i$, if $|t-1/4| \geq 1/8nN$, i.e., if $|y_i-1/4| \geq 8nN$ and $|y_i-3/4| \geq 8nN$. Thus, if $\textup{dist}(y_i,B) \geq 1/8nN$, the second bit is also computed correctly, since $1/4,3/4 \in B$.
	
	To extract the third bit, the arithmetic circuit updates $t:= t - b_2/4$ and computes $b_3 = \phi(t-1/8)$. We proceed analogously up to the $k$th bit $b_k$. By induction and the same arguments as above, it follows that the first $k$ bits of $y_i$ are computed correctly by the arithmetic circuit as long as $\textup{dist}(y_i,B) \geq 1/8nN$. In particular, this condition always holds for $y \in T_g$.
	
	By performing this bit extraction for each coordinate of each $y \in T$, we obtain the purported bit representation of $I_N(y)$ for each $y \in T$. The argumentation in the previous paragraphs shows that for all $y \in T_g$, we indeed obtain the correct bit representation of $I_N(y)$. For $y \in T_b$, we have no control over what happens, and it is entirely possible that the procedure outputs numbers that are not valid bits, i.e., not in $\{0,1\}$.
	
	\paragraph{\bf Step 3: Simulation of the Boolean circuit} In the next step, for each $y \in T$, we evaluate the circuit $C$ on the bits purportedly representing $I_N(y)$. The Boolean gates of $C$ are simulated by the arithmetic circuit as follows:
	\begin{itemize}
		\item $\lnot b := 1- b$,
		\item $b \lor b' := \min\{1, b+b'\}$,
		\item $b \land b' := \max\{0, b+b'-1\}$.
	\end{itemize}
	Note that if the input bits $b,b'$ are valid bits, i.e., in $\{0,1\}$, then the Boolean gates are simulated correctly, and the output is also a valid bit.
	
	For $y \in T_g$, since the input bits indeed represent $I_N(y)$, the simulation of $C$ will thus output the correct bit representation of $C(I_N(y)) \in [2^m]$. We can obtain the value $C(I_N(y)) \in [2^m]$ itself by decoding the bit representation, i.e., multiplying every bit by the corresponding power of 2 and adding all the terms together.
	
	Let $V(y) \in \mathbb{R}$ denote the value that this step outputs for each $y \in T$. For $y \in T_g$, we have that $V(y) = C(I_N(y))$. For $y \in T_b$, there is no guarantee other than $V(y) \in \mathbb{R}$.
	
	\paragraph{\bf Step 4: Median using a sorting network} In the last step, we want to use the $|T|=2n+1$ values that we have computed (namely $\{V(y) \, | \, y \in T\}$) to compute the final output of our arithmetic circuit. In previous constructions of this type, in particular \citep{DGP09,CDT09}, the circuit would simply output the average of the $V(y)$. However, this is not good enough to prove our statement, because even a single bad point $y$ can introduce an inversely polynomial error in the average.
	
	In order to obtain a stronger guarantee, the arithmetic circuit instead outputs the \emph{median} of the multiset $\{V(y) \, | \, y \in T\}$. The median of the given $2n+1$ values can be computed by constructing a so-called sorting network (see, e.g., \citep{knuth1998sorting}). The basic operation of a sorting network can easily be simulated by the $\max$ and $\min$ gates. It is easy to construct a sorting network for $2n+1$ values that has size polynomial in $n$. The output of the sorting network will be the same values that it had as input, but sorted. In other words, the sorting network outputs $V_1 \leq V_2 \leq \dots \leq V_{2n+1}$ such that $\{V_j \, | \, j \in [2n+1]\} = \{V(y) \, | \, y \in T\}$ as multisets. The final output of our arithmetic circuit is $V_{n+1}$, which is exactly the median of the $2n+1$ values.
	
	Recall from step 1 that $|T_b| \leq n$ and thus $|T_g| \geq n+1$. It immediately follows that either $V_{n+1}$ corresponds to $V(y)$ of a good sample $y$, or there exist $i < n+1$ and $j > n+1$ such that both $V_i$ and $V_j$ correspond to good samples. In other words, the output of the circuit satisfies $F(x) \in [\min_{y \in T_g} C(I_N(y)), \max_{y \in T_g} C(I_N(y))]$. As noted in step 1, $I_N(y) \in S(x)$ for all $y \in T$. Thus, we obtain that $F(x) \in [\min_{p \in S(x)} C(p), \max_{p \in S(x)} C(p)]$.
	
	\smallskip
	
	It follows that the linear arithmetic circuit $F$ that we have constructed indeed satisfies the statement of the Lemma. Furthermore, the construction we have described can be performed in polynomial time in $\sz(C)$, $n$, $m$ and $k$. Since $\sz(C) \geq \max\{n,k,m\}$, it is simply polynomial time in $\sz(C)$.
\end{proof}

\begin{acks}
We thank the anonymous reviewers for comments and suggestions that helped improve the presentation of the paper.
Alexandros Hollender was supported by an EPSRC doctoral studentship (Reference 1892947).
\end{acks}

\bibliographystyle{ACM-Reference-Format}
\bibliography{references}

\end{document}